\documentclass[10pt]{amsart}
\usepackage{graphicx,amssymb,amsfonts,amsmath,amsthm,newlfont}
\usepackage{epsfig,url}
\usepackage{numprint}
\usepackage{axodraw4j}

\usepackage{pstricks}
\usepackage{color}

\usepackage[all,2cell]{xy} \UseAllTwocells \SilentMatrices

\newtheorem{thm}{Theorem}[section]

\newtheorem{lem}[thm]{Lemma}

\theoremstyle{definition}

\newtheorem{conj}[thm]{Conjecture} 

\newtheorem{ex}[thm]{Example}

\theoremstyle{remark}

\newcommand{\FF}{\mathbb F}
\newcommand{\ZZ}{\mathbb Z}

\newcommand{\CC}{\mathbb C}
\newcommand{\RR}{\mathbb R}

\newcommand{\QQ}{\mathbb Q}
\newcommand{\Li}{\mathrm{Li}}

\newcommand{\F}{\mathbb F}

\newcommand{\dd}{\mathrm{d}}

\newcommand{\sE}{\mathcal{E}}
\newcommand{\sG}{\mathcal{G}}
\newcommand{\sP}{\mathcal{P}}
\newcommand{\sV}{\mathcal{V}}
\newcommand{\zz}{\overline{z}}
\newcommand{\ii}{\mathrm{i}}
\newcommand{\lfrac}[2]{\frac{\numprint{#1}}{\numprint{#2}}}
\newcommand{\dR}{\mathfrak{dr}}

\title{Numbers and Functions in Quantum Field Theory}
\author{Oliver Schnetz}
\begin{document}
\begin{abstract}
We review recent results in the theory of numbers and single-valued functions on the complex plane which arise in quantum field theory.
We use the results to calculate the renormalization functions $\beta$, $\gamma$, $\gamma_m$ of dimensionally regularized $\phi^4$ theory in the minimal subtraction
scheme up to seven loops.
\end{abstract}
\maketitle
\section{Introduction}
Quantum field theories (QFTs) are fundamental theories of physical interactions. Physical QFTs are the Electroweak theory which combines electromagnetism with the weak interaction,
Quantumchromodynamics which describes the interaction between quarks and gluons, and $\phi^4$ theory for the Higgs boson. Gravity has not yet found a quantum formulation.

Although QFTs are experimentally very well confirmed (see e.g.\ the anomalous magnetic moment of the electron for an impressive example, Sect.\ \ref{g2} \cite{A2, H1, L1, Laporta}),
a complete mathematical understanding of QFTs is lacking. On the one hand there are fundamental questions like the existence and structure of QFTs. On the other hand
there is demand for practical tools to perform QFT calculations. Due to the mathematical difficulty of QFTs progress is modest. Here, we report on some recent
results in both directions.

For a while it seemed possible that the number content of QFT is given by multiple zeta values (MZVs) which are multiple sums that generalize the Riemann zeta function at positive
integer arguments. Assuming standard conjectures it has now been proved that this is not the case \cite{K3,BD}.

With the theory of graphical functions, a tool was developed to perform multiloop calculations in massless scalar field theories \cite{gf, par, SYM}. A first notable
breakthrough was the proof of the zig-zag conjecture \cite{BK,ZZ} which gives an explicit formula for periods of zig-zag graphs (see Thm.\ \ref{zigzag}).

For more general applications it was necessary to introduce a novel family of single-valued functions on the complex plane: generalized single-valued hyperlogarithms (GSVHs,
see Sect.\ \ref{GSVH}). Generalized single-valued hyperlogarithms vastly generalize single-valued multiple polylogarithms. Nevertheless it was possible to translate the vital properties
of single-valued multiple polylogarithms into the framework of GSVHs \cite{GSVH}.
A Maple\textsuperscript\texttrademark\ package was developed that can calculate many periods in $\phi^4$ theory up to 11 loops \cite{Hyperlogproc}.
With a large amount of data available, the structure of $\phi^4$ periods could be connected to the Galois theory of algebraic integrals \cite{Bcoact1, Bcoact2, coaction}.

To make further contact to physics it is necessary to regularize integrals which diverge in four dimensions. This is often done by generalizing to $4-\epsilon$ `dimensions' (which can
be defined in a parametric representation of QFT integrals \cite{IZ}). Using GSVHs it was possible to obtain $\epsilon$-expansions for QFT periods and graphical functions.
The procedure {\tt Phi4} in {\tt HyperlogProcedures} calculates the $\beta$-function and the anomalous dimensions $\gamma$ and $\gamma_m$ up to seven loops
in the minimally subtracted $O(n)$ symmetric $\phi^4$ theory \cite{Hyperlogproc}. The self energy can be calculated to six loops (see Sect.\ \ref{phi4} for results in the case $n=1$).

\section*{Acknowledgements}
I am very grateful to D. Kreimer and F. Knop for their support. I am also deeply indebted to my co-authors F. Brown, E. Panzer, and K. Yeats.
The calculation of the seven loop renormalization functions used some input from E. Panzer's Maple package {\tt HyperInt}. He also supported the author with very valuable discussions.
Many results in this report were found when the author was visiting scientist at the Humboldt University, Berlin.
The author is supported by DFG grant SCHN~1240/2. The computer calculations were performed on the server {\tt mem} of the Dept.\ Mathematik, Friedrich-Alexander Universit\"at of Erlangen-N\"urnberg.

\section{General idea}
In QFT, graphs are used to symbolize integrals. One first has to fix a space-time dimension $d$ which, for brevity, we connect to the parameter $\lambda$ according to
$$
d=2+2\lambda>2.
$$
The graphs allowed depend on the QFT chosen. Here we are mainly interested in $\phi^4$ theory which limits the vertex degree to four.
Feynman rules translate graphs to integrals. They exist in momentum and in position space with integration over $d$-dimensional variables associated to
independent cycles or vertices, respectively. For massive theories one needs to use momentum space to obtain explicitly algebraic integrands.
Here we are mostly interested in massless calculations which allows us to use position space.

The general setup is as follows. Assume $G$ is a graph with edges $\sE(G)$ and vertices $\sV(G)$.
We do not assume here that every vertex in $G$ has maximum degree four. Every edge $e\in\sE(G)$ has a weight $\nu_e\in\RR$.
We assume that $G$ has no self loops (tadpoles). This is common in dimensionally regularized massless theories. The edge weight is additive, i.e.\
a multiple edge in $G$ is equivalent to a single edge with weight equal to the weight sum of the multiple edge. Therefore, we only need to consider single edges.
We split the set of vertices into `internal' vertices $\sV^{\mathrm{int}}(G)$ and `external' vertices $\sV^{\mathrm{ext}}(G)$.
To every vertex we associate a $d$ dimensional variable and do not distinguish between the vertex and the variable.
We use $x_i$, $i=1,\ldots,V^{\mathrm{int}}(G)=|\sV^{\mathrm{int}}(G)|$ for internal vertices and $z_i$, $i=1,\ldots,V^{\mathrm{ext}}(G)=|\sV^{\mathrm{ext}}(G)|$ for
external vertices. To every edge $e=\{u,v\}\in\sE(G)$ between the two vertices $u,v\in\sV(G)$ (internal or external) we associate a quadric $Q_e$ which is given by the
Euclidean distance between (the variables associated to) $u$ and $v$,
$$
Q_e(u,v)=||u-v||^2=(u_1-v_1)^2+\ldots+(u_d-v_d)^2.
$$
Assume the graph $G$ has the property that the following integral exists
\begin{equation}\label{fdef}
f_G^{(\lambda)}(z_1,\ldots,z_{V^{\mathrm{ext}}})=\left(\prod_{v=1}^{V^{\mathrm{int}}(G)} \int_{\RR^d}\frac{\dd^dx_v}{\pi^{d/2}}\right)\frac{1}{\prod_{e\in\sE(G)}Q_e^{\lambda\nu_e}}.
\end{equation}
Due to translational and scale invariance the integral can only exist if $G$ has at least two external vertices.
In the case of exactly two external vertices the integral is determined up to a constant by these symmetries,
\begin{equation}\label{2external}
f_G^{(\lambda)}(z_1,z_2)=P(G)||z_1-z_2||^{dV^{\mathrm{int}(G)}-2\lambda\sum_{e\in\sE(G)}\nu_e}.
\end{equation}
The number $P(G)\in\RR_+$ is the Feynman period of $G$.

In primitive logarithmically divergent physical graphs (with external legs) the residue (in the regulator) is given by the period of the graph with amputated external legs.
Therefore the calculation of periods is of great importance for renormalizing QFTs (see e.g.\ \cite{BK}).

Without loss of information we set $z_1=$\,`0'\,$=(0,\ldots,0)$ and $z_2=$\,`1'\,$=(1,0,\ldots,0)$ (we may associate to the vertex 1 any unit vector in $\RR^d$) and obtain
$$
f_G^{(\lambda)}(0,1)=P(G).
$$
In the case of three external vertices we can again exploit the symmetry of the integral to reduce the number of variables. In this case we may use a complex variable $z$ (and its
complex conjugate $\zz$) to describe the functional behavior of $f_G^{(\lambda)}$. We obtain the `graphical function' (which we also give the symbol $f_G^{(\lambda)}$) \cite{gf}
$$
f_G^{(\lambda)}(z)=f_G^{(\lambda)}\left(0,1,\left(\frac{z+\zz}{2},\frac{z-\zz}{2\ii},0,\ldots,0\right)\right).
$$
In full generality $f_G^{(\lambda)}(z)$ is a positive single-valued real analytic function on $\CC\backslash\{0,1\}$ \cite{par} with the residual symmetry
$$
f_G^{(\lambda)}(z)=f_G^{(\lambda)}(\zz).
$$
The benefit of complex numbers is that quadrics between external vertices factorize
$$
Q_{\{0,1\}}=1,\quad Q_{\{0,z\}}=z\zz,\quad Q_{\{1,z\}}=(z-1)(\zz-1).
$$
General graphs with four or more external vertices lead to functions which effectively depend on a variable in $\RR^3$. Such functions do not have this factorizing property.
An exception are `conformal' graphs with four external vertices where every internal vertex has degree $2d/(d-2)$. In this case one may use an inversion
$x_i\mapsto x_i/||x_i||^2$ to reduce the integral to the case of three external vertices e.g.\ \cite{SYM}, \cite{S2}. Here, we go the opposite direction
and `complete' graphical functions to conformal graphs with four external vertices \cite{gf}. We will see in Section \ref{gfcompletion} that this is useful to exploit the full symmetry
of graphical functions.

In the following we restrict ourselves to the above two cases, periods and graphical functions. In $\phi^4$ theory four point functions are formally conformal.
However, only the tree level contribution is convergent. A practical tool to resolve divergences is to transform all integrals in a parametric form using the
Schwinger trick \cite{IZ}. In parametric form one integrates over one-dimensional variables associated to the edges of the graph.
The dimension $d$ enters the integrand as an exponent. It is hence possible to consider $d$ as a parameter and use analytic continuation to $d=4-\epsilon$ (losing conformal invariance).
All parametric integrals have Laurant expansions at $\epsilon=0$. Graphical functions are amendable to such a procedure. A general parametric representation of graphical
functions is given in \cite{par} (which generalizes a formula in \cite{Na}).

\section{Numbers}\label{numbers}
A graph $G$ is called $\phi^4$ if $G$ has maximum vertex degree four. The period of a $\phi^4$ graph in $d=4$ dimensions is a $\phi^4$ period.
The loop order of $G$ is the number of independent cycles in $G$.

\subsection{Completion}
We `complete' a graph $G$ with two external vertices 0 and 1 by adding a new external vertex which we give the label `$\infty$' (as reference to conformal symmetry) \cite{Census}.
We add edges from $\infty$ to all internal vertices in $G$ such that every internal vertex has degree $2d/(d-2)$. Finally, we add a weighted triangle with vertices $0,1,\infty$
such that the completed graph becomes $2d/(d-2)$ regular. We denote the completion of $G$ by $\overline{G}$. It is easy to see that completion is always possible and unique.
If we employ the Feynman rule that every edge $e$ adjacent to $\infty$ has quadric $Q_e=1$ we find that the period (\ref{fdef}) does not change under completion.

The power of completion is that the period of a completed graph does not depend on the choice of the external vertices $0,1,\infty$.
This was proved in four dimensions in Theorem and Definition 2.7 in \cite{Census}. The $d$ dimensional case is strictly analogous.
Hence, in $\overline{G}$ we do not need the distinction between internal and external vertices. We henceforth consider $\overline{G}$ as an unlabeled graph.

Obviously, different graphs $G_1$ and $G_2$ can have the same completion. In this case completion implies equality of their periods: $P(G_1)=P(G_2)$ if $\overline{G_1}=\overline{G_2}$.
This identity on periods was already used in \cite{BK}. Completion is effectively a tool to organize equivalence classes of graphs with identical period.
We define $P(\overline{G}):=P(G)$ if $\overline{G}$ is the completion of $G$. The loop order of a completed graph $\overline{G}$ is defined as the number of independent cycles
in the uncompleted graph $G$. So, by definition, completion does not change the loop order.

\subsection{Existence}
In four dimensions the period of a completed graph $\overline{G}$ with edge-weights 1 exists if and only if $\overline{G}$ is internally 6-connected. This means that the only
way to cut $\overline{G}$ with less than 6 edge cuts is to separate off a vertex (Prop.\ 2.6 in \cite{Census}).

In general, a Feynman period $P(G)$ can be considered as a graphical function (see Sect.\ \ref{gfcompletion}) with an isolated external vertex $z$. Existence of $P(G)$ in general is hence a
special case of the criterion for the existence of graphical functions in Thm.\ \ref{gfexistence}.

A completed graph with existing period in four dimensions is called completed primitive \cite{Census}.

\subsection{Product identity}\label{product}
\begin{figure}[t]
\epsfig{file=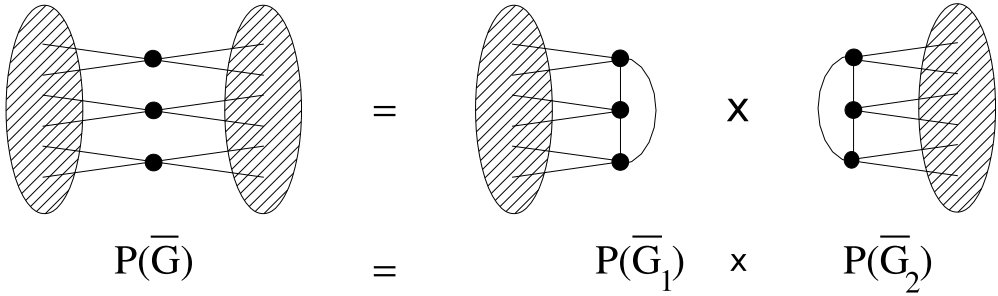,width=\textwidth}
\caption{Vertex connectivity 3 leads to a product of periods.}
\label{fig:prod}%
\end{figure}

The period of a completed graph $\overline{G}$ can only exist if $\overline{G}$ has vertex connectivity $\geq3$ (the vertex connectivity is the minimum number of vertices which,
when removed, split the graph). The period of a completed graph $\overline{G}$ with vertex-connectivity three factorizes in the way depicted in Figure \ref{fig:prod}.
Reversely, completed primitive graphs $\overline{G_1}$,
$\overline{G_2}$ with triangles can be glued along triangles to provide a completed graph with period $P(\overline{G_1})P(\overline{G_2})$. The case $d=4$ with unit edge-weights
was treated in Thm.\ 2.10 in \cite{Census}. The general case (where the weights of the triangles follow from $2d/(d-2)$ regularity) is analogous.

A completed graph with vertex connectivity three is called reducible, otherwise it is irreducible \cite{Census}.
A list of all irreducible completed primitive graphs up to eight loops (in four dimensions with unit edge weights) is given in Table 3 at the end of this report.
{\tt HyperlogProcedures} extends this list to eleven loops \cite{Hyperlogproc}.

Because not all completed primitive graphs have triangles (see e.g.\ $P_{6,4}$ in Table 3) it is unclear if, in general, the product of Feynman periods is a Feynman period.
In particular, one may ask if the $\ZZ$-span of $\phi^4$ periods is a ring (or---weaker---if $\phi^4$ periods span a $\QQ$ algebra).

\subsection{Twist and Fourier identity}
There exist two more known identities on $\phi^4$ periods. The (for graphs with many vertices) frequent twist identity and the rare Fourier identity \cite{Census}.
The Fourier identity was already used in \cite{BK}. The first example of a twist identity which is not also explained by a Fourier identity appears at eight loops.
Twist and Fourier identities are listed in Table 3.

Recent results on the Hepp invariant seem to indicate that more identities between periods exist (see Sect.\ \ref{Hepp}).

\begin{samepage}
\begin{center}
\begin{tabular}{rr|ll}
$\ell$&$\!$wt&number&value\\\hline
1&0&$Q_0=1$&1\\\hline
3&3&$Q_3=\zeta(3)$&1.202~056~903~159\\\hline
4&5&$Q_5=\zeta(5)$&1.036~927~755~143\\\hline
5&7&$Q_7=\zeta(7)$&1.008~349~277~381\\\hline
6&8&$Q_8=N_{3,5}$&0.070~183~206~556\\
&9&$Q_9=\zeta(9)$&1.002~008~392~826\\\hline
7&10&$Q_{10}=N_{3,7}$&0.090~897~338~299\\
&11&$Q_{11,1}=\zeta(11)$&1.000~494~188~604\\
&&$Q_{11,2}=-\zeta(3,5,3)\!+\!\zeta(3)\zeta(5,3)$&0.042~696~696~025\\
&&$Q_{11,3}=P_{7,11}$, Eq.\ (\ref{711})&200.357~566~429\\\hline
8&12&$Q_{12,1}=N_{3,9}$&0.096~506~102~637\\
&&$Q_{12,2}=N_{5,7}$&0.020~460~547~937\\
&&$Q_{12,3}=\pi^{12}/10!$&0.254~703~808~841\\
&13&$Q_{13,1}=\zeta(13)$&1.000~122~713~347\\
&&$Q_{13,2}=-\zeta(5,3,5)\!+\!11\zeta(5)\zeta(5,3)\!+\!5\zeta(5)\zeta(8)$&5.635~097~688~692\\
&&$Q_{13,3}=-\zeta(3,7,3)\!+\!\zeta(3)\zeta(7,3)\!+\!12\zeta(5)\zeta(5,3)\!+\!6\zeta(5)\zeta(8)\hspace*{-5pt}$&6.725~631~947~085\\
&&$Q_{13,4}=P_{8,33}$ \cite{Hyperlogproc}&468.038~498~992
\end{tabular}
\end{center}
\vskip1ex
Table 1: List of $\phi^4$-transcendentals up to loop order eight. The list is incomplete at loop order eight.
\end{samepage}

\begin{figure}%
\centering%
\includegraphics[width=0.22\textwidth]{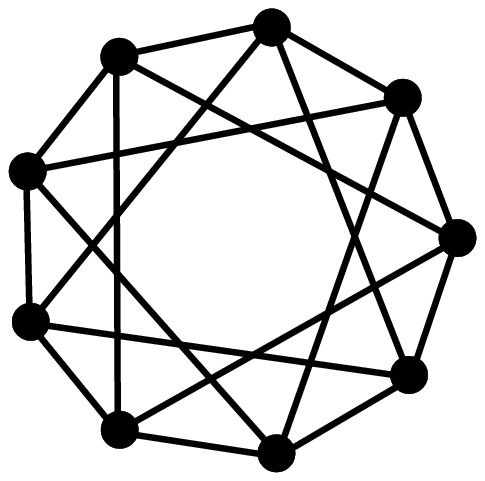}%
\caption{The completed graph $P_{7,11}$.}%
\label{fig:p711}%
\end{figure}\pagebreak

\begin{samepage}
\begin{center}
\begin{tabular}{rr|l}
$\ell$&$\!$wt&base\\[1ex]\hline
6&8&$N_{3,5}=\frac{27}{80}\zeta(5,3)+\frac{45}{64}\zeta(5)\zeta(3)-\frac{261}{320}\zeta(8)$\\[1ex]
7&10&$N_{3,7}=\frac{423}{3584}\zeta(7,3)+\frac{189}{256}\zeta(7)\zeta(3)+\frac{639}{3584}\zeta(5)^2-\frac{7137}{7168}\zeta(10)$\\[1ex]
8&12&$N_{3,9}=\frac{27}{512}\zeta(4,4,2,2)+\frac{55}{1024}\zeta(9,3)+\frac{231}{256}\zeta(9)\zeta(3)+\frac{447}{256}\zeta(7)\zeta(5)-\frac{9}{512}\zeta(3)^4$\\[1ex]
&&\hspace{13mm}$-\frac{27}{448}\zeta(7,3)\zeta(2)-\frac{189}{128}\zeta(7)\zeta(3)\zeta(2)-\frac{1269}{1792}\zeta(5)^2\zeta(2)+\frac{189}{512}\zeta(5,3)\zeta(4)$\\[1ex]
&&\hspace{13mm}$+\frac{945}{512}\zeta(5)\zeta(3)\zeta(4)+\frac{9}{64}\zeta(3)^2\zeta(6)-\frac{7322453}{5660672}\zeta(12)$\\[1ex]
&&$N_{5,7}=-\frac{81}{512}\zeta(4,4,2,2)+\frac{19}{1024}\zeta(9,3)-\frac{477}{1024}\zeta(9)\zeta(3)-\frac{4449}{1024}\zeta(7)\zeta(5)+\frac{27}{512}\zeta(3)^4$\\[1ex]
&&\hspace{13mm}$+\frac{81}{448}\zeta(7,3)\zeta(2)+\frac{567}{128}\zeta(7)\zeta(3)\zeta(2)+\frac{3807}{1792}\zeta(5)^2\zeta(2)-\frac{567}{512}\zeta(5,3)\zeta(4)$\\[1ex]
&&\hspace{13mm}$-\frac{2835}{512}\zeta(5)\zeta(3)\zeta(4)-\frac{27}{64}\zeta(3)^2\zeta(6)+\frac{3155095}{5660672}\zeta(12)$
\end{tabular}
\end{center}
\vskip1ex
Table 2: Conversion of the $N_{a,b}$s in Table 1 into MZVs.
\end{samepage}

\subsection{The number content of $\phi^4$ periods}
Up to five loops there exists at most one $\phi^4$ period per loop order. These periods are the first instances of the infinite family of zig-zag periods (see Figure \ref{fig:zigzag}).
The periods of the zig-zag family are rational multiples of the Riemann zeta function at odd arguments (see Thm.\ \ref{zigzag}).

At six loops the $\phi^4$ periods $P_{6,3}$ and $P_{6,4}$ have a zeta double sum of weight eight. In general, Feynman periods are often multiple zeta values (MZVs) which
are $\QQ$ linear combinations of multiple zeta sums
$$
\zeta(n_d,\ldots,n_1)=\sum_{k_d>\ldots>k_1\geq1}\frac{1}{k_d^{n_d}\!\cdots k_1^{n_1}}\quad\text{with}\quad n_i\in\ZZ_{>0},\; n_d\geq 2.
$$
At seven loops there exists a single period, $P_{7,11}$, which (conjecturally) is not expressible in terms of MZVs. The period $P_{7,11}$ (see Figure \ref{fig:p711})
features an extension of MZVs by (some) sixth (or third) roots of unity. We give the result in the $f^6$ alphabet with respect to the corresponding
even parity Deligne basis \cite{coaction} (a presentation of $P_{7,11}$ with smaller numerators and denominators is given in \cite{MDV}):
\begin{eqnarray}\label{711}
\frac{P_{7,11}}{\ii\sqrt{3}}&=&-\lfrac{332262}{43}f^6_8f^6_3+\lfrac{54918}{55}f^6_6f^6_5+\lfrac{1134}{13}f^6_4f^6_7-\lfrac{1874502}{3485}f^6_2f^6_9\nonumber\\
&&-\numprint{5670}f^6_2f^6_3f^6_3f^6_3-\lfrac{3216912825399005402331281812377062149}{10264478246467100965990650592350882000}(\pi\ii)^{11}.
\end{eqnarray}
There exists a lengthy conversion of the period $P_{7,11}$ in terms of multiple polylogarithms evaluated at primitive sixth roots of unity.
The period $P_{7,11}$ was calculated by Erik Panzer (using his program {\tt HyperInt}) in his PhD thesis \cite{Panzer:PhD} in terms multiple polylogarithms.
With {\tt HyperlogProcedures} the result was converted into the $f$ alphabet \cite{Hyperlogproc}.

At eight loops still most periods are MZVs. Beyond MZVs we found for the period $P_{8,33}$ an expression of weight 13 which is similar to $P_{7,11}$.
Moreover, there exist four periods of an entirely new type. The geometry underlying these periods is no longer a punctured sphere $\CC\backslash\{0,1,\ldots\}$.
Instead of point punctures we obtain in two cases K3 surfaces \cite{K3}. In the other two cases we found a threefold and a fivefold,
respectively. All four varieties are modular of low level \cite{mod}.

Beyond eight loops periods associated to non modular varieties are expected \cite{mod}.

\subsection{The coaction conjectures}
In (\ref{711}) a letter $f^6$ of even weight (subscript) appears only in the leftmost position. This is a consequence of a Galois structure in $\phi^4$ periods.

In \cite{KZ} M. Kontsevich and D. Zagier defined the $\QQ$ algebra of periods $\sP$ as integrals of rational forms over $\QQ$.
Feynman periods are periods in this sense. By general philosophy there should exist a Galois coaction on $\sP$
\cite{Andre:GaloisMotivesTranscendental,GON,MMZ,BrownDecom,Bsv,P3P,Bcoact2},
\begin{equation}\label{co}
\Delta\colon \sP \longrightarrow \sP^\dR \otimes_{\QQ} \sP,
\end{equation}
where the left hand side of the tensor product is the Hopf algebra of de Rham periods. (Note that in some publications $\sP^\dR$ coacts on the right hand side.)
In the special case of $\phi^4$ periods the right hand side (and also the left hand side) of (\ref{co}) seems to be severely restricted.
A mathematical theory with first results is in \cite{Bcoact1}. The data of approximately 300 known $\phi^4$ periods up to 11 loops led to the following possible
scenarios \cite{coaction} for the $\QQ$ algebra $\sP_{\phi^4}$ generated by $\phi^4$ periods. (More precisely, in this article $\Delta$ is the unipotent part of the coaction.)
\begin{itemize}
\setlength{\itemindent}{1cm}
\item[{\bf Scenario 1.}]
\begin{equation}\label{sc1}
\Delta\colon \sP_{\phi^4} \longrightarrow \sP^\dR \otimes_{\QQ} \sP_{\phi^4}.
\end{equation}
\item[{\bf Scenario 2.}]
\begin{equation}\label{sc2}
\Delta'\colon \sP_{\phi^4,\leq n} \longrightarrow \sP^\dR \otimes_{\QQ} \sP_{F,\leq n-1},
\end{equation}
where
$$\Delta'x=\Delta x-1\otimes x$$
is the reduced coaction.
\end{itemize}

Scenario 2 means that for a given $\phi^4$ period of $n$ loops the right hand side of the tensor product is in the $\QQ$ vector space spanned by Feynman periods of all graphs with
at most $n-1$ loops.

Note that $\sP_{\phi^4}$ and $\sP_F$ are very sparse in $\sP$, so that the Scenarios 1 and 2 have huge predictive power on $\sP_{\phi^4}$.
Presumably $\sP_{\phi^4}$ also is very sparse in $\sP_F$. The coaction conjectures are (\ref{sc1}) and (\ref{sc2}).

\subsection{The $c_2$ invariant}
The $c_2$ invariant assigns to every graph with at least three vertices an infinite sequence which is indexed by prime powers $q=p^n$ \cite{SchnetzFq},
$$
c_2:G\mapsto (c_2(G)_q)_q=(c_2(G)_2,c_2(G)_3,c_2(G)_4,c_2(G)_5,c_2(G)_7,\ldots),
$$
where $c_2(G)_q$ is in $\ZZ/q\ZZ$. The $c_2$ invariant is linked to the period integral in four dimensions. With a metric signature $(+,-,+,-)$ instead
of the Euclidean signature it can be defined via the point-count $N_q$ of the singular locus of the period integral (\ref{fdef}) over the finite field $\FF_q$ \cite{c2,Dc2}.
For graphs with at least three vertices, $N_q$ is divisible by $q^2$ and we define
$$
c_2(G)_q\equiv N_q/q^2\mod q.
$$
In practice, it is more efficient to calculate the $c_2$ invariant in parametric space where the above equation is still valid \cite{SchnetzFq,K3}.

The power of the $c_2$ invariant is twofold: First, there exist powerful tools which make it possible to determine the $c_2$ invariant for many graphs.
If the $c_2$ invariant cannot be fully calculated, it is still possible to determine the $c_2$ invariant for small primes \cite{mod,Yc2}. For all completed primitive $\phi^4$ graphs
up to ten loops the $c_2$ is known for at least the first six primes (in most cases much more) \cite{mod}.

Second, the $c_2$ invariant has some predictive power for the period. In particular, if two graphs have the same period, they are conjectured to have the same $c_2$ invariant,
$$
P(G_1)=P(G_2)\Rightarrow c_2(G_1)=c_2(G_2).
$$
All graphs with $c_2$ invariant $-1$ (i.e.\ $c_2(G)_q\equiv-1\mod q$ for all $q$) should have an MZV period. If the $c_2$ invariant is $-z_2$, with
\begin{equation}
z_N(q)=\left\{
\begin{array}{rl}
1&\hbox{if }N|q-1,\\
0&\hbox{if gcd}(N,q)>1,\\
-1&\hbox{otherwise,}
\end{array}\right.
\end{equation}
then the period is expected to be an Euler sum (due to the coaction conjectures, in many cases
these periods are still MZVs \cite{coaction}). The $c_2$ invariant of the periods $P_{7,11}$ and $P_{8,33}$ are $-z_3$.
This links the $c_2$ invariant to the sixth roots of unity which exist in (\ref{711}).

The connection between $\phi^4$ periods and higher dimensional geometries (in some graphs with at least eight loops) is proved with the $c_2$ invariant \cite{K3,mod}.

This led to the proof (assuming standard transcendentality conjectures) that not all $\phi^4$ periods are MZVs or extensions of MZVs by algebraic numbers \cite{K3,BD}.
Concretely, it was shown that $P_{8,37}$ is linked in such a way to the geometry of a K3 surface (which is modular of weight 3 level 7) that the `motivic' period cannot be mixed Tate.

The $c_2$ invariant can also be zero (i.e.\ $c_2(G)_q\equiv0\mod q$ for all $q$). In this case the period (conjecturally) has `weight drop'. This means that the transcendental weight of
the period is strictly smaller than the maximum value $2\ell-3$ in loop order $\ell$.

Note that the $c_2$ invariant seems very sparse for $\phi^4$ periods. It can be conjectured that for any fixed dimension (of the lowest dimensional manifold whose point-count gives the $c_2$)
there exist only finitely many $c_2$s to all loop orders. For dimension 0 we have $-z_N$ for $N=1,2,3,4$. In dimension 1 there seems to be nothing.
We possibly only have three two-folds (which are modular K3s of low level) and five modular three-folds in $\phi^4$ \cite{mod}.

This is in stark contrast to the situation in $\sP_F$ (including non-$\phi^4$ graphs) where we expect that we see basically any geometry over $\ZZ$ in the $c_2$.

An interesting recent result by K. Yeats (private communication 2017) is that we seem to see any finite sequence of prime remainders in the $c_2$ of $\phi^4$ graphs.
Note that this does not contradict the sparsity of $\phi^4$ $c_2$s (because $c_2$s are infinite sequences).

\subsection{The Hepp invariant}\label{Hepp}
For any graph $G$ with edge weights $\{\nu_e\}$, $e\in\sE(G)$ and $a\in\CC$ we recursively define the following Hepp invariant:
\begin{equation}
H_a(G)=\frac{\sum_{e\in\sE(G)}\nu_eH_a(G\backslash e)}{N_G-ah_1(G)},
\end{equation}
where $N(G)=\sum_{e\in\sE(G)}\nu_e$ is the sum of edge weights and $h_1(G)$ is the (Betti) number of independent cycles in $G$.
A graph with no edges has Hepp invariant 1.

\begin{lem} The Hepp invariant has the following properties:
\begin{enumerate}
\item If the removal of the edge $e$ disconnects the graph $G$ then $H_a(G)=H_a(G\backslash e)$.
\item If G has vertex connectivity $\leq1$ then $H_a$ factorizes, i.e.\ $H_a(G)$ is the product of the Hepp invariants of its components (cutting at split vertices without
removing edges).
\item Assume a vertex $v$ in $G$ is adjacent to exactly two edges $e$ and $f$ with weights $\nu_e$ and $\nu_f$.
We construct a smaller graph $G'$ by contracting the edge $e$ (or $f$) in $G$ and giving $f$ ($e$) the weight $\nu_e+\nu_f$.
Then $H_a(G)=H_a(G')$.
\end{enumerate}
\end{lem}
\begin{proof}
Straight forward induction over the number of edges in $G$.
\end{proof}
The above lemma can be used to efficiently calculate the Hepp invariant of reasonably large graphs.
Note that every forest has Hepp invariant 1.

If a graph $G$ has a Feynman period in $d=4$ dimensions then $H_a(G)$ trivially has a simple pole at $a=2$. We define the Hepp period of $G$ as the residue,
\begin{equation}
H(G)=-h_1(G)2^{-h_1(G)}\mathrm{res}_{a=2}H_a(G).
\end{equation}
Erik Panzer recently found that the Hepp period is closely related to the Feynman period \cite{EPHepp, KP6loopbeta}.
It approximates the Feynman period surprisingly well,
\begin{equation}\label{Heppapprox}
P(G)\approx0.545^{h_1(G)-1}H(G)^{1.355}.
\end{equation}
with an error of a few percent. The numerical estimates for the unknown eight loop periods in Table 3 were obtained by a refinement of (\ref{Heppapprox}).

Even more surprisingly, the Hepp period seems to know all identities between periods.
\begin{conj}[E. Panzer]\label{Heppcon}
\begin{equation}
P(G_1)=P(G_2)\Leftrightarrow H(G_1)=H(G_2)
\end{equation}
\end{conj}
At eight loops the above conjecture requires $P_{8,30}=P_{8,36}$ and $P_{8,31}=P_{8,35}$. There is no sequence of twists or Fourier identities known that proves these identities.

Assuming Conj.\ \ref{Heppcon} for completion, the product identity (Sect.\ \ref{product}) was proved by E. Panzer (private communication, May 3, 2016) and independently by K. Yeats
(private communication, June 7, 2016) to hold for Hepp periods.

It is easy to see that $P(G)\leq 2^{h_1(G)}H(G)$ which is the classical Hepp bound for periods. Erik Panzer has achieved a certain refinement of this bound \cite{EPHepp}.
We conjecture that the Hepp period is a (crude) upper bound for the period, $P(G)\leq H(G)$.

\section{Functions}
Let $G$ be a graph with external vertices $0,1,z$ such that the graphical function $f_G^{(\lambda)}$ exists.

\subsection{Completion}\label{gfcompletion}
Like in the case of periods, completion exploits conformal symmetry to handle equivalence classes of closely related graphical functions.
Again, we add an external vertex `$\infty$' which connects to all internal vertices of $G$ with weights that give the internal vertices weighted degree $2d/(d-2)$.
Now, we add edges $\{z,\infty\}$, $\{0,1\}$, $\{0,\infty\}$, $\{1,\infty\}$ such that all external vertices have weighted degree 0. This provides the completion $\overline{G}$ of $G$.
The graphical function $f_{\overline{G}}^{(\lambda)}$ of the completed graph $\overline{G}$ is defined as the graphical function of
$\overline{G}\setminus\infty$ (edges adjacent to $\infty$ have quadric 1).
Clearly, the graphical function does not change under completion. Completion is always possible and unique. In three or four dimensions we also
know that the completed graph $\overline{G}$ has integer edge weights if $G$ has (Lemma 3.18 in \cite{gf}).

A permutation of external vertices in a completed graph results in a M\"obuis transformation of the argument $z$:

\begin{thm}[Theorem 3.20 in \cite{gf}]\label{completionthm}
Let $\sigma:\{0,1,z,\infty\}\to\{0,1,\phi(z),\infty\}$ be a permutation of $\{0,1,z,\infty\}$ followed by a M\"obius transformation $z\mapsto\phi(z)$ such
that $\sigma$ preserves the cross ratio $(0,1;z,\infty)$, i.e.
$$
\frac{-z}{1-z}=\frac{(\sigma(0)-\sigma(z))(\sigma(1)-\sigma(\infty))}{(\sigma(1)-\sigma(z))(\sigma(0)-\sigma(\infty))}.
$$
Then $f_{\overline{G}}^{(\lambda)}=f_{\sigma(\overline{G})}^{(\lambda)}$, where the M\"obius transformation in the label of $\sigma(\overline{G})$ acts on the argument of the graphical function.
In particular, $f_{\overline{G}}^{(\lambda)}$ is invariant under double transpositions of external labels.
\end{thm}

Because edges between external vertices produce trivial factors, completed graphical functions with stripped off edges between external vertices are
equivalence classes graphical functions of the same type. In four dimensions all those graphical functions with edge weights one and at most seven vertices are known \cite{Hyperlogproc}.

\begin{figure}%
\begin{center}
\fcolorbox{white}{white}{
  \begin{picture}(282,90)(-30,5)
    \SetWidth{1.0}
    \SetColor{Black}
    \Vertex(48,74){2.8}
    \Vertex(48,10){2.8}
    \Vertex(16,42){2.8}
    \Vertex(48,42){2.8}
    \Vertex(176,10){2.8}
    \Vertex(176,74){2.8}
    \Vertex(144,42){2.8}
    \Vertex(176,42){2.8}
    \Vertex(208,42){2.8}
    \Line(48,74)(48,10)
    \Line(16,42)(48,42)
    \Line(144,42)(208,42)
    \Line(176,74)(176,10)
    \Photon(208,42)(176,10){2.5}{5}
    \Photon(176,74)(144,42){2.5}{5}
    \Text(56,7)[lb]{\Black{$z$}}
    \Text(56,72)[lb]{\Black{$0$}}
    \Text(14,28)[lb]{\Black{$1$}}
    \Text(184,72)[lb]{\Black{$0$}}
    \Text(142,28)[lb]{\Black{$1$}}
    \Text(184,7)[lb]{\Black{$z$}}
    \Text(206,50)[lb]{\Black{$\infty$}}
    \Text(-20,40)[lb]{\Black{$G=$}}
    \Text(108,40)[lb]{\Black{$\overline{G}=$}}
  \end{picture}
}
\end{center}
\caption{A four dimensional graphical function with three edges and its completion.
The wiggly lines refer to edge weights $-1$.}
\label{fig:Li2}
\end{figure}
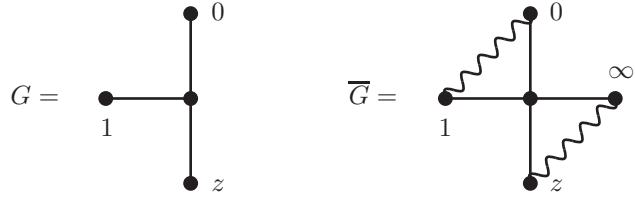

\begin{ex}\label{ex1}
The smallest non-trivial graphical function has four vertices and three edges (see Figure \ref{fig:Li2}).
In four dimensions, this graphical function is given by the Bloch-Wigner dilogarithm (see e.g.\ \cite{Zagierdilog}).
$$
f_G^{(\lambda)}(z)=\frac{4\ii D(z)}{z-\zz},
$$
with
\begin{equation}\label{D}
D(z)=\mathrm{Im}(\Li_2(z)+\ln(1-z)\ln|z|).
\end{equation}
Theorem \ref{completionthm} and $f_G^{(\lambda)}(z)=f_G^{(\lambda)}(\zz)$ reflect the symmetries of $D$.
\end{ex}

\subsection{Existence}
Existence of the graphical function $f_G^{(\lambda)}$ is best formulated in terms of the completion $\overline{G}$.

\begin{thm}[Lemma 3.19 in \cite{gf}]\label{gfexistence}
The graphical function $f_{\overline{G}}^{(\lambda)}$ exists if and only if for any vertex subset $V\subset\sV(\overline{G})$ with at least two vertices and at most one external vertex
the following inequality holds,
\begin{equation}\label{ultraviolet}
(d-2)N_g<d(|V|-1),
\end{equation}
where $N_g$ is the sum of edge weights in the induced subgraph $g$ of $V$ in $\overline{G}$ (i.e.\ $g$ has all the edges of $\overline{G}$ that have both vertices in $V$).
\end{thm}

In Example \ref{ex1} we only have the case $|V|=2$ and $N_g=1$, so that $f_G^{(\lambda)}(z)$ exists. In $d$ dimensions the edge $\{z,\infty\}$ has weight $(6-d)/(d-2)$
so that (\ref{ultraviolet}) implies that the graphical functions of Example \ref{ex1} exists in any dimension greater than three.

\subsection{General properties of graphical functions}

Graphical functions should be considered as functions on $\overline\CC=\CC\cup\{\infty\}$. They have the following general properties:

\begin{thm}\label{generalproperties} Let $G$ be a graph such that the graphical function $f_G^{(\lambda)}$ exists.
\begin{enumerate}
\item[(G1)]
\begin{equation}\label{reflection}
f_G^{(\lambda)}(z)=f_G^{(\lambda)}(\zz).
\end{equation}
\item[(G2)] $f_G^{(\lambda)}$ is a positive single-valued real analytic function on $\overline\CC\setminus\{0,1,\infty\}$.
\item[(G3)] The radius of convergence of the real analytic expansion of $f_G^{(\lambda)}$ at $z_0\in\overline\CC\setminus\{0,1,\infty\}$ is the distance from $z_0$ to the nearest singularity
of $f_G^{(\lambda)}$.

Let $\nu_z^>$ ($\nu_z^<$) be the sum of positive (negative) weights of edges adjacent to $z$. Let the dimension $d=2\lambda+2$ be even.
Then, if $z_0\in\{0,1\}$, we have for $|z-z_0|<1$ and coefficients $c_{\ell,m,n}(z_0)\in\CC$:
\begin{equation}\label{01expansion}
f_G^{(\lambda)}(z)=\sum_{\ell=0}^{V^{\mathrm{int}}}\sum_{m=M_{z_0}}^\infty\sum_{n=N_{z_0}}^\infty c_{\ell,m,n}(z_0)\log^\ell[(z-z_0)(\zz-\overline{z_0})](z-z_0)^m(\zz-\overline{z_0})^n,
\end{equation}
where
$$
M_{z_0},N_{z_0}>-\lambda\nu_z^>.
$$
If $z_0=\infty$ we have for $|z|>1$ and coefficients $c_{\ell,m,n}(\infty)\in\CC$:
\begin{equation}\label{inftyexpansion}
f_G^{(\lambda)}(z)=\sum_{\ell=0}^{V^{\mathrm{int}}}\sum_{m=-\infty}^{M_\infty}\sum_{n=-\infty}^{N_\infty}c_{\ell,m,n}(\infty)\log^\ell(z\zz)z^m\zz^n,
\end{equation}
where
$$
M_\infty,N_\infty<-\lambda\nu_z^<.
$$
\end{enumerate}
\end{thm}

Property (G1) is immediate by symmetry. It also follows from the parametric representation of graphical functions \cite{par}. Property (G2) is proved in \cite{par} while
(G3) will be handled in \cite{gf2}. Conjecture \ref{asymptconj} gives additional information on the leading terms of the above expansions.
In the case of odd dimensions there exist expansions at 0, 1, $\infty$ which are similar to (\ref{01expansion}) and (\ref{inftyexpansion}) with a square root
$\sqrt{(z-z_0)(\zz-\zz_0)}$ \cite{gf2}.

\subsection{Appending edges}\label{append}
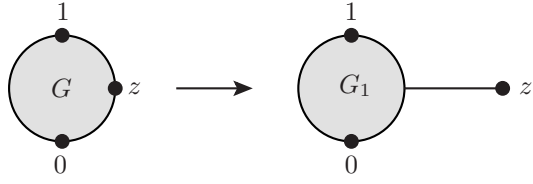
\begin{figure}
\begin{center}
\fcolorbox{white}{white}{
  \begin{picture}(336,79) (-52,-5)
    \SetWidth{0.8}
    \SetColor{Black}
    \GOval(32,34)(20,20)(0){0.882}
    \Vertex(32,54){2.8}
    \Vertex(32,14){2.8}
    \Vertex(52,34){2.8}
    \Text(30,60)[lb]{\normalsize{\Black{$1$}}}
    \Text(29,2)[lb]{\normalsize{\Black{$0$}}}
    \Text(57,32)[lb]{\normalsize{\Black{$z$}}}
    \Text(28,31)[lb]{\normalsize{\Black{$G$}}}

    \Line[arrow,arrowpos=1,arrowlength=5,arrowwidth=2,arrowinset=0.2](75,34)(100,34)

    \GOval(141,34)(20,20)(0){0.882}
    \Line(161,34)(196,34)
    \Vertex(141,54){2.8}
    \Vertex(141,14){2.8}
    \Vertex(198,34){2.8}
    \Text(139,60)[lb]{\normalsize{\Black{$1$}}}
    \Text(139,2)[lb]{\normalsize{\Black{$0$}}}
    \Text(205,32)[lb]{\normalsize{\Black{$z$}}}
    \Text(137,31)[lb]{\normalsize{\Black{$G_1$}}}
  \end{picture}
}
\end{center}
\caption{Appending an edge to the vertex $z$ in $G$ gives $G_1$.}
\label{fig:append}
\end{figure}

Edges between external vertices give factors in a graphical function. In the case that an edge is appended to the vertex $z$ creating a new vertex $z$ (see Figure \ref{fig:append})
the graphical functions are related by a differential equation \cite{gf,parG}.
\begin{lem}\label{appendlem}
In the setup of Figure \ref{fig:append} we have
\begin{equation}\label{diffeq}
\left(-\frac{1}{(z-\zz)}\,\partial_z\partial_{\zz}\,(z-\zz)+\frac{\lambda-1}{z-\zz}(\partial_z-\partial_{\zz})\right)f^{(\lambda)}_{G_1}(z)
=\frac{1}{\Gamma(\lambda)}f^{(\lambda)}_G(z),
\end{equation}
where $\Gamma(\lambda)=\int_0^\infty x^{\lambda-1}\exp(-x)\dd x$ is the gamma function.
\end{lem}

The differential equation is particularly simple in $d=4$ dimensions ($\lambda=1$). In this case we (uniquely) obtain the graphical function of $G_1$ by single-valued
integration with respect to $z$ and $\zz$ (see \cite{gf}).

\begin{lem}\label{kernellem}
Let $\lambda=1$. The differential operator on the left hand side of (\ref{diffeq}) has trivial kernel in the space of functions with general properties (G1) -- (G3).
\end{lem}
\begin{proof}
Assume $f$ with properties (G1) -- (G3) is in the kernel of the differential operator for $\lambda=1$.
Because of (G2), $g(z)=\partial_z(z-\zz)f(z)$ is meromorphic. We use (\ref{01expansion}) with $N_0,M_0,N_1,M_1\geq0$ (because $\nu_z^>=1$ in $G_1$) and
conclude that $g$ is holomorphic on $\CC$. From (\ref{inftyexpansion}) with $N_\infty,M_\infty\leq-1$ we get $g(\infty)=0$, hence, by Liouville's theorem, $g=0$.
Therefore $f=h(\zz)/(z-\zz)$ for some anti-holomorphic function $h$. With (G1) we obtain $h=0$.
\end{proof}

Beginning with the empty graphical function we can construct many graphical functions by appending edges (see \cite{gf, ZZ}). A particularly simple class of such graphical functions
is handled by the following theorem.
\begin{thm}
Let $G$ be a graph with external vertex width four (constructible in \cite{gf}), i.e.\ $G$ can be constructed from the empty graph by adding edges between external vertices,
permuting external vertices, and appending edges. If $G$ has no edges between external vertices then we obtain in $d=4$ dimensions,
\begin{equation}
f_G^{(1)}(z)=P(z)/(z-\zz),
\end{equation}
where $P$ is a single-valued multiple polylogarithm \cite{BrSVMP} of weight $2V^{\mathrm{int}}(G)$.
\end{thm}
The proof of the theorem will be in \cite{gf2}. Note that the theorem is consistent with Example \ref{ex1}.

In $d=4-\epsilon$ dimensions edges can still be appended if one expands $f^{(\lambda)}_G$ into a Laurant series in $\epsilon$ to a given order. The procedure, however, is
more subtle, see Sect.\ \ref{dimreg}.

\subsection{Identities}
There exist many identities for graphical functions. A Fourier-identity relating planar duals was proved in \cite{par}. Like in the case of periods there also exists a twist identity.
Some complicated graphical functions can be calculated by a Gegenbauer technique \cite{gf,gf2}. All known identities are included in {\tt HyperlogProcedures} which (among other things)
can calculate many graphical functions \cite{Hyperlogproc}. If a graphical function is inaccessible to all of these methods then sometimes it can still be calculated by
parametric integration, due to F. Brown \cite{BrH1,BrH2} and E. Panzer \cite{Panzer:PhD,Panzer:HyperInt}.

In \cite{fishnet} B. Basso and L. J. Dixon provide an intriguingly simple formula for graphical functions of the $(m,n)$ `fishnet' topology. The result is conjectured from
methods and properties of $N=4$ Supersymmetric Yang-Mills Theory. For $m,n>1$ the fishnet topology is inaccessible to the tools presented here. This indicates that much more
powerful methods may exist for the calculation of graphical functions.

\subsection{From graphical functions to periods}
\begin{figure}
\begin{center}
\fcolorbox{white}{white}{
  \begin{picture}(336,75) (24,-38)
    \SetColor{Black}
    \SetWidth{0.8}
    \Vertex(53,33){2.8}
    \Vertex(64,-13){2.8}
    \Vertex(40,-13){2.8}
    \Vertex(76,5){2.8}
    \Vertex(32,24){2.8}
    \Vertex(71,24){2.8}
    \Vertex(28,3){2.8}
    \Arc(52,9)(24.413,145,505)
    \Arc(86.871,38.557)(34.324,-170.683,-108.465)
    \Arc(85.63,2.475)(26.026,124.203,213.79)
    \Arc(18.357,-68.443)(71.964,51.658,83.102)
    \Arc(15.214,0.679)(28.734,-26.183,54.255)
    \Arc(10.5,43)(43.661,-66.371,-13.241)
    \Arc(51,61.5)(42.5,-118.072,-61.928)
    \Arc(73.318,-34.136)(39.172,87.539,145.591)
    \Text(44,-35)[lb]{\normalsize{\Black{$\overline{Z_5}$}}}

    \Vertex(133,33){2.8}
    \Vertex(132,-16){2.8}
    \Vertex(158,9){2.8}
    \Vertex(108,9){2.8}
    \Vertex(150,26){2.8}
    \Vertex(151,-8){2.8}
    \Vertex(115,-8){2.8}
    \Vertex(116,26){2.8}
    \Arc(133,9)(24.413,145,505)
    \Arc(162.86,39.265)(30.51,-168.151,-97.268)
    \Arc(173.5,8.5)(28.504,142.125,217.875)
    \Arc(168.6,-28.567)(39.034,105.757,159.662)
    \Arc(131.231,-54.462)(50.493,66.95,109.953)
    \Arc(97.667,10)(24.333,-41.112,41.112)
    \Arc[clock](106.1,-16.8)(25.962,83.587,3.976)
    \Arc(97.4,46.567)(39.034,-74.243,-20.338)
    \Arc(133.405,57.177)(36.204,-120.555,-62.718)
    \Text(129,-35)[lb]{\normalsize{\Black{$\overline{Z_6}$}}}

    \Vertex(191,-17){2.8}
    \Vertex(207,6){2.8}
    \Vertex(214,-17){2.8}
    \Vertex(231,6){2.8}
    \Vertex(236,-17){2.8}
    \Vertex(252,6){2.8}
    \Line(191,-17)(236,-17)
    \Line(191,-17)(207,6)
    \Line(207,6)(214,-17)
    \Line(214,-17)(231,6)
    \Line(231,6)(236,-17)
    \Line(236,-17)(252,6)
    \Line(252,6)(207,6)
    \Arc[clock](221.5,-5.5)(33.534,-159.944,-339.944)
    \Text(211,-35)[lb]{\normalsize{\Black{$Z_5$}}}

    \Vertex(284,-17){2.8}
    \Vertex(299,6){2.8}
    \Vertex(308,-17){2.8}
    \Vertex(321,6){2.8}
    \Vertex(332,-17){2.8}
    \Vertex(345,6){2.8}
    \Vertex(355,-17){2.8}
    \Line(284,-17)(355,-17)
    \Line(284,-17)(299,6)
    \Line(299,6)(345,6)
    \Line(299,6)(308,-17)
    \Line(308,-17)(321,6)
    \Line(321,6)(332,-17)
    \Line(332,-17)(345,6)
    \Line(345,6)(355,-17)
    \Arc[clock](319.375,-7.6)(36.602,-165.119,-373.267)
    \Text(315,-35)[lb]{\normalsize{\Black{$Z_6$}}}
  \end{picture}
}
\end{center}
\caption{Completed ($\overline{Z_\bullet}$) and uncompleted ($Z_\bullet$) zig-zag graphs with five and six loops.}
\label{fig:zigzag}
\end{figure}
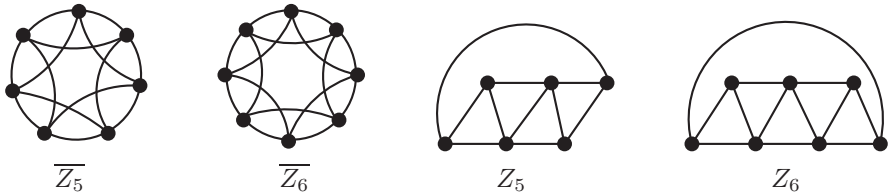

There exist several options to derive Feynman periods from graphical functions.

First, one can specify the variable $z$ in $f_G^{(\lambda)}(z)$ to 0, 1, or $\infty$.
In the case of the zig-zag graphs depicted in Figure \ref{fig:zigzag} this method leads to a proof of a conjecture by D. Broadhurst and D. Kreimer in 1995 \cite{BK}.

\begin{thm}[F. Brown, O. Schnetz, \cite{ZZ}]\label{zigzag}
The period of the graph $Z_n$ is given by
\begin{equation}\label{PZ}
P(Z_n)=4\frac{(2n-2)!}{n!(n-1)!}\Big(1-\frac{1-(-1)^n}{2^{2n-3}}\Big)\zeta(2n-3).
\end{equation}
\end{thm}

Second, in integer dimensions $d\geq3$ one can integrate $f_G^{(\lambda)}(z)$ over the external variable $z$ \cite{gf}.
This effectively makes $z$ an internal variable. For a general graph in integer dimensions this is the best method to calculate Feynman periods.
It is used in \cite{Hyperlogproc}.

Third, in the case of dimensional regularization we need to treat the dimension $d$ as a parameter and expand periods in $4-\epsilon$ `dimensions' at
$\epsilon=0$.
Using the second method one obtains an integration measure of $(z-\zz)^{2-\epsilon}$ which does not expand in $\epsilon$ into GSVHs (see Sect.\ \ref{GSVH}).
One can resort to the following procedure to integrate over $z$:
Add an edge of weight $-1$ between the external vertices 0 and $z$. Then, append an edge to the vertex $z$ creating a new vertex $z$.
Set $z=0$ so that the newly appended edge cancels the previously added edge of weight $-1$. This integrates over $z$ in any possibly non-integer dimension $d$.
This method is used for perturbative calculations in dimensionally regularized $\phi^4$ theory (see Sect.\ \ref{phi4}) \cite{Hyperlogproc}. It is also possible to devise a more direct
method to perform $4-\epsilon$ dimensional integrations (private communication with E. Panzer). However, this is conceptually more demanding and in practice the
method that appends an edge works quite well.

\subsection{Generalized single-valued hyperlogarithms}\label{GSVH}

Many graphical functions can be expressed in terms of iterated integrals \cite{Chen}. Although, by (G2) of Sect.\ \ref{generalproperties}, graphical functions have only
singularities at 0, 1, and $\infty$ it turns out that single-valued multiple polylogarithms \cite{BrSVMP} (i.e.\ letters 0 and 1 in the iterated
integrals) are too restrictive. Let us consider the following example (which shows the same mechanism although it is not a graphical function).

\begin{ex}\label{exGSVH}
let $r>0$. Define
\begin{equation}\label{GSVHex}
f(z)=\frac{\log(z\zz/r)}{z-r/\zz}.
\end{equation}
Then $f$ is real analytic on $\overline\CC\setminus\{0,\infty\}$ because the zero locus $z=r/\zz$ in the denominator is canceled by the numerator.
By general principles, there should exist a single-valued primitive of $f$ which also is real analytic on $\overline\CC\setminus\{0,\infty\}$.
By single-valuedness the primitive is determined up to a rational function in $\zz$. It is unique in the space of hyperlogarithms which vanish at 0.
\end{ex}

We generalize the above example in the sense that we consider bilinear denominators in $z$ and $\zz$ (i.e.\ $a+bz+c\zz+dz\zz$, with $a,b,c,d\in\CC$).
We require that non point-like zero loci in the denominator are canceled by the numerator. This leads to functions with point-like singularities, so that single-valuedness makes sense.
We can construct single-valued primitives of single-valued functions leading to the space of generalized single-valued hyperlogarithms (GSVHs).

There also exist GSVHs with no singularities on the complex plane. A trivial example of this type is $\log(1+z\zz/r)$ for $r>0$. In the context of graphical functions,
however, we have GSVHs of the type in Example \ref{exGSVH}. A QFT example of weight three is the single-valued primitive of $4\ii D(z)/(z-\zz)$ where $D$ is the Bloch-Wigner dilogarithm
(\ref{D}) (see also \cite{Duhr}).

Let $\sG$ be the space of GSVHs. Upon differentiation with respect to $z$ and $\zz$ we obtain the spaces
$\partial_z\sG$, $\partial_{\zz}\sG$, $\partial_{\zz}\partial_z\sG$. For example, $f$ in (\ref{GSVHex}) is in $\partial_z\sG$.
We would like to construct an algorithm for single-valued integration in $\partial_z\sG$.

Functions in $\partial_z\sG$ have expansions (\ref{01expansion}) and (\ref{inftyexpansion}) for general $z_0\in\overline\CC$.
The holomorphic residue res$_{z_0}$ is the coefficient $c_{0,-1,0}(z_0)$, whereas the anti-holomorphic residue $\overline{\mathrm{res}}_{z_0}$ is the coefficient $c_{0,0,-1}(z_0)$.
let $\pi_0$ ($\overline{\pi_0}$) be the projection onto the (anti-)residue free part,
\begin{eqnarray*}
\pi_0:\partial_z\sG\to\partial_z\sG&,&\quad f(z)\mapsto f(z)-\sum_{z_0\in\CC}\frac{\mathrm{res}_{z_0}(f)}{z-z_0},\\
\overline{\pi_0}:\partial_{\zz}\sG\to\partial_{\zz}\sG&,&\quad f(z)\mapsto f(z)-\sum_{z_0\in\CC}\frac{\overline{\mathrm{res}}_{z_0}(f)}{\zz-\overline{z_0}}.
\end{eqnarray*}

An efficient method to obtain single-valued primitives relies on the commutative hexagon in Figure \ref{fig:GSVH}, where $\int_{\mathrm{sv}}$ stands for single-valued integration.

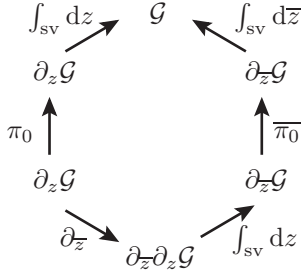
\begin{figure}
\begin{center}
\fcolorbox{white}{white}{
  \begin{picture}(160,135) (19,-11)
    \SetWidth{1.0}
    \SetColor{Black}
    \Line[arrow,arrowpos=1,arrowlength=5,arrowwidth=2,arrowinset=0.2](48,38)(48,56)
    \Line[arrow,arrowpos=1,arrowlength=5,arrowwidth=2,arrowinset=0.2](128,38)(128,56)
    \Line[arrow,arrowpos=1,arrowlength=5,arrowwidth=2,arrowinset=0.2](54,17)(71,7)
    \Line[arrow,arrowpos=1,arrowlength=5,arrowwidth=2,arrowinset=0.2](122,77)(105,87)
    \Line[arrow,arrowpos=1,arrowlength=5,arrowwidth=2,arrowinset=0.2](105,7)(122,17)
    \Line[arrow,arrowpos=1,arrowlength=5,arrowwidth=2,arrowinset=0.2](54,77)(71,87)
    \Text(42,24)[lb]{\Black{$\partial_z\sG$}}
    \Text(122,24)[lb]{\Black{$\partial_{\zz}\sG$}}
    \Text(122,65)[lb]{\Black{$\partial_{\zz}\sG$}}
    \Text(42,65)[lb]{\Black{$\partial_z\sG$}}
    \Text(32,43)[lb]{\Black{$\pi_0$}}
    \Text(133,43)[lb]{\Black{$\overline{\pi_0}$}}
    \Text(40,84)[lb]{\Black{$\int_{\mathrm{sv}}\dd z$}}
    \Text(118,84)[lb]{\Black{$\int_{\mathrm{sv}}\dd\zz$}}
    \Text(77,-5)[lb]{\Black{$\partial_{\zz}\partial_z\sG$}}
    \Text(86,88)[lb]{\Black{$\sG$}}
    \Text(52,2)[lb]{\Black{$\partial_{\zz}$}}
    \Text(117,0)[lb]{\Black{$\int_{\mathrm{sv}}\dd z$}}
  \end{picture}
}
\end{center}
\caption{The inductive construction of GSVHs by a commutative hexagon.}
\label{fig:GSVH}
\end{figure}

\begin{thm}
The diagram in Figure \ref{fig:GSVH} commutes.
\end{thm}

The proof of the theorem will be in \cite{GSVH}. By virtue of Figure 6 we can express a single-valued primitive with respect to $z$ also as a single-valued primitive with respect
to $\zz$. Because single-valued integration in $z$ (resp.\ $\zz$) equals ordinary integration up to an anti-holomorphic (resp.\ holomorphic) function,
knowing both integrands determines the single-valued primitive up to a constant (which is fixed by the condition that the single-valued primitive vanishes at $z=0$).
Using integration by parts at the bottom right arrow in Figure 6, we can reduce single-valued integration to lower weights.

If $f\in\partial_z\sG$ has weight 0 then the integration from $\partial_{\zz}\partial_z\sG$ to $\partial_{\zz}\sG$ is purely rational. The single-valued integration of the
residues (which has to be done separately) is trivial: The single-valued primitive of $1/(z-c)$ is $\log[(z-c)(\zz-\overline{c})]$ for any $c\in\CC$.

\begin{ex}[Example \ref{exGSVH} continued]
In terms of iterated integrals (writing from right to left) we obtain for the numerator of the integrand
$$
\log(z\zz/r)=I(z,0,0)+I(\zz,0,0)-I(r,0,0).
$$
The single-valued primitive of $f$ has the general form (note that $f$ is residue-free)
\begin{equation}\label{exeq1}
\int_{\mathrm{sv}}f(z)\dd z=I(z,r/\zz,0,0)+I(z,r/\zz,0)[I(\zz,0,0)-I(r,0,0)]+g(\zz)
\end{equation}
for some anti-holomorphic $g$. Differentiation with respect to $\zz$ yields
$$
\partial_{\zz}f(z)=\frac{1}{z\zz-r}-r\frac{\log(z\zz/r)}{(z\zz-r)^2}.
$$
Using integration by parts on the second term we find
$$
\int_{\mathrm{sv}}-r\frac{\log(z\zz/r)}{(z\zz-r)^2}\dd z=\frac{r\log(z\zz/r)}{\zz(z\zz-r)}-\int_{\mathrm{sv}}\frac{r}{z\zz(z\zz-r)}\dd z.
$$
Adding the two terms which remain to be integrated, the factor $(z\zz-r)$ cancels (it has to) and we obtain
$$
\int_{\mathrm{sv}}\partial_{\zz}f(z)\dd z=\frac{r\log(z\zz/r)}{\zz(z\zz-r)}+\frac{\log(z\zz)}{\zz}.
$$
In fact, there is an ambiguity in form of a rational function in $\zz$. However, because the result has to be in $\partial_{\zz}\sG$, the ambiguity can only
be an anti-residue and it is removed by the projection $\overline{\pi_0}$. The above expression has an anti-residue at $z=0$ with value $\log(r)$. Subtraction yields
$$
\overline{\pi_0}\int_{\mathrm{sv}}\partial_{\zz}f(z)\dd z=\frac{\log(z\zz/r)}{\zz-r/z}.
$$
Using the commutative hexagon we obtain by integration with respect to $\zz$,
\begin{equation}\label{exeq2}
\int_{\mathrm{sv}}f(z)\dd z=I(\zz,r/z,0,0)+I(\zz,r/z,0)[I(z,0,0)-I(r,0,0)]+h(z)
\end{equation}
for some holomorphic function $h$. If we write (\ref{exeq2}) as hyperlogarithms in $z$ with coefficients which are hyperlogarithms in $\zz$ we get (\ref{exeq1}) with $h(z)$
instead of $g(\zz)$. (Alternatively we may treat $z$ and $\zz$ as independent variables in (\ref{exeq2}) and consider the limit $z\to0$.)
We conclude that in this example $h(z)=g(\zz)$ is a constant. This constant is zero because the single-valued integral is required to vanish at $z=0$.
\end{ex}

At four dimensions some (few) graphical functions exist which can be expressed in terms of ordinary single-valued multiple polylogarithms (see e.g.\ \cite{ZZ}).
A large majority of graphical functions which can be expressed in terms of iterated integrals are GSVHs which are not single-valued multiple polylogarithms.
In $4-\epsilon$ `dimensions' every non-trivial graphical function expands in $\epsilon$ with coefficients which are not single-valued multiple polylogarithms (but often GSVHs).

\subsection{$4-\epsilon$ dimensions}\label{dimreg}
We can use the parametric representation of graphical functions \cite{par} to define graphical function for non-integer $d$. Using $4-\epsilon$ `dimensions' regularizes graphical
functions: Graphical functions that diverge in 4 dimensions may (and generically do) exist in $4-\epsilon$ dimensions.

Although general properties (G1) and (G2) in Theorem \ref{generalproperties} remain valid, graphical functions in non-integer dimensions can hardly be calculated in terms of known functions.
However, in QFT it suffices to know their Laurant expansions at $\epsilon=0$ to some (small) order in $\epsilon$. For the Laurant coefficients
(G3) holds and often they can be expressed in terms of GSVHs.

The main tool for constructing these coefficients is again appending edges (see Sect.\ \ref{append}).
Equation (\ref{diffeq}) can be solved iteratively in powers of $\epsilon$. However, in this approach we cannot directly use Lemma \ref{kernellem} to avoid the kernel of the
differential operator. We first need to subtract poles at $z=z_0$, $z_0\in\{0,1\}$, which are of order four (or higher) in $|z-z_0|$. (Using completion an analogous subtraction
is necessary at $z_0=\infty$.) For these singular contributions we need exact results; it is not sufficient to know them to a limited order in $\epsilon$.
In a renormalizable QFT we only have to deal with `logarithmic' singularities. I.e.\ the case of poles of order four suffices.
This leading order asymptotic behavior of graphical functions is obtained by the following result which we leave as a well tested conjecture.

\begin{figure}
\begin{center}
\fcolorbox{white}{white}{
  \begin{picture}(352,80) (19,-30)
    \SetWidth{1.0}
    \SetColor{Black}
    \GOval(64,16)(16,16)(0){0.882}
    \GOval(128,16)(16,16)(0){0.882}
    \GOval(240,16)(16,16)(0){0.882}
    \GOval(304,16)(16,16)(0){0.882}
    \SetWidth{2.0}
    \Line(32,16)(48,16)
    \Line(144,16)(160,16)
    \Line(208,16)(224,16)
    \Line(320,16)(336,16)
    \Line(240,0)(240,-16)
    \Line(304,0)(304,-16)
    \Line(64,0)(96,-16)
    \Line(128,0)(96,-16)
    \Line(80,16)(112,16)
    \Line(304,-16)(336,16)
    \Vertex(32,16){2.8}
    \Vertex(96,-16){2.8}
    \Vertex(240,-16){2.8}
    \Vertex(304,-16){2.8}
    \Vertex(160,16){2.8}
    \Vertex(208,16){2.8}
    \Vertex(336,16){2.8}
    \Arc[clock](86.222,-13.333)(61.648,151.587,47.337)
    \Arc[clock](106,-12)(60.828,133.668,27.408)
    \Arc(312,8)(25.298,161.565,251.565)
    \Arc(296,8)(25.298,71.565,288.435)
    \Text(172,9)[lb]{\LARGE\Black{$\sim\sum$}}
    \Text(94,-29)[lb]{\Black{$0$}}
    \Text(238,-29)[lb]{\Black{$0$}}
    \Text(303,-29)[lb]{\Black{$0$}}
    \Text(29,5)[lb]{\Black{$z$}}
    \Text(206,5)[lb]{\Black{$z$}}
    \Text(159,4)[lb]{\Black{$1$}}
    \Text(335,4)[lb]{\Black{$1$}}
  \end{picture}
}
\end{center}
\caption{The asymptotic expansion of graphical functions at $z=0$. The bold lines stand for sets of edges.}
\label{fig:asympt}
\end{figure}
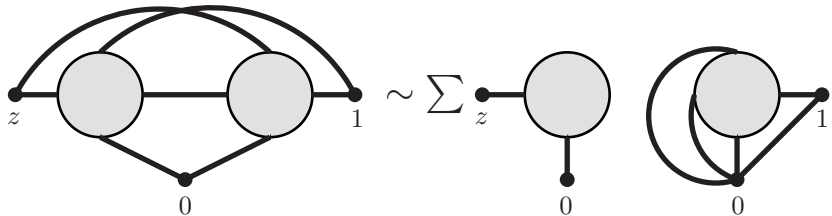

\begin{conj}\label{asymptconj}
Let $G$ be a graph with $\sV^{\mathrm{int}}$ internal and $\sV^{\mathrm{ext}}=\{0,1,z\}$ external vertices such that the graphical function $f_G^{(\lambda)}$ exists.
Let $G[V]$ be the subgraph of $G$ which is induced by $V$, i.e.\ the subgraph which contains the vertices $V$ and all edges of $G$ with both vertices in $V$.
Further let $z_0\in\{0,1\}$ and $G[V=z_0]$ be the graph $G/G[V]$ where one identifies all vertices in $V$ with the vertex $z_0$. Then (see Figure \ref{fig:asympt})
we obtain the asymptotic expansions at $z=z_0$ by
\begin{equation}\label{01asympt}
f_G^{(\lambda)}(z)=\sum_{V\subseteq\sV^{\mathrm{int}}}f_{G[V\cup\{z_0,z\}]}^{(\lambda)}(z)f_{G[V\cup\{z_0,z\}=z_0]}^{(\lambda)}(1+O(|z-z_0|^2))
\end{equation}
whenever the right hand side exists. The asymptotic expansion at $z=\infty$ is given by
\begin{equation}\label{inftyasympt}
f_G^{(\lambda)}(z)=\sum_{V\subseteq\sV^{\mathrm{int}}}f_{G[V\cup\{0,1\}]}^{(\lambda)}f_{G[V\cup\{0,1\}=0]}^{(\lambda)}(z)(1+O(|z|^{-2}))
\end{equation}
whenever the right hand side exists.
\end{conj}
Note that on the right hand side of the above equations one has graphs with two external vertices (see (\ref{2external})).
The calculation of their functions amounts to calculating periods (see Sect.\ \ref{numbers}) which is much simpler than the calculation of graphical functions.
Equations (\ref{01asympt}) and (\ref{inftyasympt}) are formally obtained by rescaling some internal variables $x_i\mapsto x_i|z|$ followed
by a naive expansion in the integrand. The sum is over all possible ways to do this. In (rare) situations the right hand sides may fail to exist
(although working in $4-\epsilon$ dimensions). In these cases we have no result for the asymptotic expansion of the left hand side.

In addition to appending edges there exists a variety of tools that allows one to calculate the $\epsilon$-expansions of graphical functions in $4-\epsilon$ dimensions to low orders
in $\epsilon$ \cite{7phi4}.

\subsection{$\beta$, $\gamma$, $\gamma_m$, and the self-energy in dimensionally regularized $\phi^4$ theory}\label{phi4}
Most efficiently one calculates $\phi^4$ renormalization functions in the minimal subtraction scheme of dimensional regularization \cite{IZ}.
The seven loop $\beta$ function, anomalous dimension $\gamma$, and anomalous mass dimension $\gamma_m$ are (The Feynman graphs were generated with M. Borinsky's program
{\tt feyngen} \cite{feyngen}.)

\begin{eqnarray}
\beta&=&\left(\frac{195654269}{23040}+\frac{15676169}{720}\zeta(3)-\frac{316009}{3840}\pi^4+\frac{18326039}{480}\zeta(5)-\frac{129631}{5040}\pi^6\right.\nonumber\\
&&+\,\frac{516957}{20}\zeta(3)^2-\frac{4453}{60}\pi^4\zeta(3)+\frac{1536173}{20}\zeta(7)-\frac{20425591}{1260000}\pi^8\nonumber\\
&&+\,116973\zeta(3)\zeta(5)+\frac{947214}{25}\zeta(5,3)-\frac{1010}{63}\pi^6\zeta(3)+\frac{613}{5}\pi^4\zeta(5)+4176\zeta(3)^3\nonumber\\
&&+\,\frac{547118}{3}\zeta(9)-\frac{45106}{43659}\pi^{10}-48\pi^4\zeta(3)^2+\frac{84231}{2}\zeta(3)\zeta(7)-\frac{273030}{7}\zeta(5)^2\nonumber\\
&&+\,\frac{8460}{7}\zeta(7,3)-\frac{174}{25}\pi^8\zeta(3)+\frac{6227}{35}\pi^6\zeta(5)-\frac{56043}{25}\pi^4\zeta(7)\nonumber\\
&&-\,504387\pi^2\zeta(9)+46845\zeta(3)^2\zeta(5)+27216\zeta(3)\zeta(5,3)-\frac{336258}{5}\zeta(5,3,3)\nonumber\\
&&\left.+\,\frac{52756839}{10}\zeta(11)+24P_{7,11}\right)g^8\nonumber\\
&&+\left(-\frac{18841427}{11520}-\frac{779603}{240}\zeta(3)+\frac{5663}{480}\pi^4-\frac{63723}{10}\zeta(5)+\frac{6691}{1890}\pi^6-\frac{8678}{5}\zeta(3)^2\right.\nonumber\\
&&+\,\frac{9}{5}\pi^4\zeta(3)-\frac{63627}{5}\zeta(7)+\frac{88181}{78750}\pi^8-4704\zeta(3)\zeta(5)-\frac{51984}{25}\zeta(5,3)-768\zeta(3)^3\nonumber\\
&&\left.-\,\frac{46112}{3}\zeta(9)\right)g^7\nonumber\\
&&+\left(\frac{764621}{2304}+\frac{7965}{16}\zeta(3)-\frac{1189}{720}\pi^4+987\zeta(5)-\frac{5}{14}\pi^6+45\zeta(3)^2+1323\zeta(7)\right)g^6\nonumber\\
&&+\left(-\frac{3499}{48}-78\zeta(3)+\frac{1}{5}\pi^4-120\zeta(5)\right)g^5+\left(\frac{145}{8}+12\zeta(3)\right)g^4-\frac{17}{3}g^3+3g^2\\
&\approx&474651g^8-34776.1g^7+2848.57g^6-271.606g^5+32.5497g^4-5.66667g^3+3g^2,\nonumber
\end{eqnarray}
where $P_{7,11}$ is given in (\ref{711}). This confirms (and goes beyond) a recent six loop result by M.V. Kompaniets and E. Panzer \cite{KP6loopbeta}.

\begin{eqnarray}
\gamma&=&\left(-\frac{214519}{5120}-\frac{52883}{1920}\zeta(3)-\frac{4247}{23040}\pi^4+\frac{8023}{320}\zeta(5)-\frac{71}{1080}\pi^6-\frac{523}{40}\zeta(3)^2\right.\nonumber\\
&&-\,\left.\frac{1}{20}\pi^4\zeta(3)+\frac{3573}{40}\zeta(7)-\frac{2063}{210000}\pi^8+27\zeta(3)\zeta(5)+\frac{162}{25}\zeta(5,3)\right)g^7\nonumber\\
&&+\left(\frac{73667}{9216}+\frac{295}{192}\zeta(3)+\frac{73}{1920}\pi^4-\frac{37}{8}\zeta(5)+\frac{5}{756}\pi^6-\frac{1}{2}\zeta(3)^2\right)g^6\nonumber\\
&&+\left(-\frac{3709}{2304}+\frac{3}{16}\zeta(3)-\frac{1}{180}\pi^4\right)g^5+\frac{65}{192}g^4-\frac{1}{16}g^3+\frac{1}{12}g^2\\
&\approx&-124.159g^7+14.3840g^6-1.92558g^5+0.338542g^4-0.0625g^3+0.0833333g^2.\nonumber
\end{eqnarray}
This confirms (and goes beyond) a recent six loop calculation of D.V. Batkovich, K.G. Chetyrkin, and M.V. Kompaniets \cite{6loopgamma}.

\begin{eqnarray}
\gamma_m&=&\left(-\frac{24838423}{13824}-\frac{2399489}{864}\zeta(3)+\frac{329}{960}\pi^4-\frac{25511}{24}\zeta(5)-\frac{1865}{1134}\pi^6\right.\nonumber\\
&&-\,\frac{140153}{48}\zeta(3)^2-\frac{68}{45}\pi^4\zeta(3)+\frac{46625}{12}\zeta(7)-\frac{83003}{378000}\pi^8-4519\zeta(3)\zeta(5)\nonumber\\
&&-\,\frac{6147}{5}\zeta(5,3)-\frac{412}{189}\pi^6\zeta(3)+\frac{167}{30}\pi^4\zeta(5)+424\zeta(3)^3+\frac{60289}{12}\zeta(9)\nonumber\\
&&\left.-\,\frac{45106}{654885}\pi^{10}-\frac{16}{5}\pi^4\zeta(3)^2+\frac{777}{2}\zeta(3)\zeta(7)+\frac{31778}{7}\zeta(5)^2+\frac{564}{7}\zeta(7,3)\right)g^7\nonumber\\
&&+\left(\frac{7915913}{23040}+\frac{472891}{1440}\zeta(3)+\frac{113}{192}\pi^4+\frac{4019}{40}\zeta(5)+\frac{163}{540}\pi^6+\frac{446}{5}\zeta(3)^2\right.\nonumber\\
&&\left.+\,\frac{3}{5}\pi^4\zeta(3)-\frac{4629}{20}\zeta(7)+\frac{2063}{35000}\pi^8-288\zeta(3)\zeta(5)-\frac{972}{25}\zeta(5,3)\right)g^6\nonumber\\
&&+\left(-\frac{158849}{2304}-\frac{1519}{48}\zeta(3)-\frac{13}{72}\pi^4-\zeta(5)-\frac{5}{126}\pi^6+9\zeta(3)^2\right)g^5\nonumber\\
&&+\left(\frac{477}{32}+\frac{3}{2}\zeta(3)+\frac{1}{30}\pi^4\right)g^4-\frac{7}{2}g^3+\frac{5}{6}g^2-g\nonumber\\
&\approx&-13759.8g^7+1354.64g^6-150.756g^5+19.9563g^4-3.5g^3+0.833333g^2-g.
\end{eqnarray}
This confirms (and goes beyond) a recent six loop result by M.V. Kompaniets and E. Panzer \cite{KP6loopbeta}.

The six loop self-energy $\Sigma$ is
\begin{eqnarray}
\frac{\Sigma(p)}{p^2}&=&\left[-\frac{27}{2}L^5-\frac{3643}{24}L^4+\left(-\frac{648011}{864}-16\zeta(3)\right)L^3\right.\nonumber\\
&&+\,\left(-\frac{291187}{144}-82\zeta(3)-\frac{1}{20}\pi^4-20\zeta(5)\right)L^2\nonumber\\
&&+\,\left(-\frac{1699885}{576}-\frac{32953}{192}\zeta(3)-\frac{11}{48}\pi^4-\frac{211}{4}\zeta(5)-\frac{5}{378}\pi^6+\zeta(3)^2\right)L\nonumber\\
&&-\,\frac{33992153}{18432}-\frac{683389}{4608}\zeta(3)-\frac{18403}{69120}\pi^4-\frac{8681}{192}\zeta(5)\nonumber\\
&&-\,\left.\frac{359}{18144}\pi^6-\frac{83}{48}\zeta(3)^2+\frac{1}{360}\pi^4\zeta(3)+\frac{5}{6}\zeta(7)\right]g^6\nonumber\\
&&+\,\left[\frac{9}{2}L^4+\frac{1375}{36}L^3+\left(\frac{12935}{96}+2\zeta(3)\right)L^2+\left(\frac{8353}{36}+5\zeta(3)+\frac{1}{90}\pi^4\right)L\right.\nonumber\\
&&+\,\left.\frac{1874629}{11520}+\frac{6319}{1440}\zeta(3)+\frac{251}{14400}\pi^4-\frac{1}{5}\zeta(5)\right]g^5\nonumber\\
&&+\,\left[-\frac{3}{2}L^3-\frac{53}{6}L^2-\frac{1867}{96}L-\frac{2017}{128}+\frac{3}{32}\zeta(3)\right]g^4\nonumber\\
&&+\,\left[\frac{1}{2}L^2+\frac{7}{4}L+\frac{167}{96}\right]g^3+\left[-\frac{1}{6}L-\frac{13}{48}\right]g^2,
\end{eqnarray}
where
$$
L=\frac{1}{2}\log\left(\frac{4\pi\Lambda^2}{\exp(C)p^2}\right),
$$
with the renormalization scale $\Lambda$ and Euler-Mascheroni constant $C=0.577\ldots$.
This result confirms an unpublished five loop result by D. Broadhurst for the propagator $1/(p^2-\Sigma(p))$ at $L=0$ \cite{5loopB}.

The analogous results for the $O(n)$-symmetric model have also been calculated. They are available in {\tt HyperlogProcedures} \cite{Hyperlogproc}.
With these results E. Panzer improved his QFT predictions for critical exponents in three dimensional statistical models (private communication, \cite{KP6loopbeta}).

While eight loop calculations of the anomalous dimensions $\gamma$ and $\gamma_m$ seem possible, an eight loop result for the $\beta$ function demands
serious determination. For the time being the author is not pursuing eight loop calculations.

\subsection{The anomalous magnetic moment of the electron}\label{g2}
So far, all results were obtained in the framework of massless bosonic $\phi^4$ theory. How does the picture change for a physical gauge theory with massive fermions?
An excellent test is the QED contribution to the anomalous magnetic moment of the electron $a_e$ where three orders in $\alpha/\pi$ are known \cite{L1}
with a recent partial fourth order result by S. Laporta \cite{Laporta}.

In the $f$ alphabet for MZVs with extensions by all sixth roots of unity we obtain---we use the letters $g^6$ to make the distinction to (\ref{711}) which refers to
a number subset which has no weight 1 letters $g^6_1\cong2\log2$ and $\log3$ (the letter $\log3$ is absent in all known terms of $a_e$):
\begin{eqnarray}\label{ae}
a_e&=&\frac{1}{2}\left(\frac{\alpha}{\pi}\right)+\left(\frac{197}{144}+\frac{1}{12}\pi^2+\frac{27}{32}g^6_3-\frac{1}{4}g^6_1\pi^2\right)\left(\frac{\alpha}{\pi}\right)^2\nonumber\\
&&+\,\left(\frac{28259}{5184}+\frac{17101}{810}\pi^2+\frac{139}{16}g^6_3-\frac{149}{9}g^6_1\pi^2-\frac{525}{32}g^6_1g^6_3+\frac{1969}{8640}\pi^4-\frac{1161}{128}g^6_5\right.\nonumber\\
&&\qquad+\,\left.\frac{83}{64}g^6_3\pi^2\right)\left(\frac{\alpha}{\pi}\right)^3.
\end{eqnarray}
In the $f$ alphabet the Galois coaction (\ref{co}) is deconcatenation. It is therefore easy to read off the Galois conjugates on the right hand side of the tensor product
in the coaction. Up to weight three we only have the four Galois conjugates
$$1,\quad\pi^2,\quad g^6_3,\quad g^6_1\pi^2.$$
Although this list follows from a three loop result it can be conjectured that the list is complete to all loop orders (see e.g.\ \cite{Bcoact1, Bcoact2}).
In general we expect that one is able to extract the complete list of Galois conjugates of weight $\leq n$ from an $n$ loop result.

In \cite{Laporta} Laporta presents an explicit result for the hyperlogarithmic part of the fourth order $a_e$. The conversion into the $f$ alphabet is given in \cite{motg2}.
It is similar to (\ref{ae}) with additional extensions of MZVs by fourth roots of unity. With $f^4_2\cong 2\ii\,$Im$\,\Li_2(\ii)$ we obtain the following six Galois conjugates of weight 4:
$$g^6_4,\quad g^6_1g^6_3,\quad g^6_2\pi^2,\quad f^4_2\pi^2,\quad g^6_1g^6_1\pi^2,\quad \pi^4.$$
The result is preliminary because it does not contain the Galois conjugates of the non-hyperlogarithmic part of $a_e$. It is also possible that Laporta's result misses some
hyperlogarithmic terms (hidden in the non-hyperlogarithmic part). Because the coaction conjectures basically work graph by graph, it still makes sense to analyze the partial
result \cite{motg2}.

The two most remarkable properties of the motivic structure of $a_e$ are
\begin{enumerate}
\item The $\QQ$ vector spaces of Galois conjugates at given weight have very low dimensions. This is a strong sparsity property of QED, similar to the one found for $\sP_{\phi^4}$.
\item The type of numbers in $a_e$ corresponds to what we found in $\sP_{\phi^4}$. We have MZVs ($g^6_3,g^6_5$), Euler sums ($g^6_1,g^6_1g^6_1,g^6_1g^6_3$), extensions by sixth roots
of unity ($g^6_2,g^6_4$), and extensions by fourth roots of unity ($f^4_2$). These correspond to the $c_2$ invariants $-1$, $-z_2$, $-z_3$, $-z_4$ found in $\phi^4$ theory.
There seems to exist nothing else in the polylogarithmic part of $a_e$ up to loop order four. The only difference to massless $\phi^4$ theory is that the numbers come at smaller
loop orders, namely 2,2,4,4 for $a_e$ in contrast to 3,9,7,8 for $\phi^4$.
\end{enumerate}

\newpage
\begin{tabular}{llllllll}
name&graph&numerical value&$|$Aut$|$&index&anc.&$-c_2$&remarks, [Lit]\\[1ex]
\multicolumn{2}{l}{weight}&\multicolumn{6}{l}{exact value}\\[1ex]\hline\hline
$P_1$&\hspace*{-2mm}\raisebox{-9mm}{\includegraphics[width=12mm]{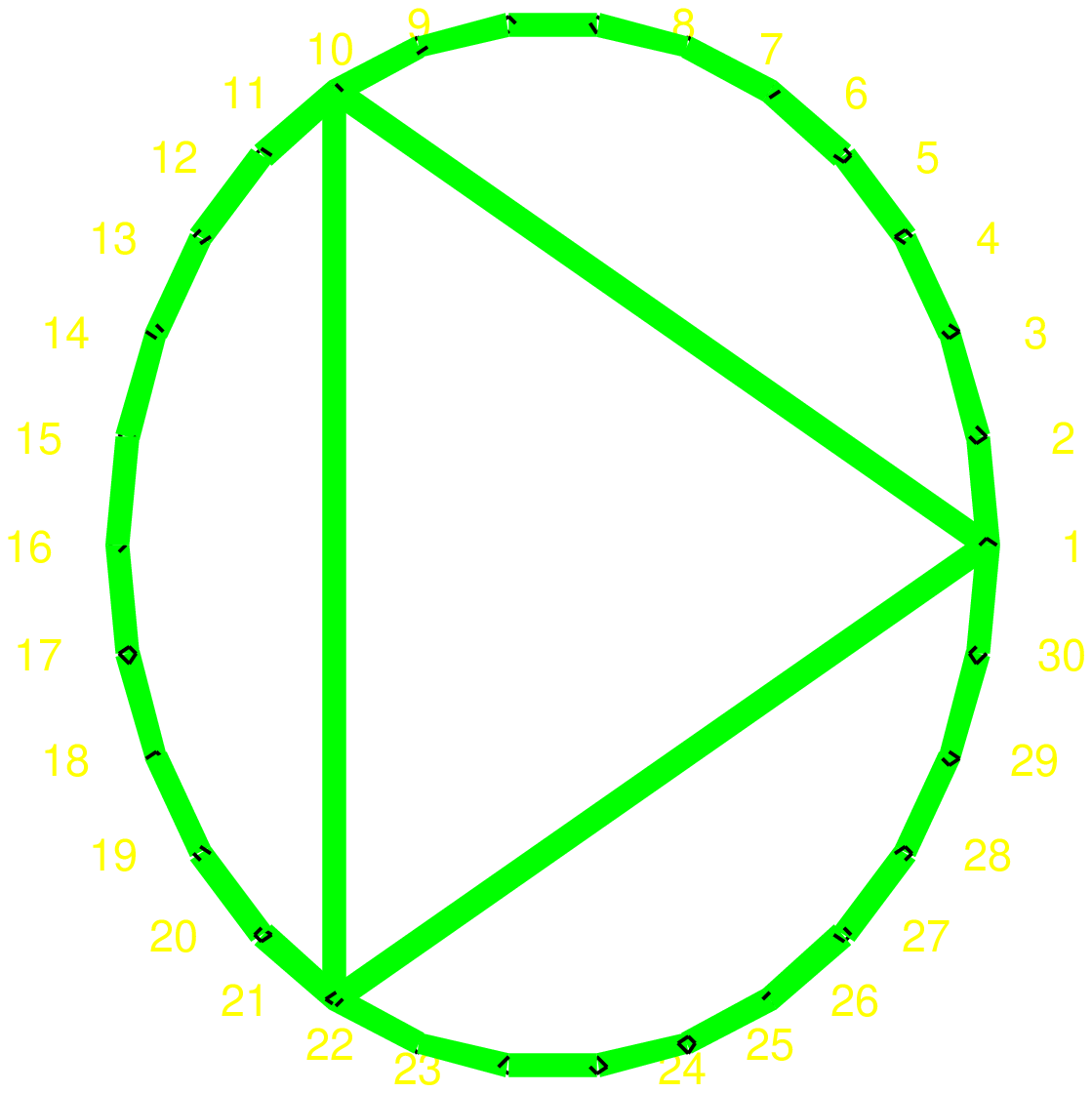}}&1&48&---&$P_1$&---&$C^3_{1,1}$\\[-6mm]
0&&\multicolumn{6}{l}{1}\\[1ex]\hline\hline
$P_3$&\hspace*{-2mm}\raisebox{-9mm}{\includegraphics[width=12mm]{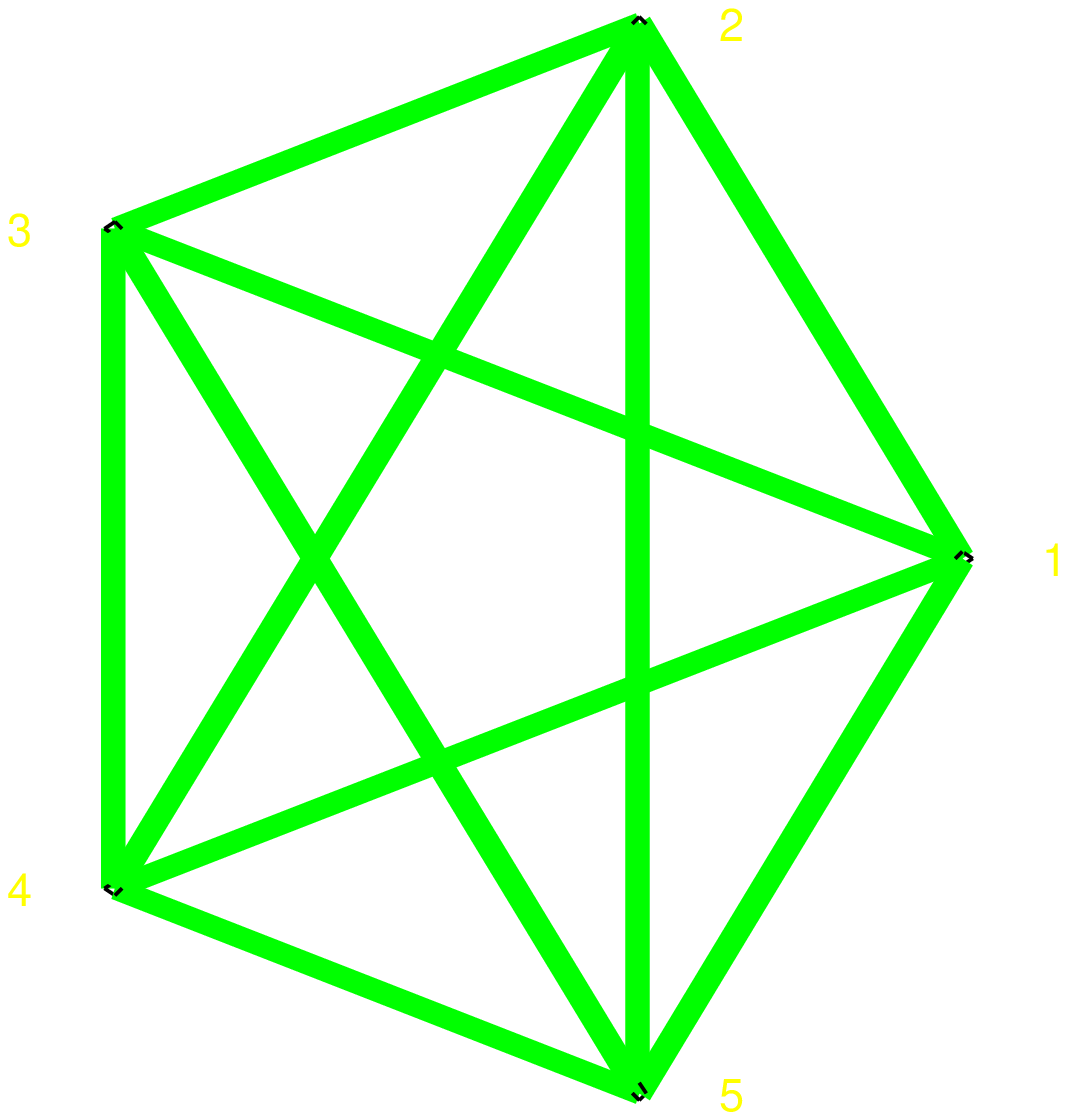}}&7.212~341~418&120&6&$P_3$&1&$C^5_{1,2}$, $K_5$, \cite{C4}\\[-6mm]
3&&\multicolumn{6}{l}{$6Q_3$}\\[1ex]\hline\hline
$P_4$&\hspace*{-2mm}\raisebox{-9mm}{\includegraphics[width=12mm]{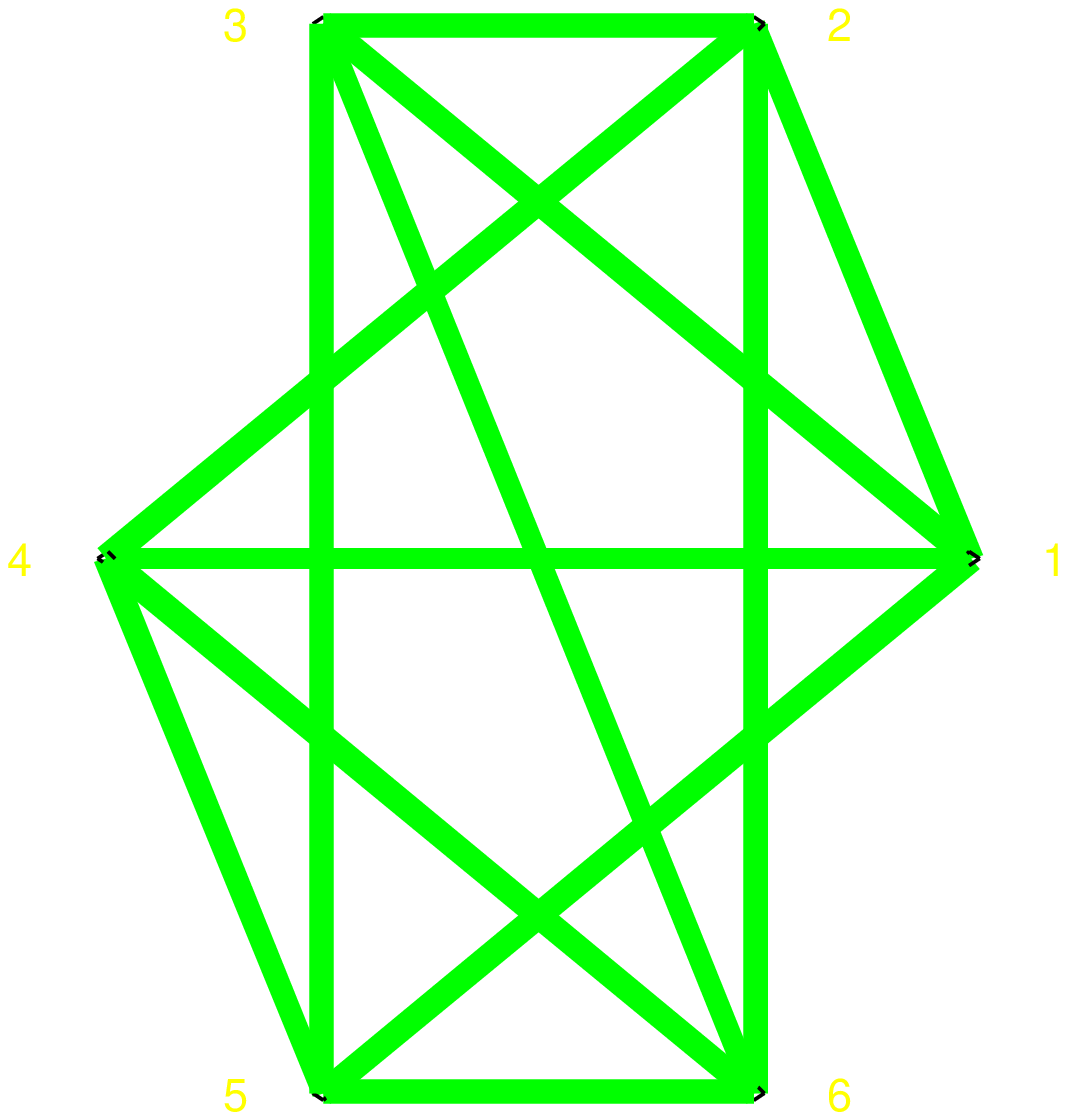}}&20.738~555~102&48&40&$P_3$&1&$C^6_{1,2}$, $O_3$, \cite{C5}\\[-6mm]
5&&\multicolumn{6}{l}{$20Q_5$}\\[1ex]\hline\hline
$P_5$&\hspace*{-2mm}\raisebox{-9mm}{\includegraphics[width=12mm]{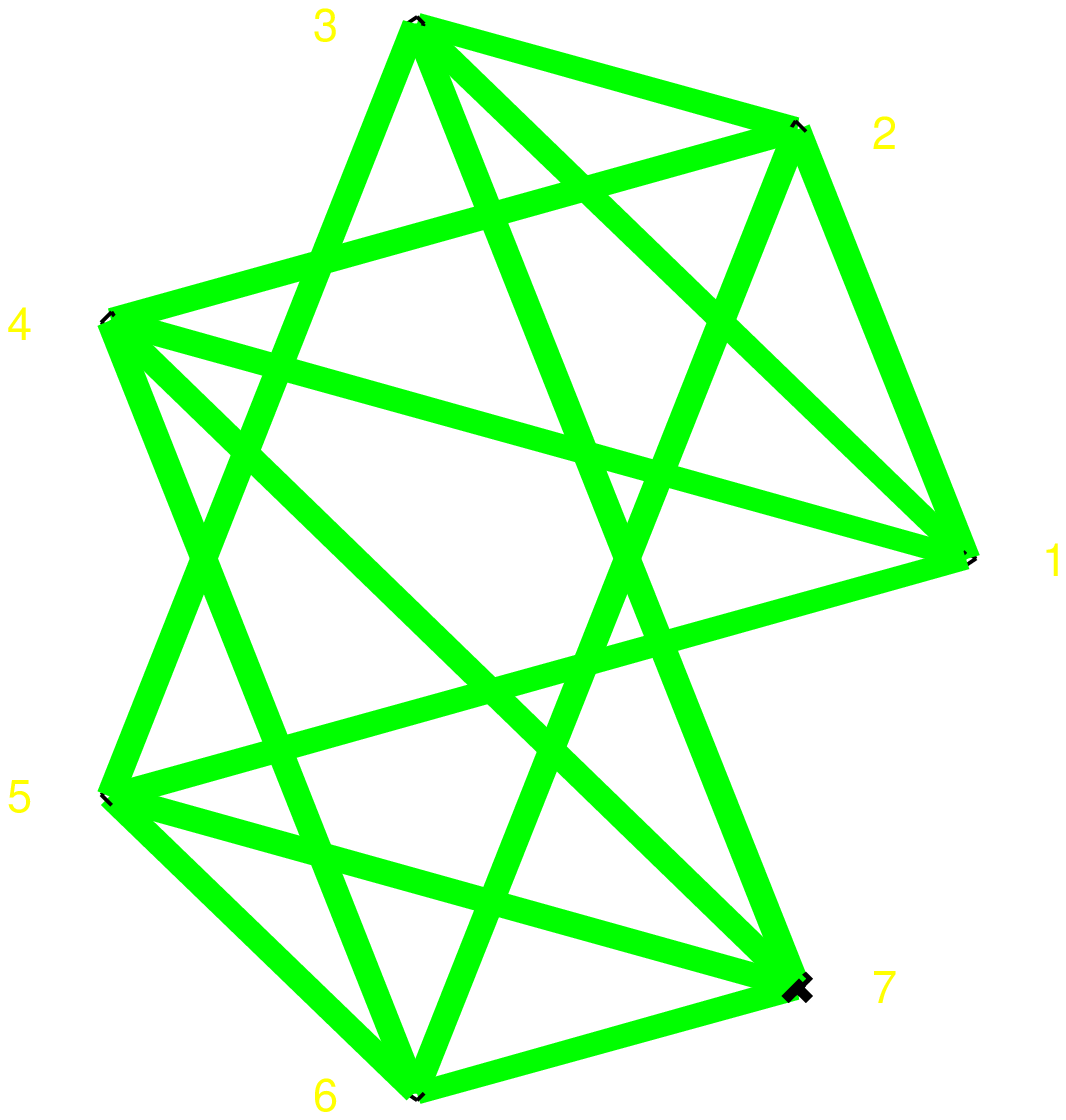}}&55.585~253~915&14&882&$P_3$&1&$C^7_{1,2}$, $\overline{C_7}$, \cite{K2}\\[-6mm]
7&&\multicolumn{6}{l}{$\frac{441}{8}Q_7$}\\[1ex]\hline\hline
$P_{6,1}$&\hspace*{-2mm}\raisebox{-9mm}{\includegraphics[width=12mm]{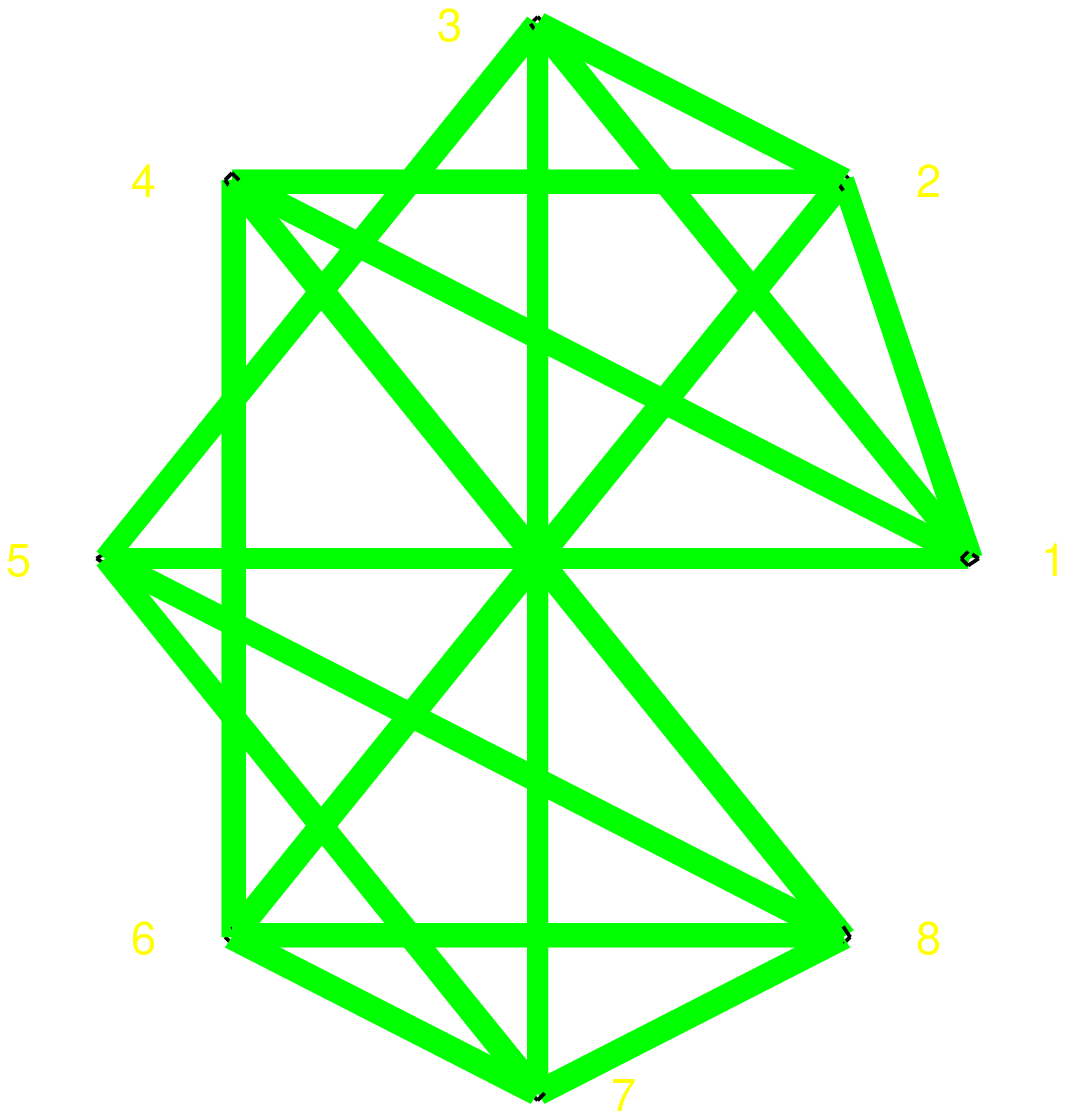}}&168.337~409~994&16&24192&$P_3$&1&$C^8_{1,2}$ \cite{U1}\\[-6mm]
9&&\multicolumn{6}{l}{$168Q_9$}\\[1ex]\hline
$P_{6,2}$&\hspace*{-2mm}\raisebox{-9mm}{\includegraphics[width=12mm]{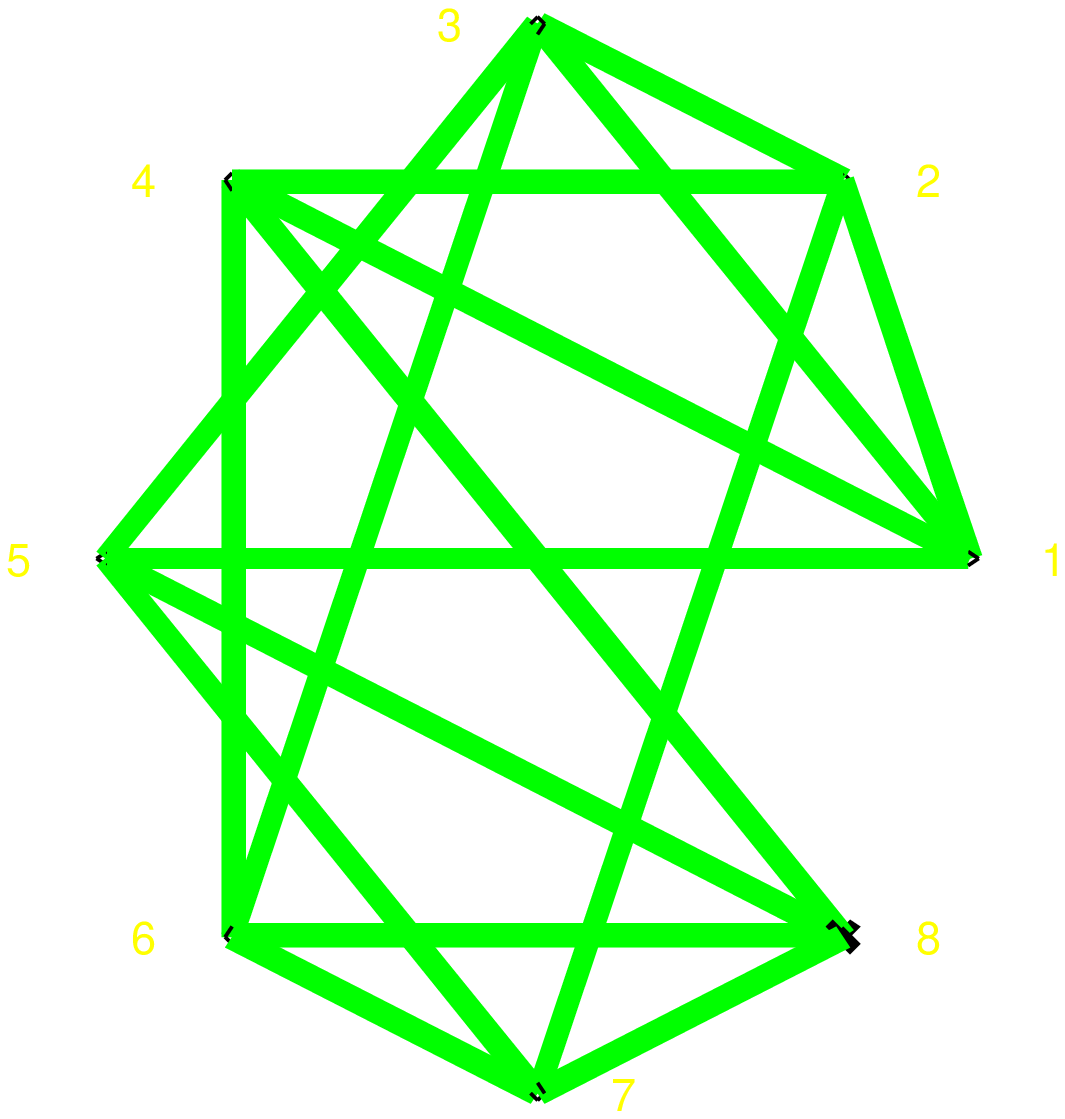}}&132.243~533~110&4&16&$P_3$&1&\cite{BK}\\[-6mm]
9&&\multicolumn{6}{l}{$\frac{1063}{9}Q_9+8Q_3^3$}\\[1ex]\hline
$P_{6,3}$&\hspace*{-2mm}\raisebox{-9mm}{\includegraphics[width=12mm]{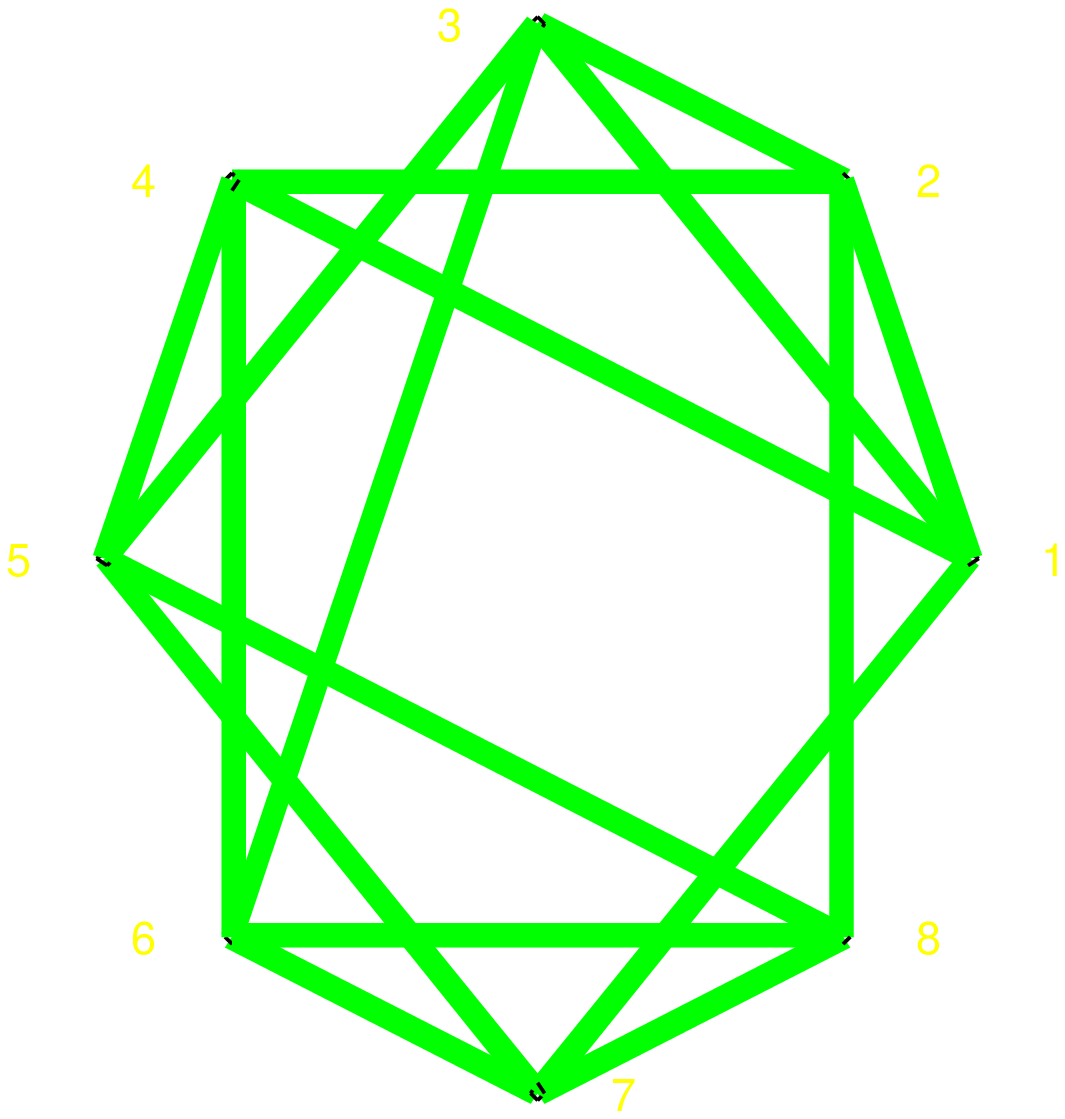}}&107.711~024~841&16&72&$P_3^2$&0&\cite{BK}\\[-6mm]
8&&\multicolumn{6}{l}{$256Q_8+72Q_3Q_5$}\\[1ex]\hline
$P_{6,4}$&\hspace*{-2mm}\raisebox{-9mm}{\includegraphics[width=12mm]{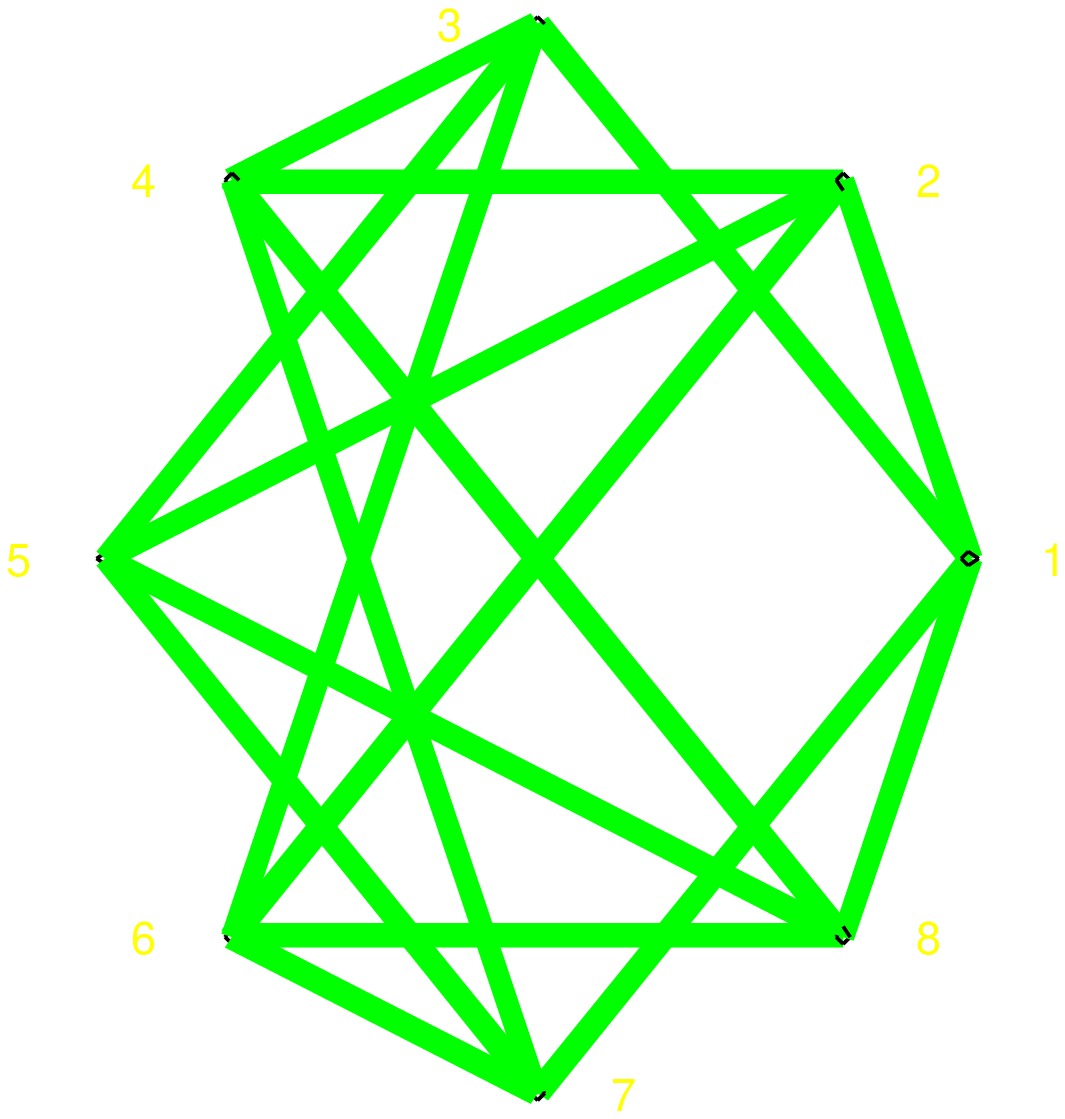}}&71.506~081~796&1152&1728&$P_{6,4}$&0&$C^8_{1,3}$ \cite{BK}, \cite{S3}\\[-6mm]
8&&\multicolumn{6}{l}{$-4096Q_8+288Q_3Q_5$}\\[1ex]\hline\hline
$P_{7,1}$&\hspace*{-2mm}\raisebox{-9mm}{\includegraphics[width=12mm]{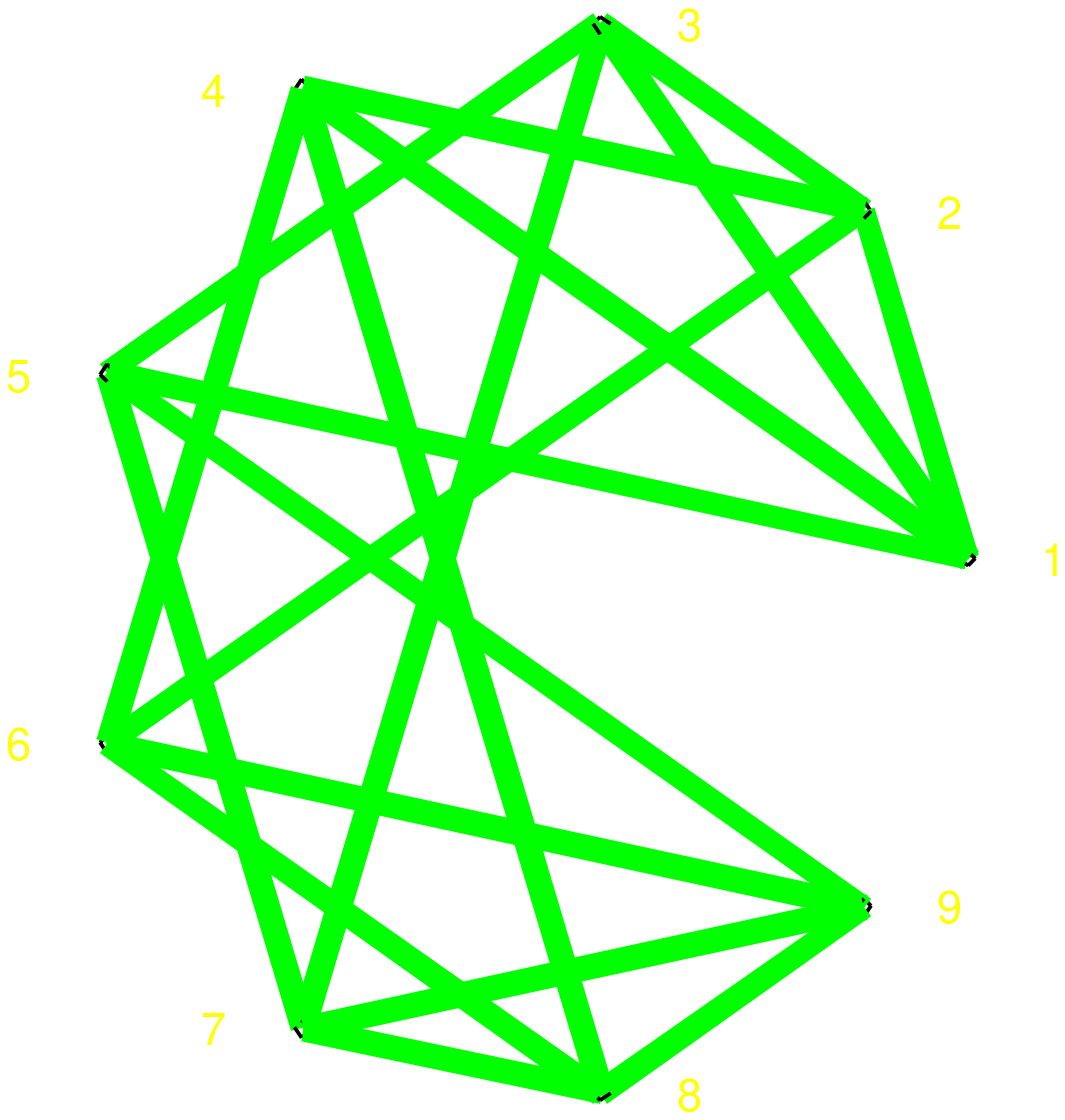}}&527.745~051~766&18&405108&$P_3$&1&$C^9_{1,2}$ \cite{BK}\\[-6mm]
11&&\multicolumn{6}{l}{$\frac{33759}{64}Q_{11,1}$}\\[1ex]\hline
$P_{7,2}$&\hspace*{-2mm}\raisebox{-9mm}{\includegraphics[width=12mm]{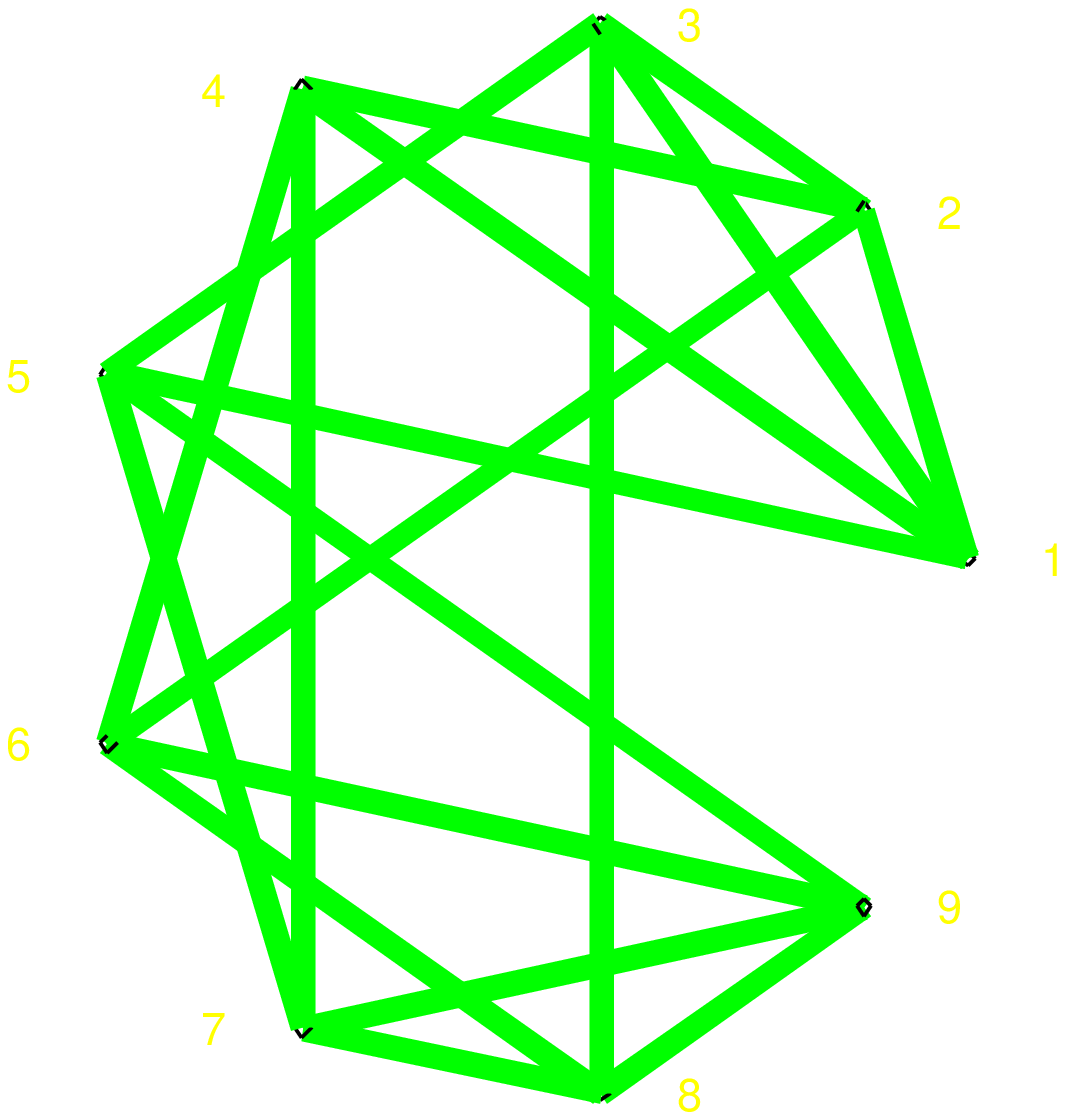}}&380.887~829~534&2&20&$P_3$&1&\cite{BK}\\[-6mm]
11&&\multicolumn{6}{l}{$\frac{62957}{192}Q_{11,1}+9Q_{11,2}+35Q_3^2Q_5$}\\[1ex]\hline
$P_{7,3}$&\hspace*{-2mm}\raisebox{-9mm}{\includegraphics[width=12mm]{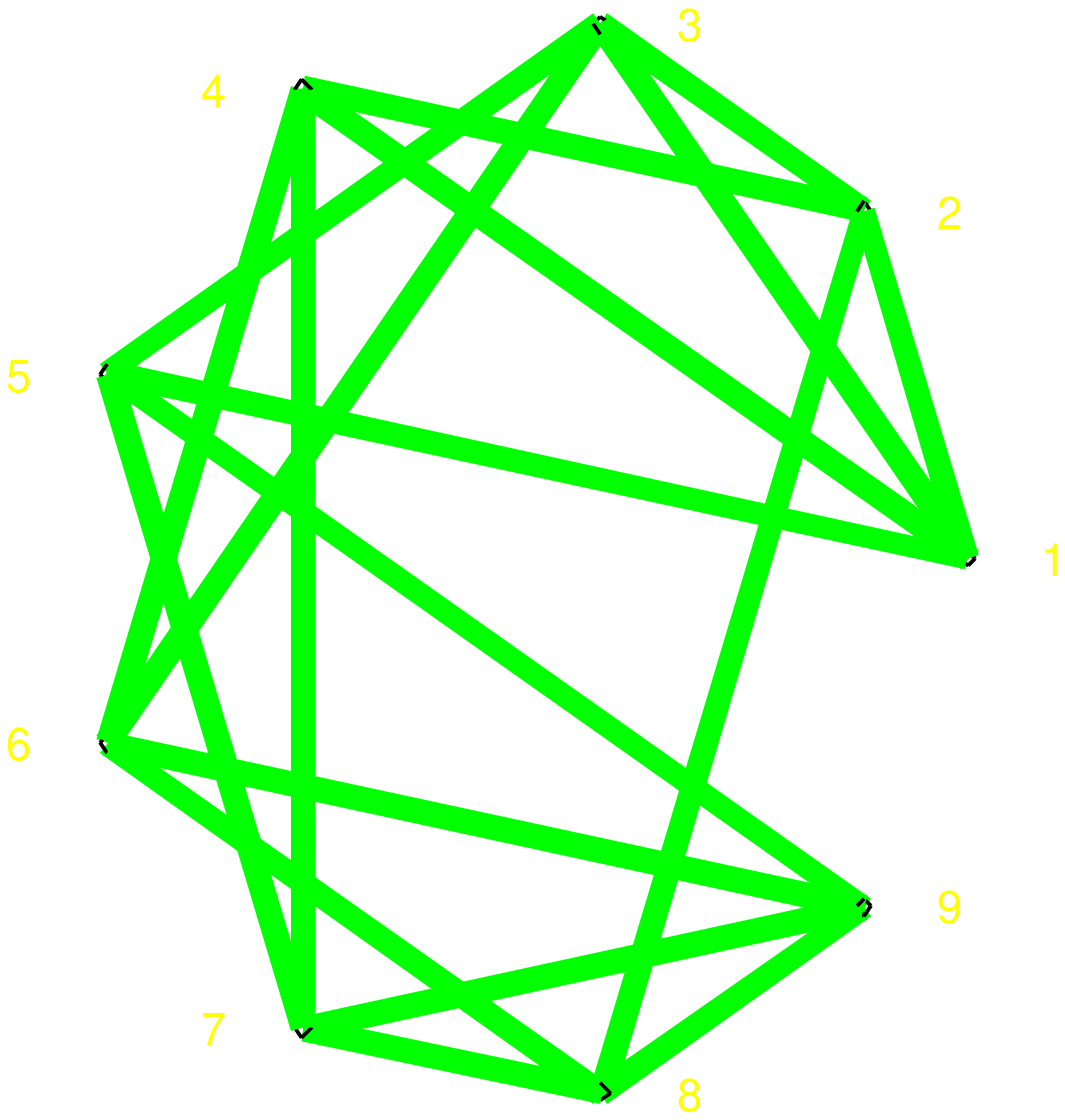}}&336.067~072~110&2&16&$P_3$&1&\cite{BK}\\[-6mm]
11&&\multicolumn{6}{l}{$\frac{73133}{240}Q_{11,1}+\frac{144}{5}Q_{11,2}+20Q_3^2Q_5$}\\[1ex]\hline
$P_{7,4}$&\hspace*{-2mm}\raisebox{-9mm}{\includegraphics[width=12mm]{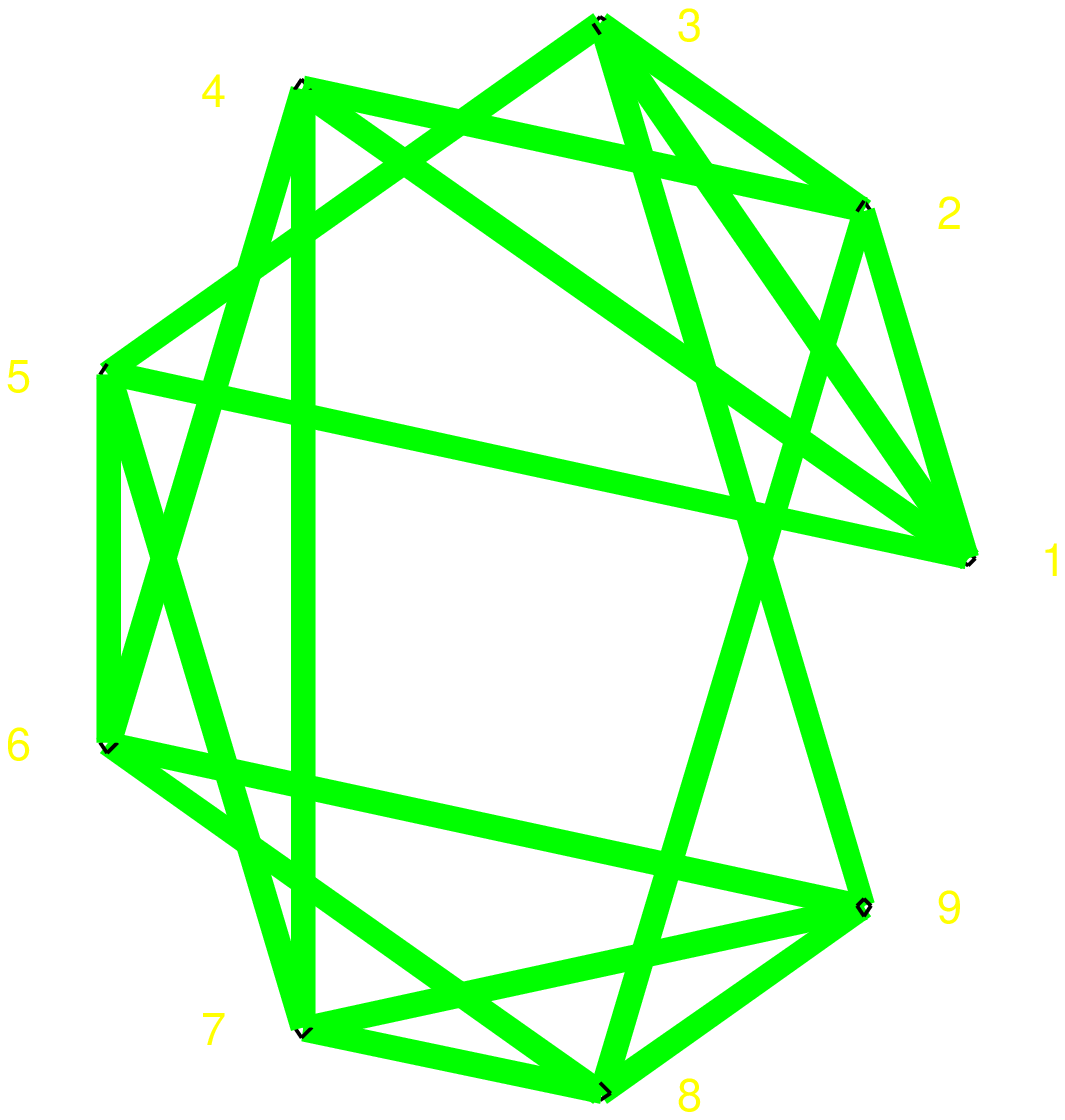}}&294.035~314~185&4&320&$P_3^2$&0&\cite{BK}\\[-6mm]
10&&\multicolumn{6}{l}{$420Q_3Q_7-200Q_5^2$}\\[1ex]\hline
$P_{7,5}$&\hspace*{-2mm}\raisebox{-9mm}{\includegraphics[width=12mm]{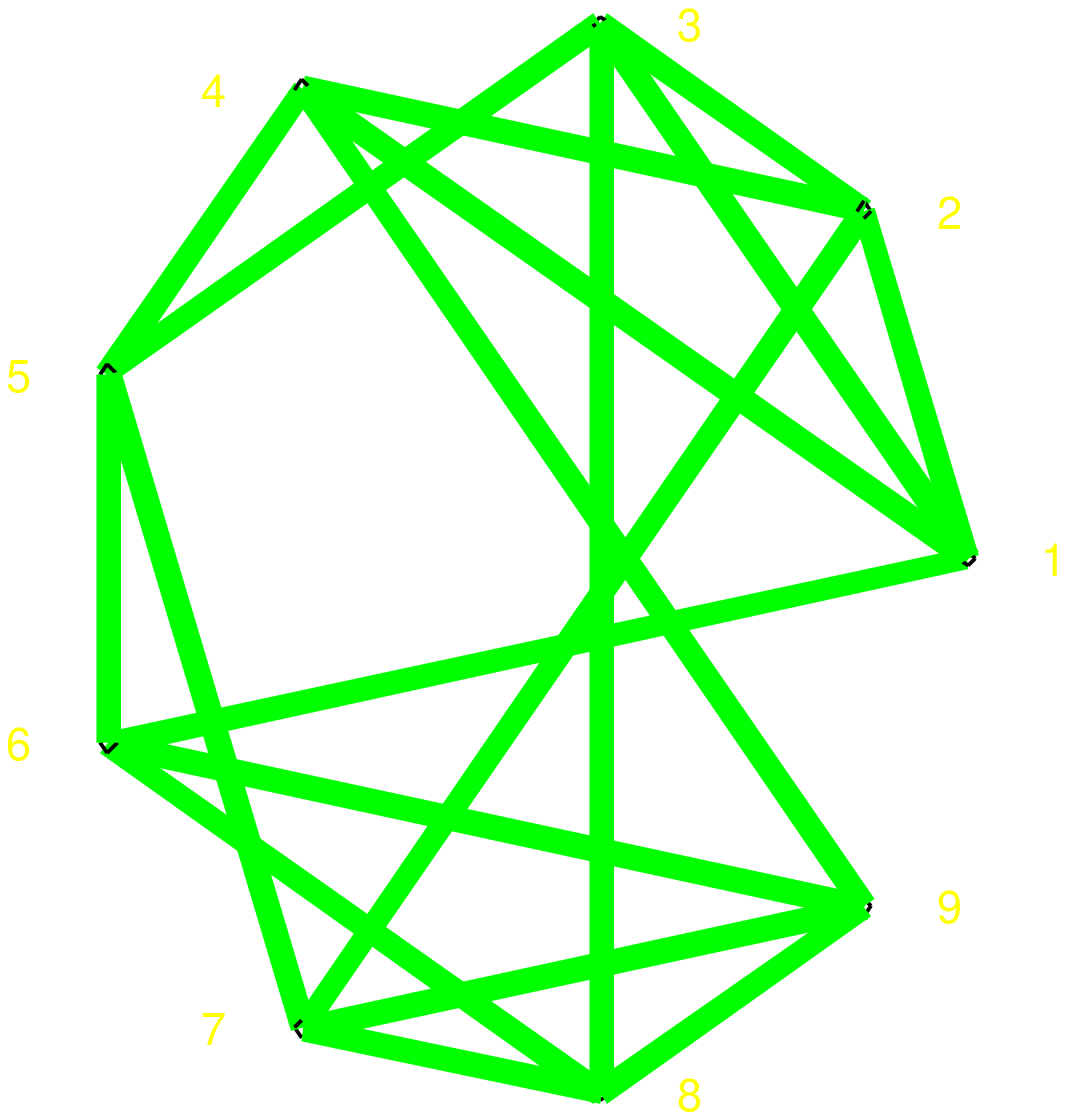}}&254.763~009~595&8&144&$P_3^2$&0&\cite{BK}\\[-6mm]
10&&\multicolumn{6}{l}{$-189Q_3Q_7+450Q_5^2$}\\[1ex]\hline
$P_{7,6}$&\hspace*{-2mm}\raisebox{-9mm}{\includegraphics[width=12mm]{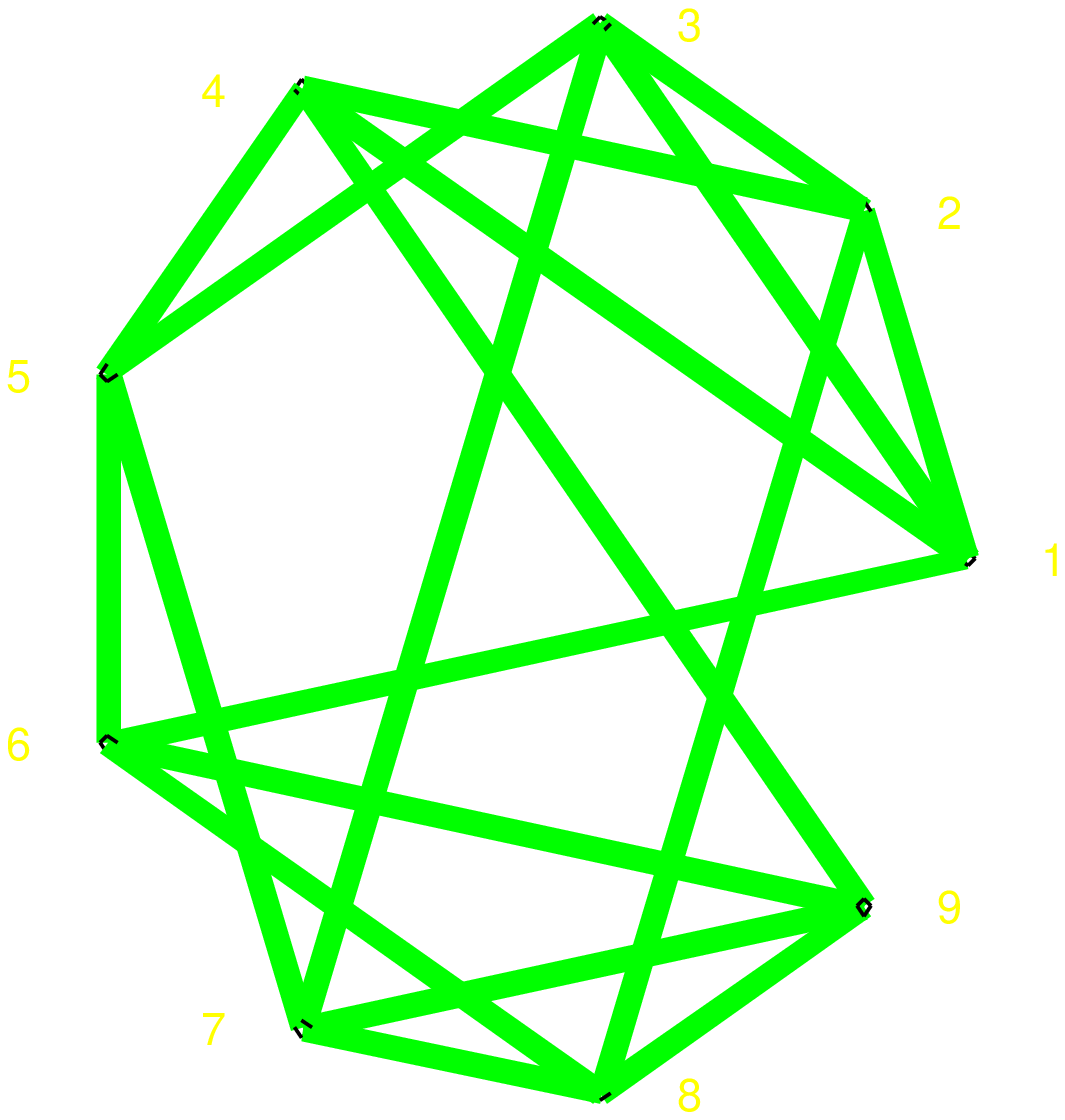}}&273.482~574~258&2&20&$P_3$&1&\cite{BK}\\[-6mm]
11&&\multicolumn{6}{l}{$\frac{14279}{64}Q_{11,1}-51Q_{11,2}+35Q_3^2Q_5$}\\[1ex]\hline
$P_{7,7}$&\hspace*{-2mm}\raisebox{-9mm}{\includegraphics[width=12mm]{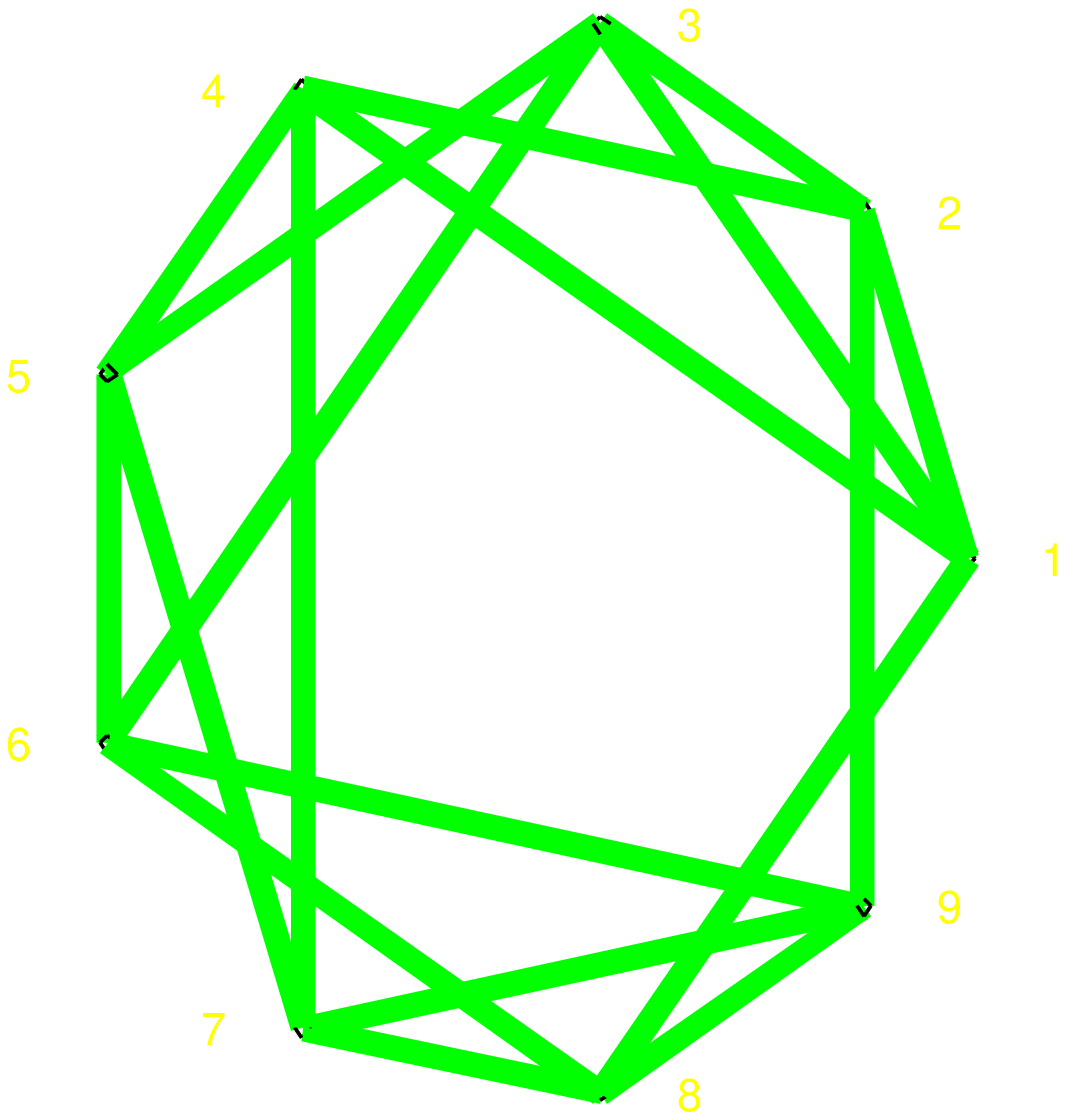}}&294.035~314~185&8&320&$P_3^2$&0&Fourier, twist\\[-6mm]
10&&\multicolumn{6}{l}{$P_{7,4}$\hspace*{91mm}}
\end{tabular}

\begin{tabular}{llllllll}
name&graph&numerical value&$|$Aut$|$&index&anc.&$-c_2$&remarks, [Lit]\\[1ex]
\multicolumn{2}{l}{weight}&\multicolumn{6}{l}{exact value}\\[1ex]\hline\hline
$P_{7,8}$&\hspace*{-2mm}\raisebox{-9mm}{\includegraphics[width=12mm]{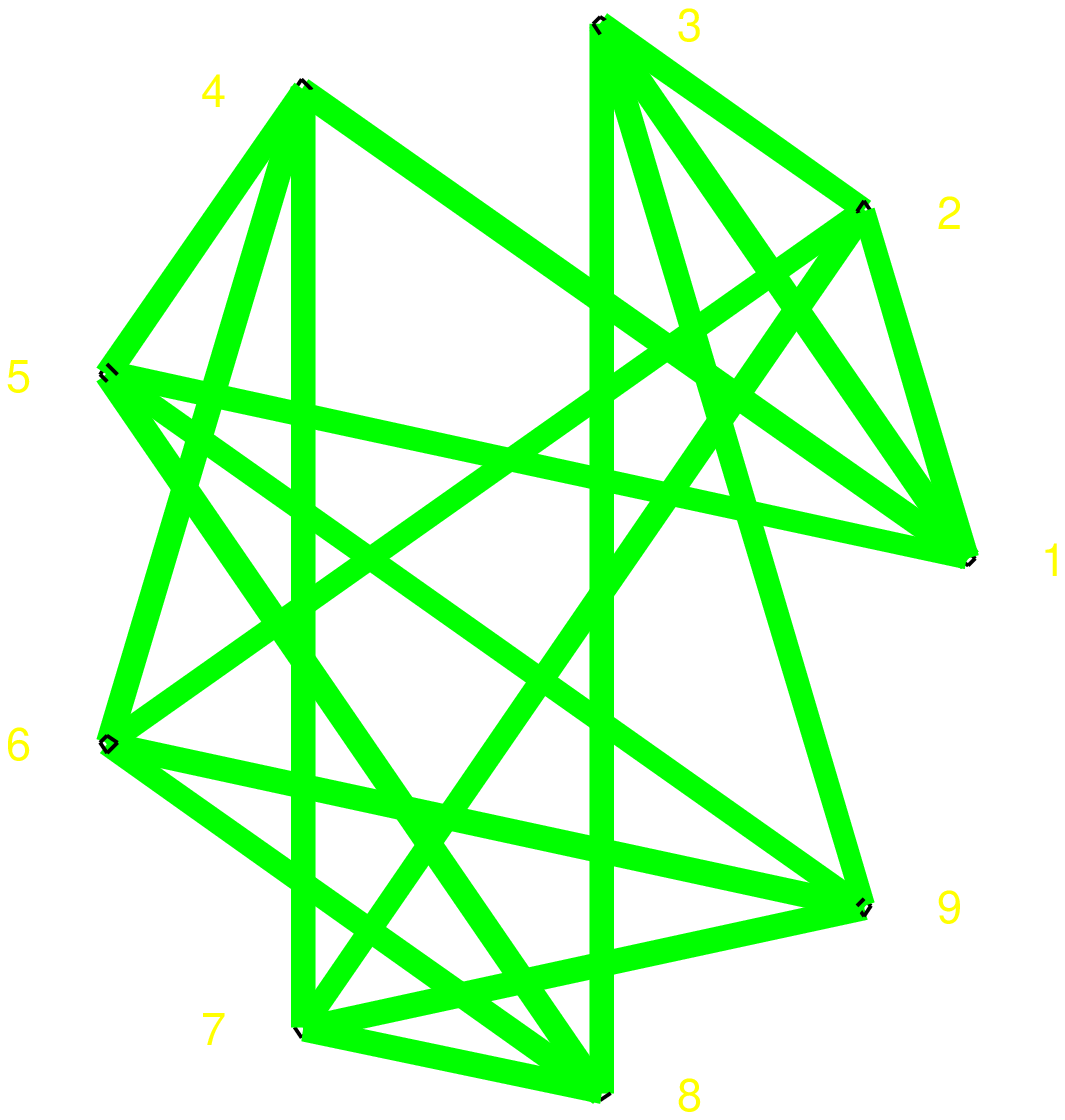}}&183.032~420~030&16&16&$P_{7,8}$&$z_2$&\\[-6mm]
11&&\multicolumn{6}{l}{$\frac{22383}{20}Q_{11,1}-\frac{4572}{5}Q_{11,2}+1792Q_3Q_8-700Q_3^2Q_5$}\\[1ex]\hline
$P_{7,9}$&\hspace*{-2mm}\raisebox{-9mm}{\includegraphics[width=12mm]{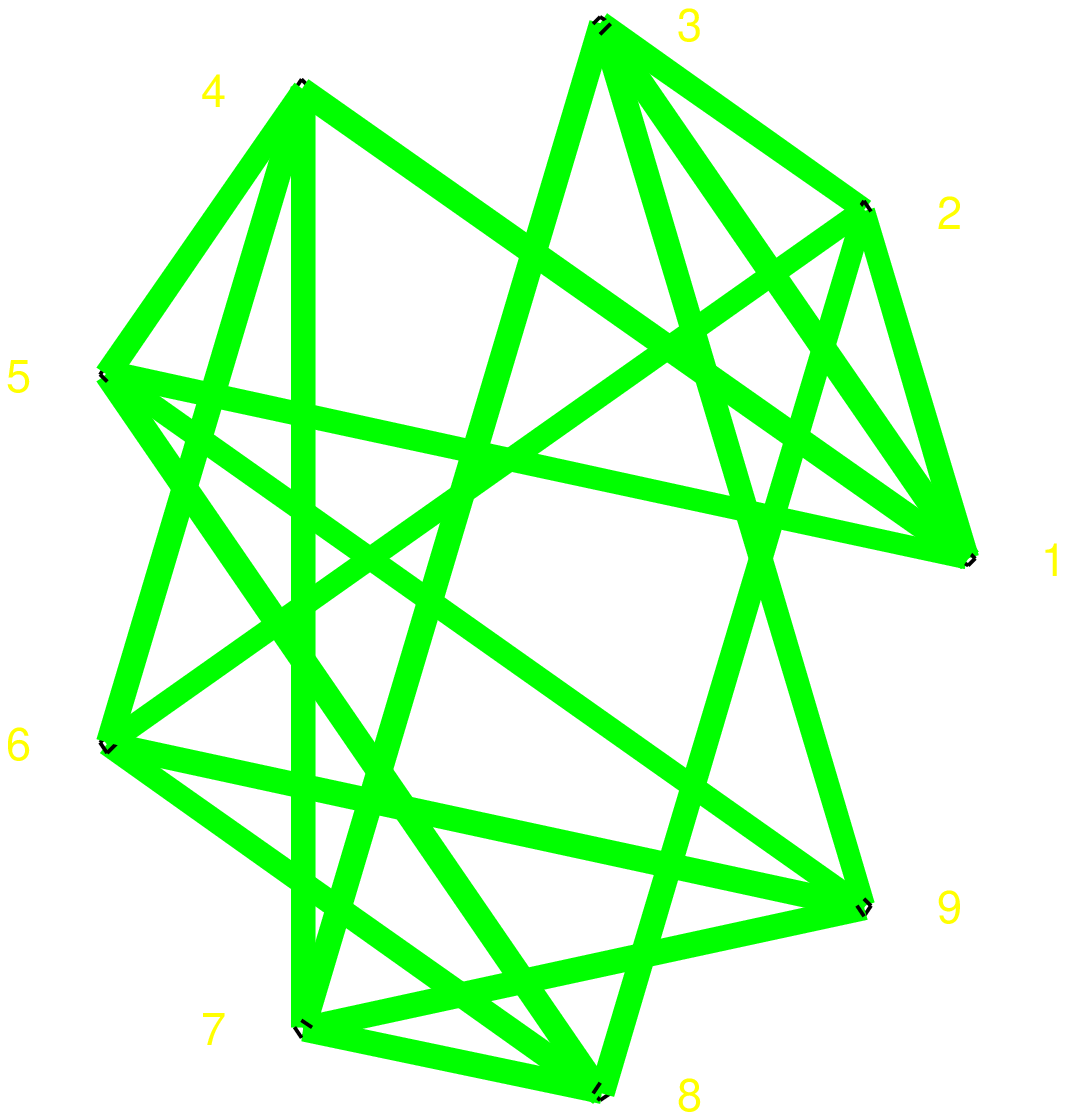}}&216.919~375~587&12&3&$P_{7,9}$&$z_2$&\cite{B31}\\[-6mm]
11&&\multicolumn{6}{l}{$\frac{92943}{160}Q_{11,1}-\frac{3381}{20}Q_{11,2}+896Q_3Q_8-\frac{1155}{4}Q_3^2Q_5$}\\[1ex]\hline
$P_{7,10}$&\hspace*{-2mm}\raisebox{-9mm}{\includegraphics[width=12mm]{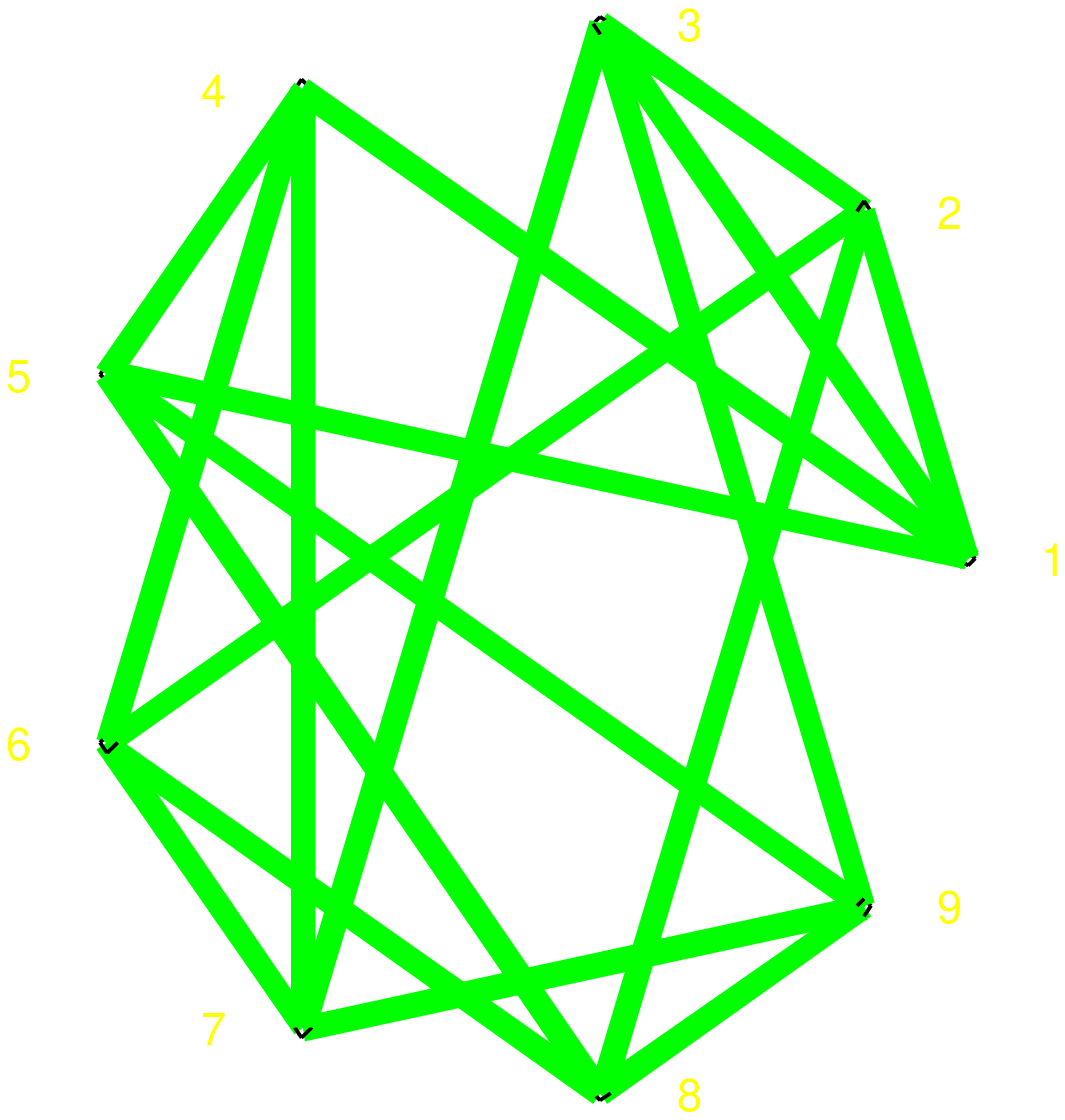}}&254.763~009~595&72&144&$P_{7,10}$&0&$K_3$\raisebox{.7ex}{\fbox{}}$K_3$, Fourier\\[-6mm]
10&&\multicolumn{6}{l}{$P_{7,5}$}\\[1ex]\hline
$P_{7,11}$&\hspace*{-2mm}\raisebox{-9mm}{\includegraphics[width=12mm]{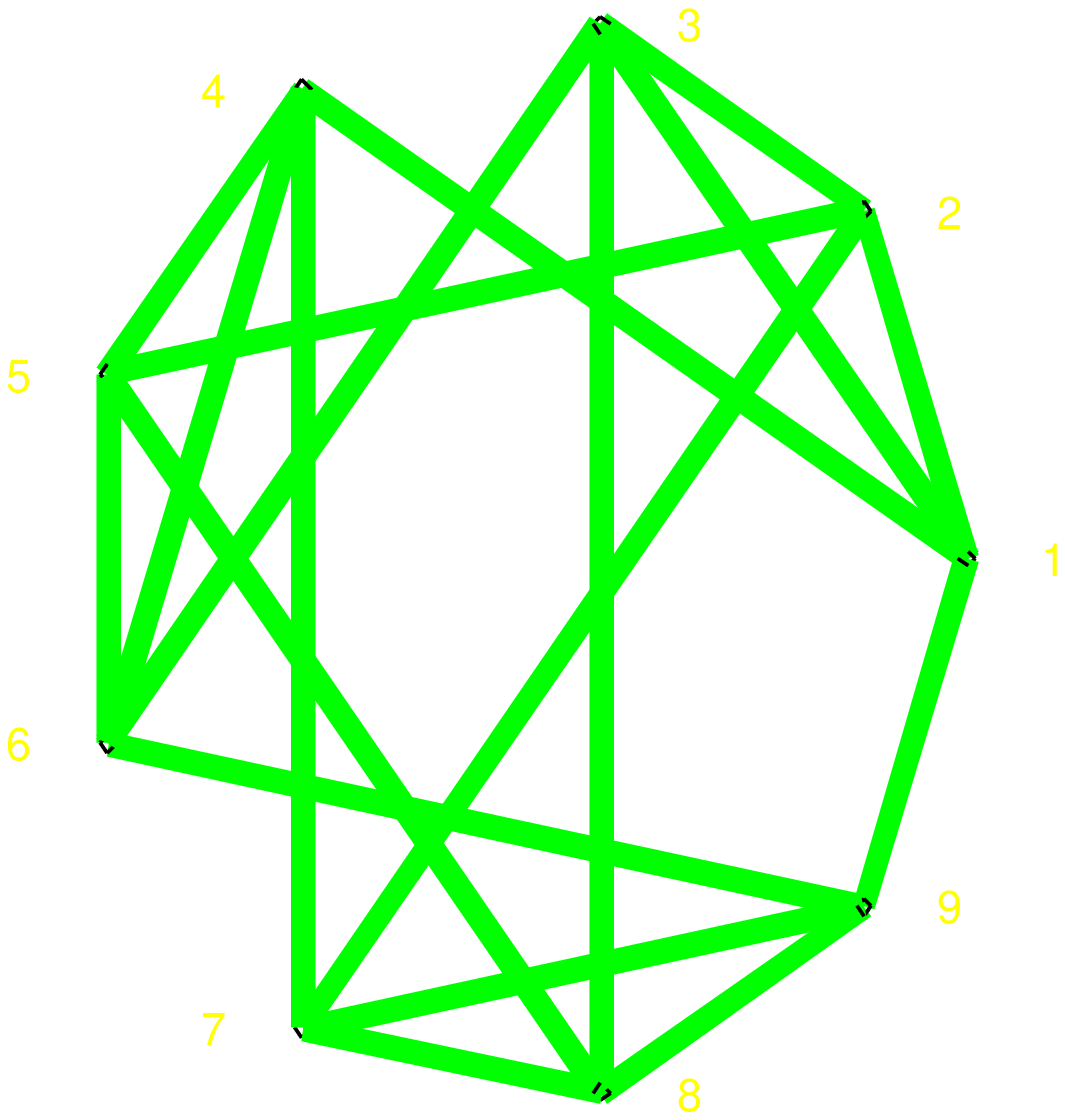}}&200.357~566~429&18&?&$P_{7,11}$&$z_3$&$C^9_{1,3}$, \cite{Panzer:PhD}\\[-6mm]
11?&&\multicolumn{6}{l}{$Q_{11,3}$}\\[1ex]\hline\hline
$P_{8,1}$&\hspace*{-2mm}\raisebox{-9mm}{\includegraphics[width=12mm]{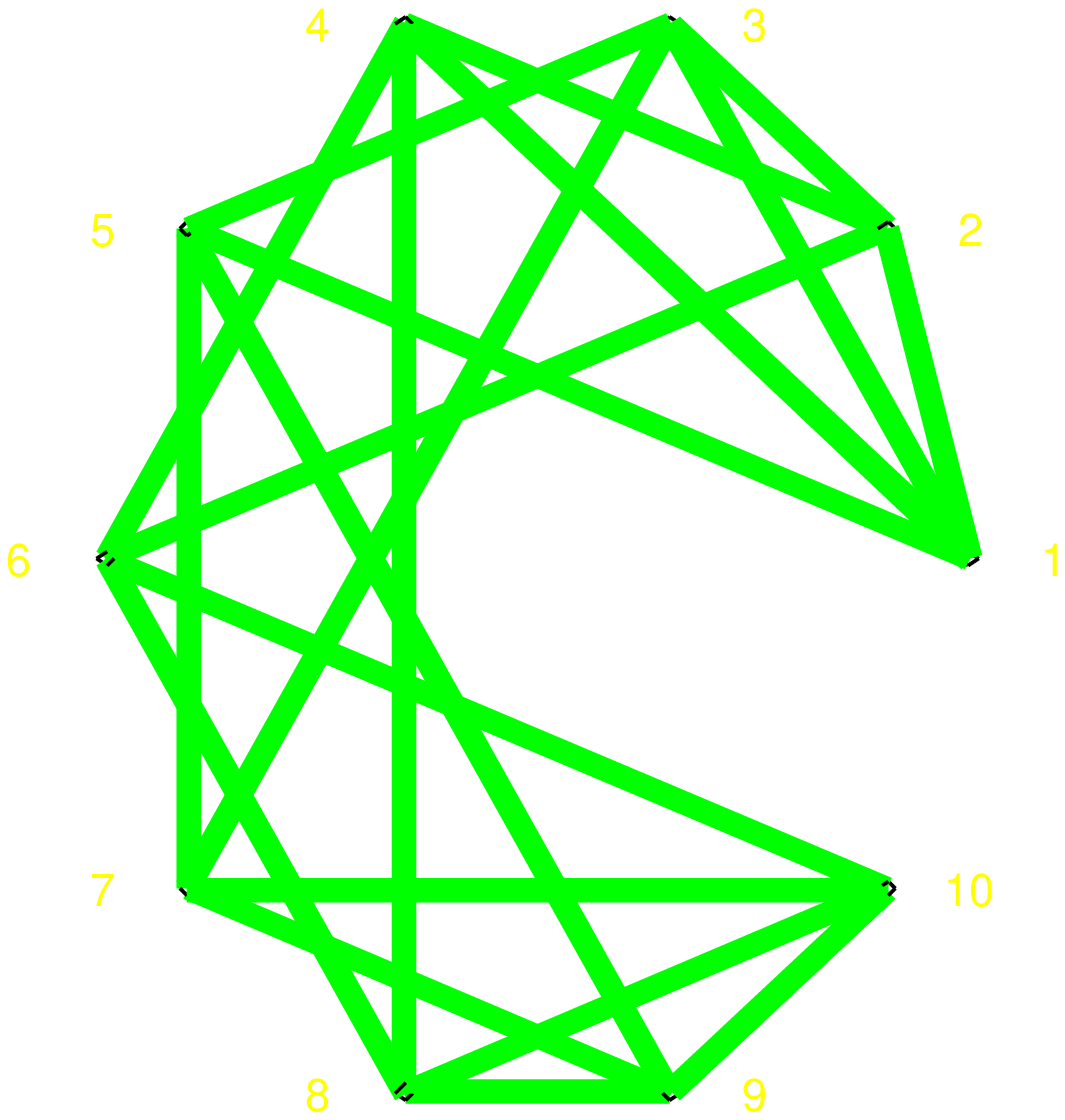}}&1~716.210~576~104&20&2635776&$P_3$&1&$C^{10}_{1,2}$\\[-6mm]
13&&\multicolumn{6}{l}{$1716Q_{13,1}$}\\[1ex]\hline
$P_{8,2}$&\hspace*{-2mm}\raisebox{-9mm}{\includegraphics[width=12mm]{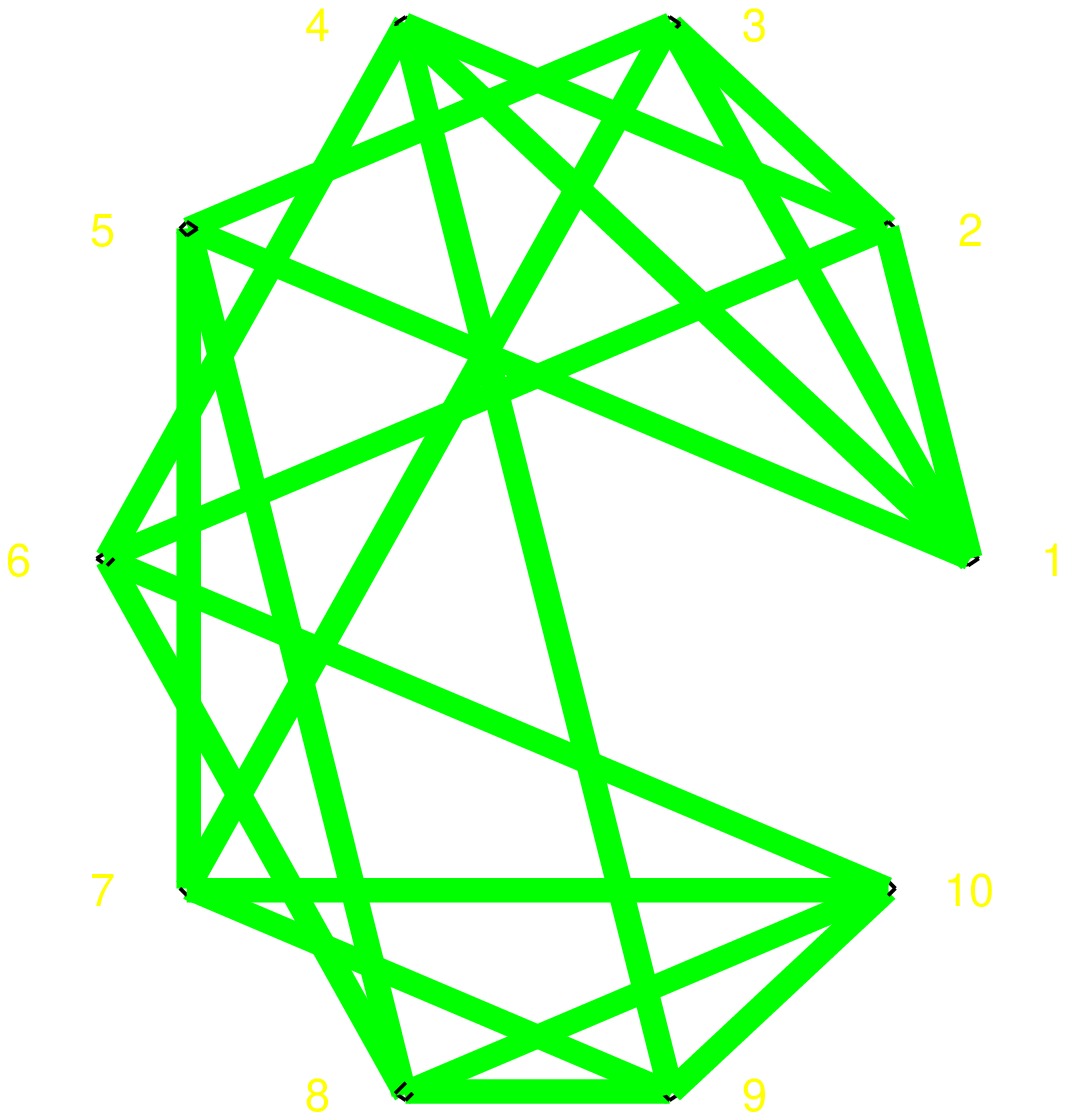}}&1~145.592~929~599&2&12&$P_3$&1&\\[-6mm]
13&&\multicolumn{6}{l}{$\frac{25147347}{22400}Q_{13,1}-\frac{16881}{1400}Q_{13,2}+\frac{459}{112}Q_{13,3}+\frac{1305}{8}Q_3^2Q_7-135Q_3Q_5^2$}\\[1ex]\hline
$P_{8,3}$&\hspace*{-2mm}\raisebox{-9mm}{\includegraphics[width=12mm]{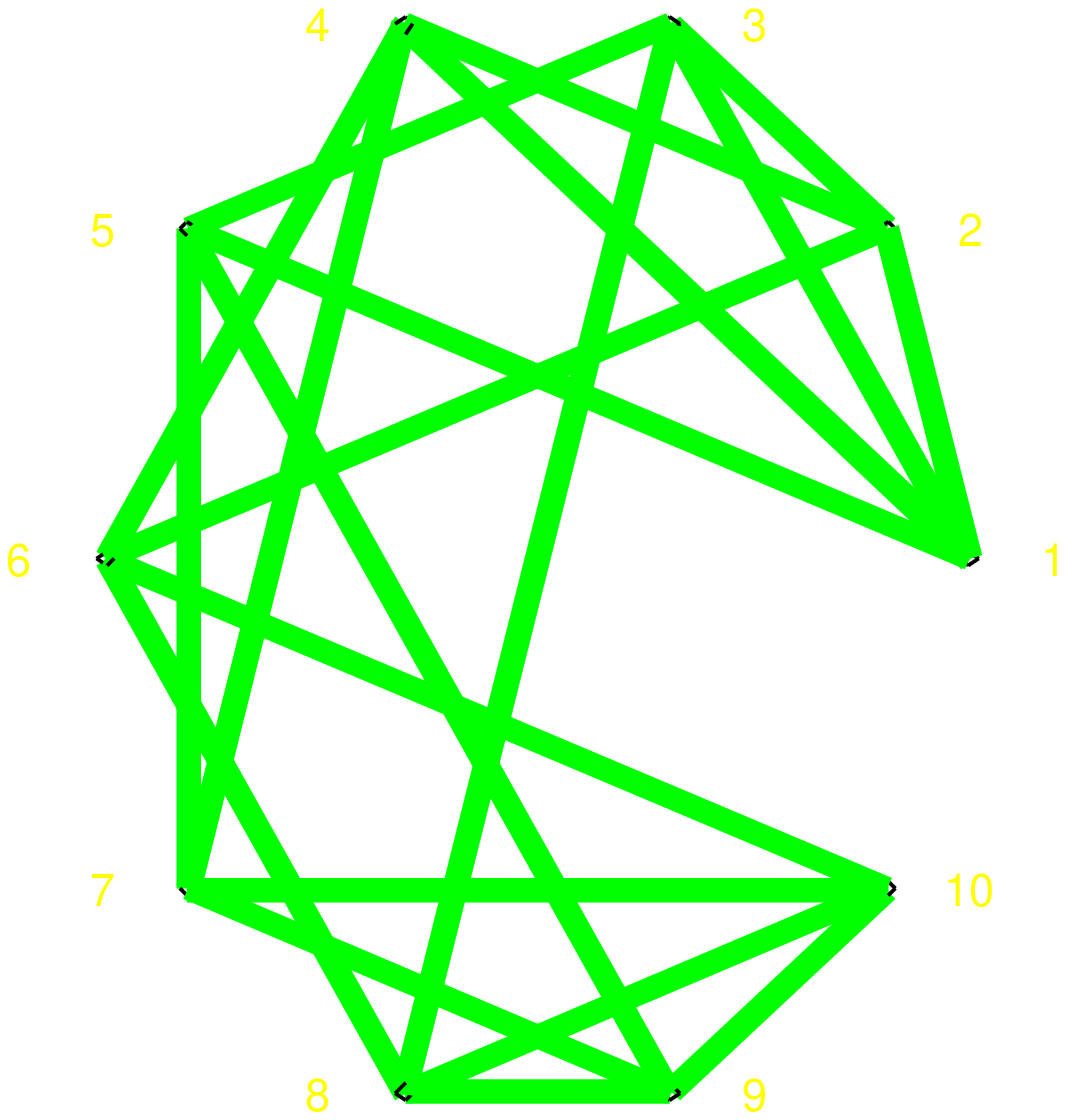}}&1~105.107~697~390&4&1280&$P_3$&1&\\[-6mm]
13&&\multicolumn{6}{l}{$298Q_{13,1}+56Q_{13,2}-20Q_{13,3}-280Q_3^2Q_7+800Q_3Q_5^2$}\\[1ex]\hline
$P_{8,4}$&\hspace*{-2mm}\raisebox{-9mm}{\includegraphics[width=12mm]{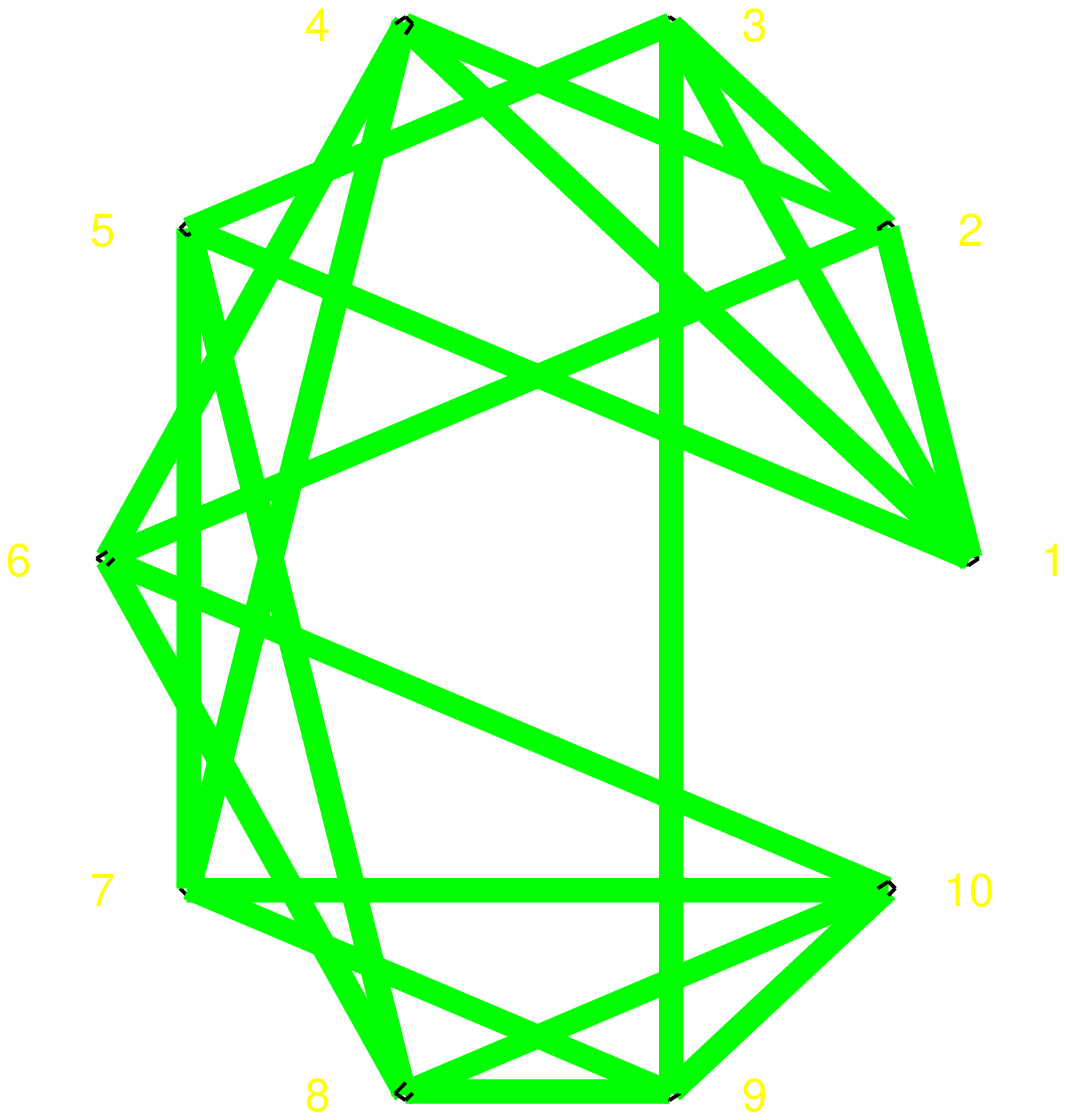}}&966.830~801~986&1&12&$P_3$&1&\\[-6mm]
13&&\multicolumn{6}{l}{$\frac{17124243}{22400}Q_{13,1}-\frac{19689}{1400}Q_{13,2}+\frac{1755}{112}Q_{13,3}+\frac{9}{8}Q_3^2Q_7+135Q_3Q_5^2$}\\[1ex]\hline
$P_{8,5}$&\hspace*{-2mm}\raisebox{-9mm}{\includegraphics[width=12mm]{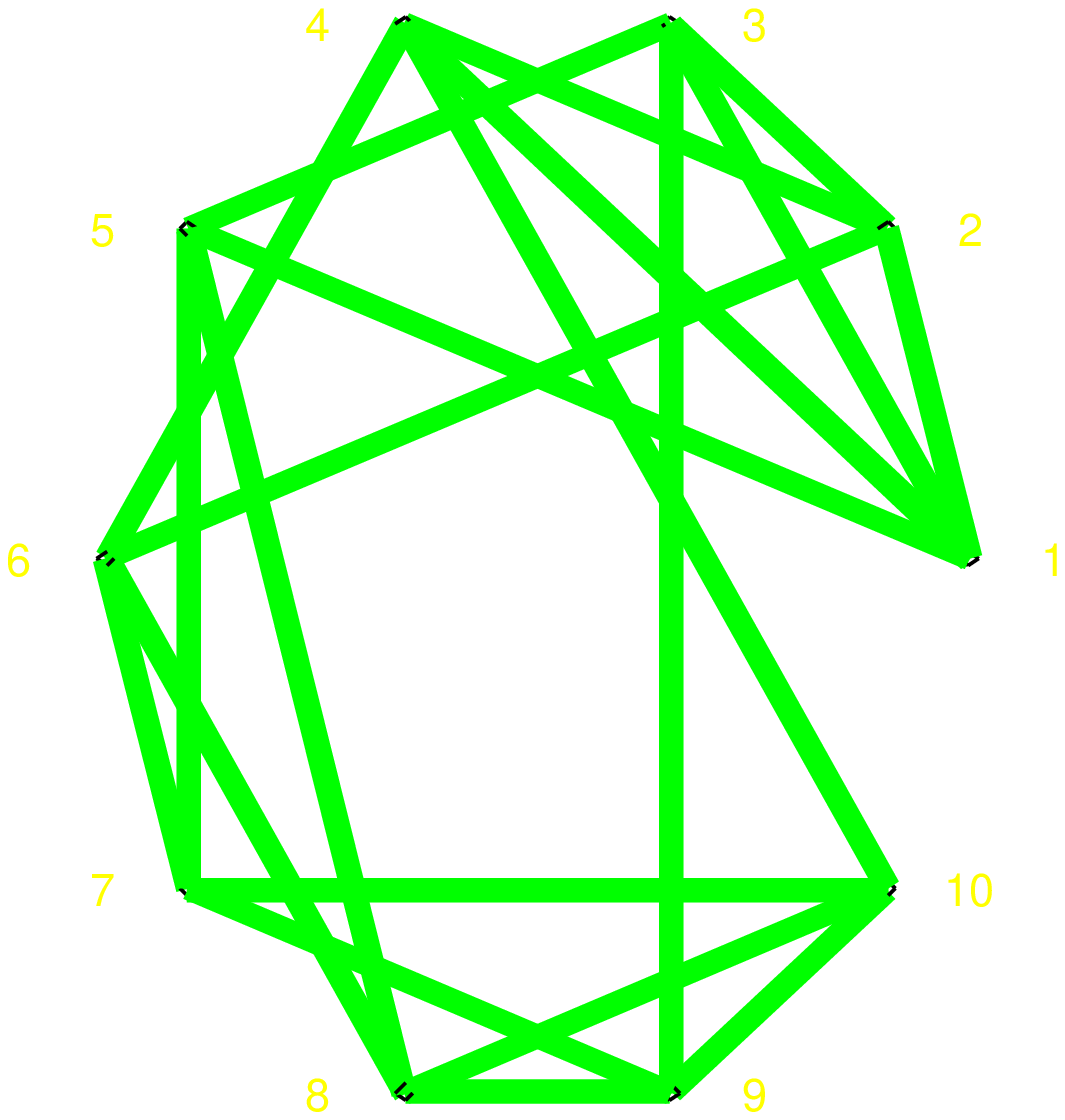}}&844.512~518~603&4&24&$P_3^2$&0&\\[-6mm]
12&&\multicolumn{6}{l}{$1536Q_{12,1}-1280Q_{12,2}+36Q_3Q_9+\frac{1299}{2}Q_5Q_7$}\\[1ex]\hline
$P_{8,6}$&\hspace*{-2mm}\raisebox{-9mm}{\includegraphics[width=12mm]{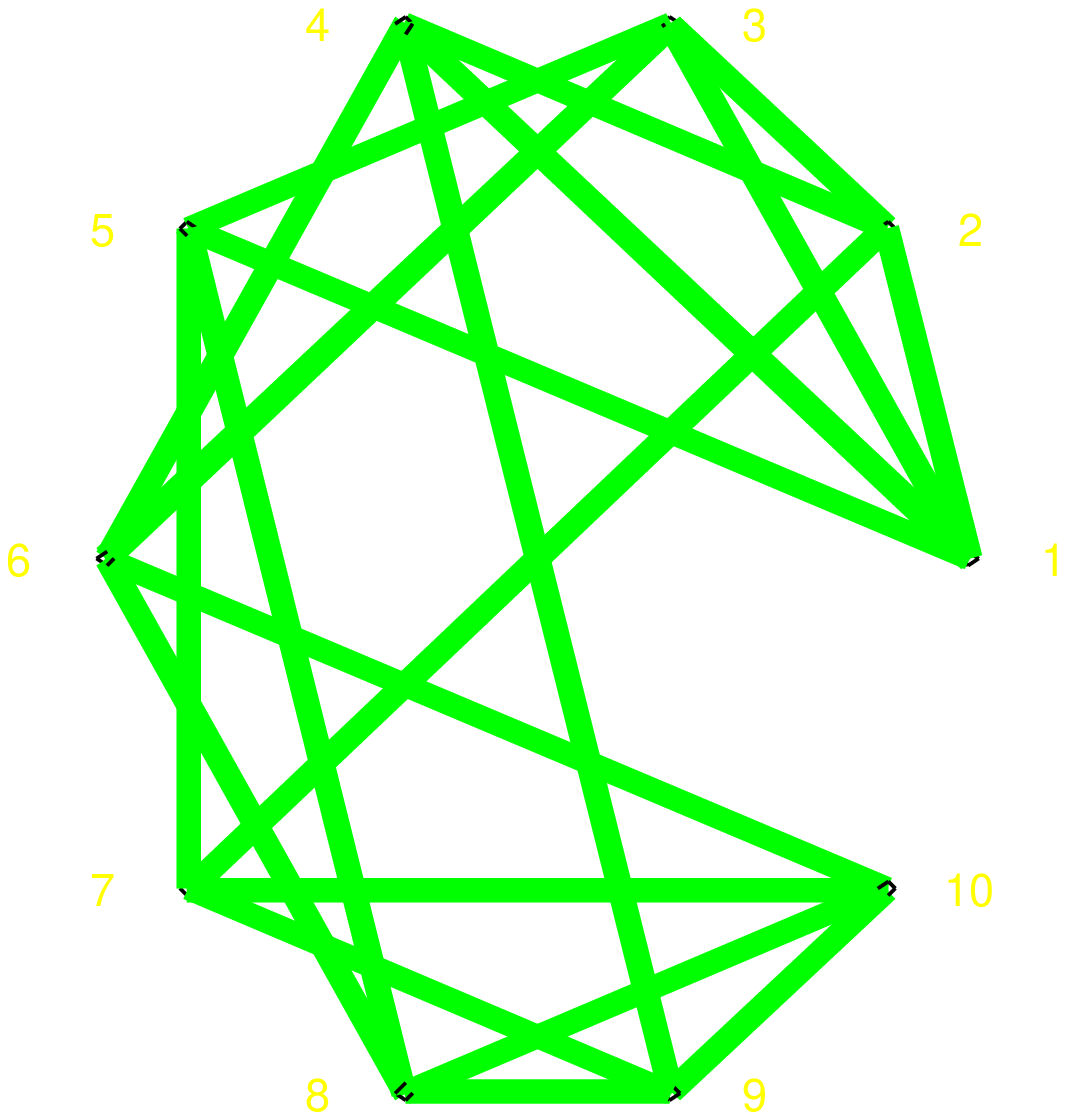}}&904.280~824~357&4&32&$P_3$&1&\\[-6mm]
13&&\multicolumn{6}{l}{$\frac{214841}{336}Q_{13,1}-\frac{423}{7}Q_{13,2}+\frac{705}{14}Q_{13,3}+183Q_3^2Q_7$}\\[1ex]\hline
$P_{8,7}$&\hspace*{-2mm}\raisebox{-9mm}{\includegraphics[width=12mm]{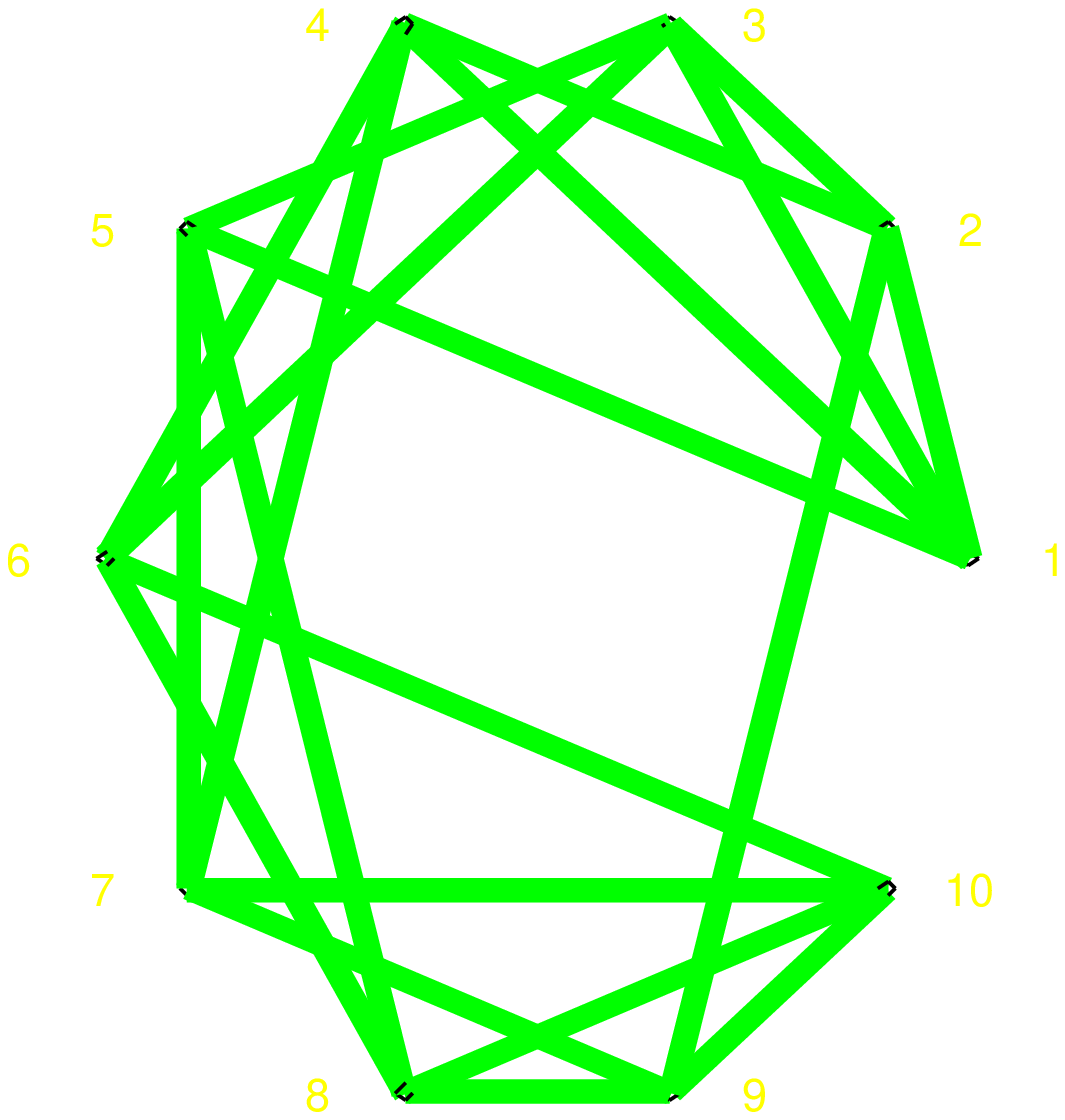}}&847.646~115~639&2&144&$P_3$&1&\\[-6mm]
13&&\multicolumn{6}{l}{$\frac{2061501}{2800}Q_{13,1}+\frac{13527}{175}Q_{13,2}-\frac{675}{14}Q_{13,3}$}\\[1ex]\hline
$P_{8,8}$&\hspace*{-2mm}\raisebox{-9mm}{\includegraphics[width=12mm]{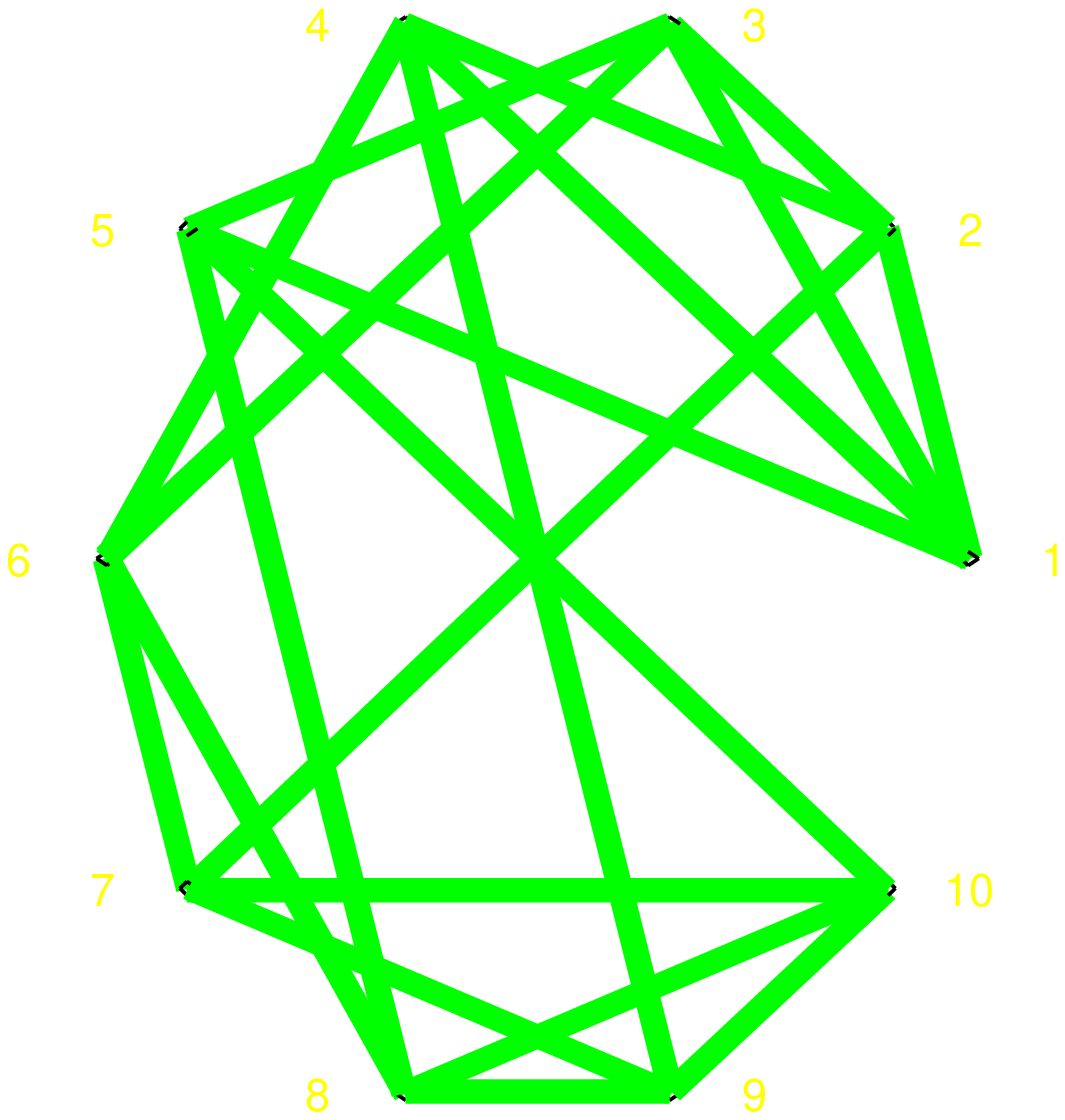}}&847.646~115~639&2&144&$P_3$&1&twist\\[-6mm]
13&&\multicolumn{6}{l}{$P_{8,7}$}\\[1ex]\hline
$P_{8,9}$&\hspace*{-2mm}\raisebox{-9mm}{\includegraphics[width=12mm]{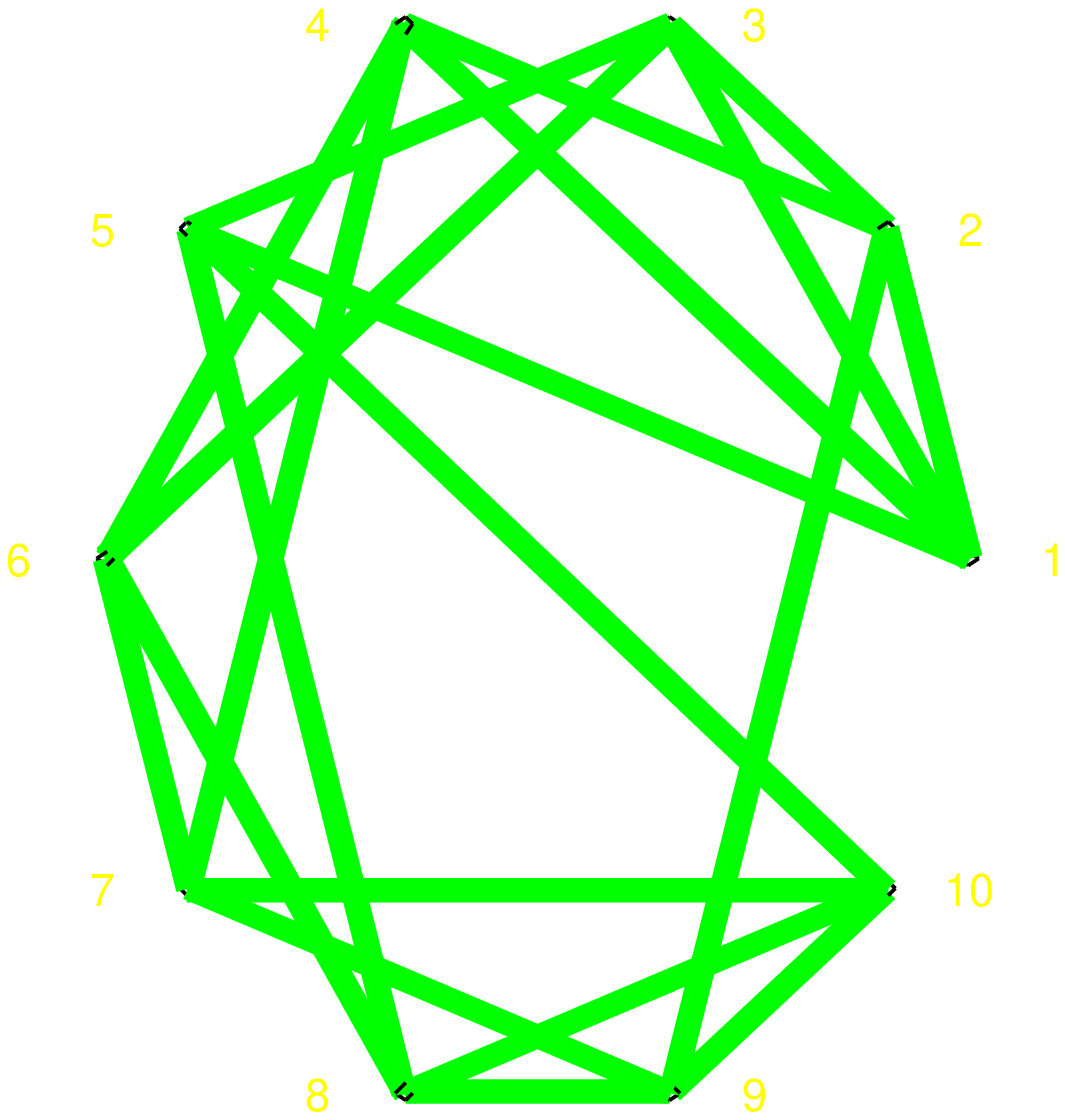}}&904.280~824~357&2&32&$P_3$&1&twist\\[-6mm]
13&&\multicolumn{6}{l}{$P_{8,6}$}\\[1ex]\hline
$P_{8,10}$&\hspace*{-2mm}\raisebox{-9mm}{\includegraphics[width=12mm]{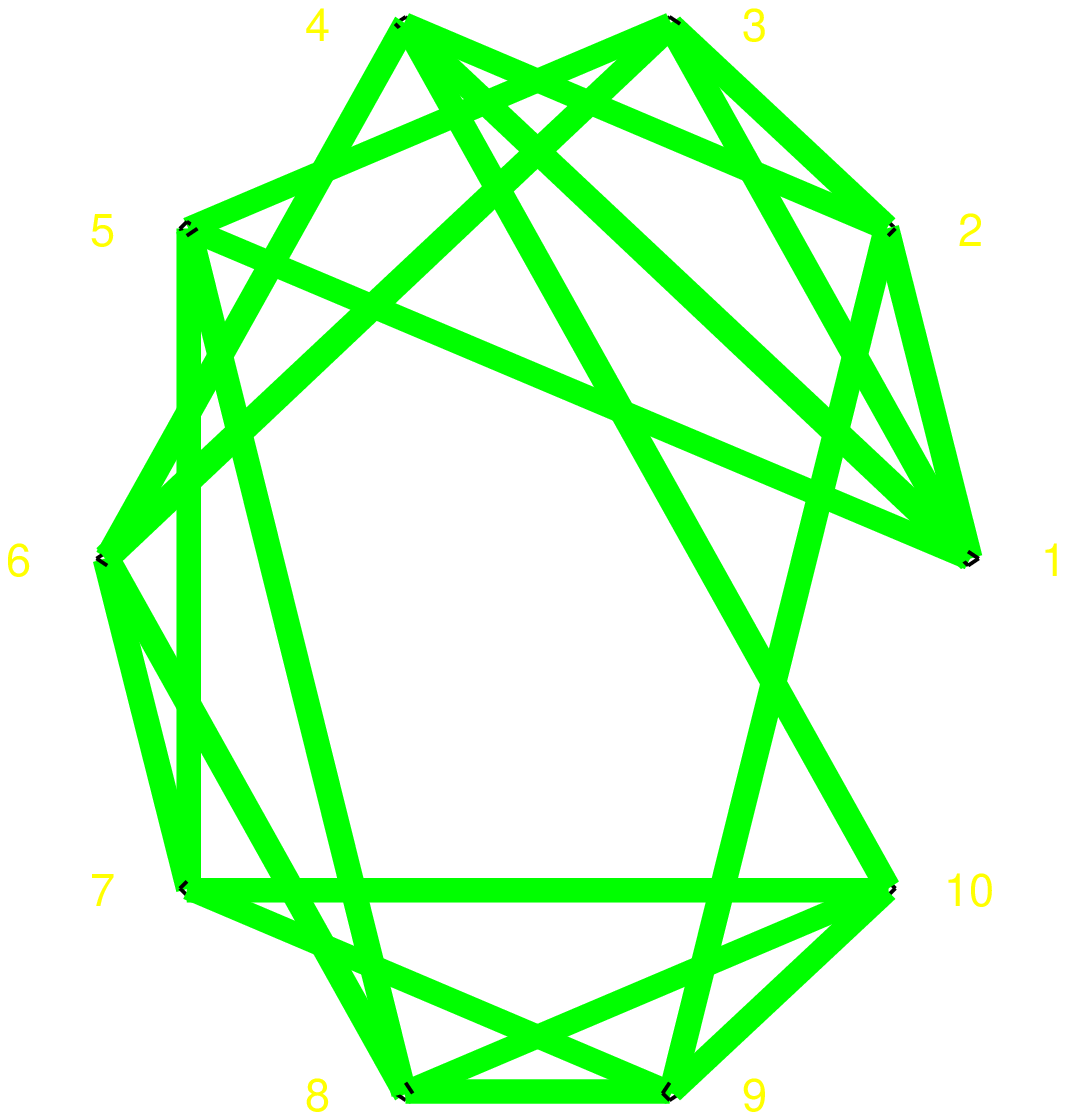}}&735.764~103~468&2&72&$P_3^2$&0&\\[-6mm]
12&&\multicolumn{6}{l}{$1536Q_{12,1}-1280Q_{12,2}-\frac{63}{2}Q_3Q_9+\frac{2493}{4}Q_5Q_7$}\\[1ex]\hline
$P_{8,11}$&\hspace*{-2mm}\raisebox{-9mm}{\includegraphics[width=12mm]{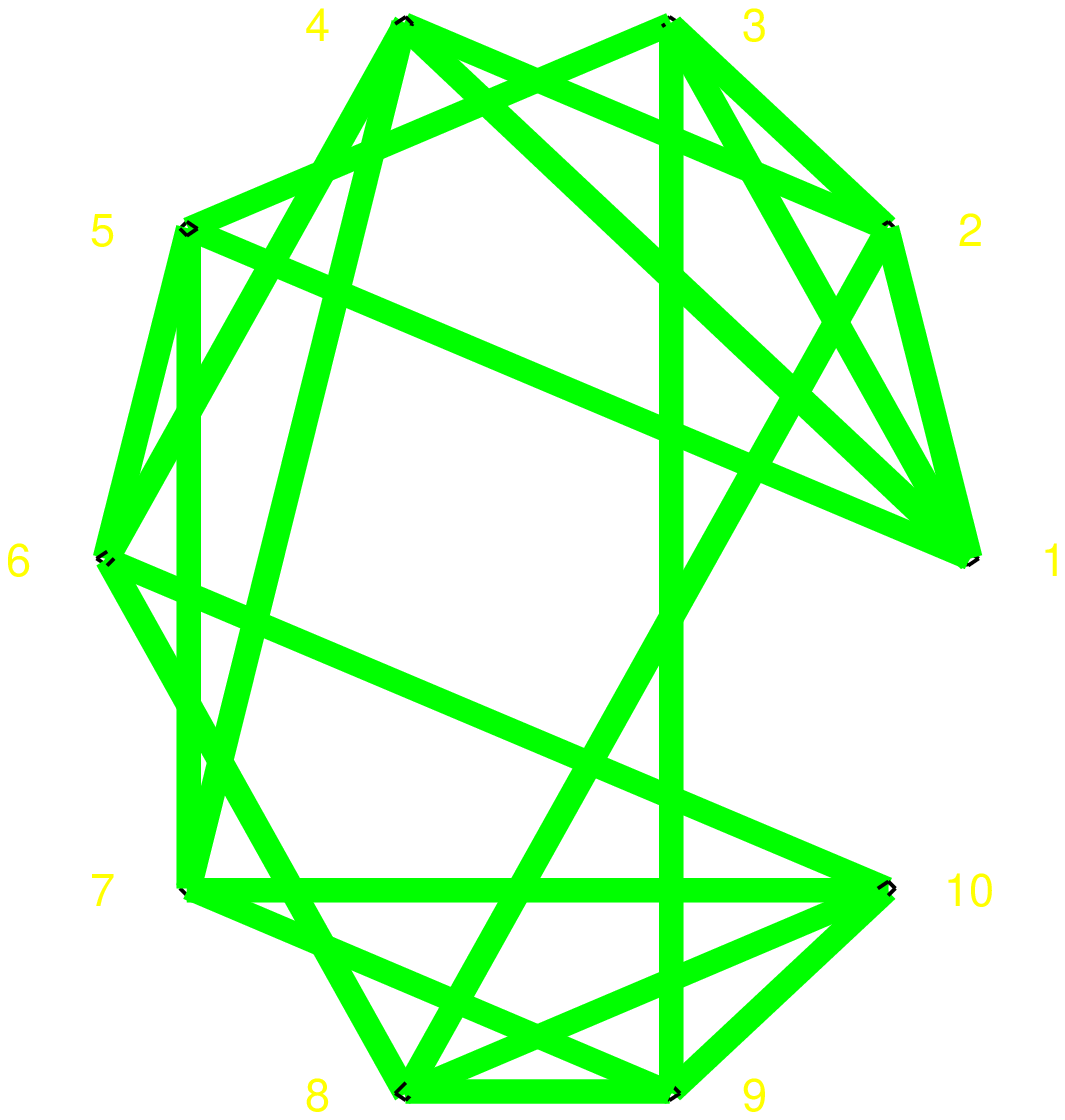}}&805.347~388~507&4&16&$P_3^2$&0&\\[-6mm]
12&&\multicolumn{6}{l}{$\frac{10240}{69}Q_{12,1}+\frac{81920}{69}Q_{12,2}-\frac{2560}{69}Q_{12,3}+\frac{45503}{69}Q_3Q_9+\frac{305}{46}Q_5Q_7-12Q_3^4$}
\end{tabular}

\begin{tabular}{llllllll}
name&graph&numerical value&$|$Aut$|$&index&anc.&$-c_2$&remarks, [Lit]\\[1ex]
\multicolumn{2}{l}{weight}&\multicolumn{6}{l}{exact value}\\[1ex]\hline\hline
$P_{8,12}$&\hspace*{-2mm}\raisebox{-9mm}{\includegraphics[width=12mm]{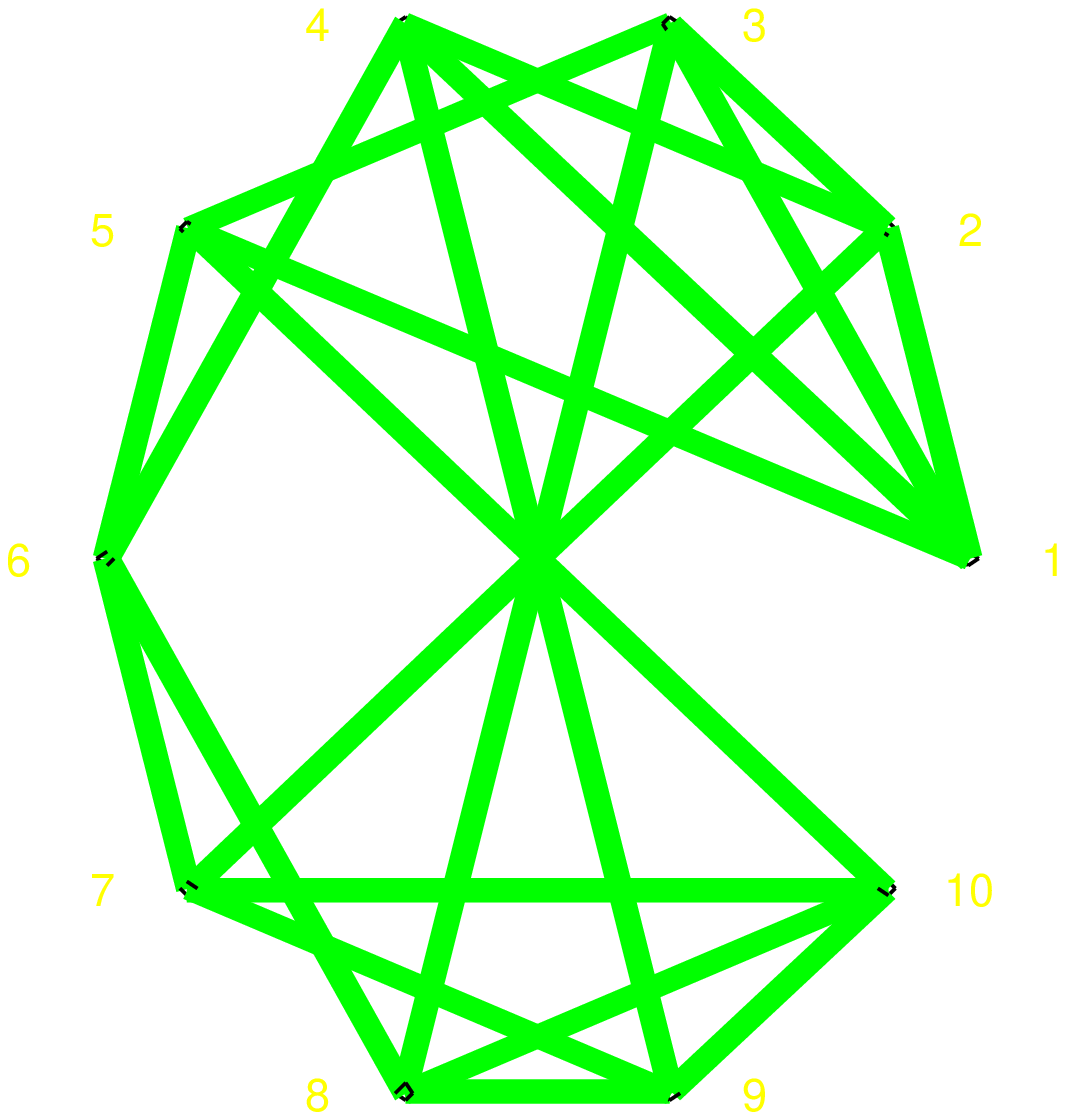}}&688.898~361~296&2&288&$P_3^2$&0&\\[-6mm]
12&&\multicolumn{6}{l}{$1024Q_{12,2}-1008Q_3Q_9+1800Q_5Q_7$}\\[1ex]\hline
$P_{8,13}$&\hspace*{-2mm}\raisebox{-9mm}{\includegraphics[width=12mm]{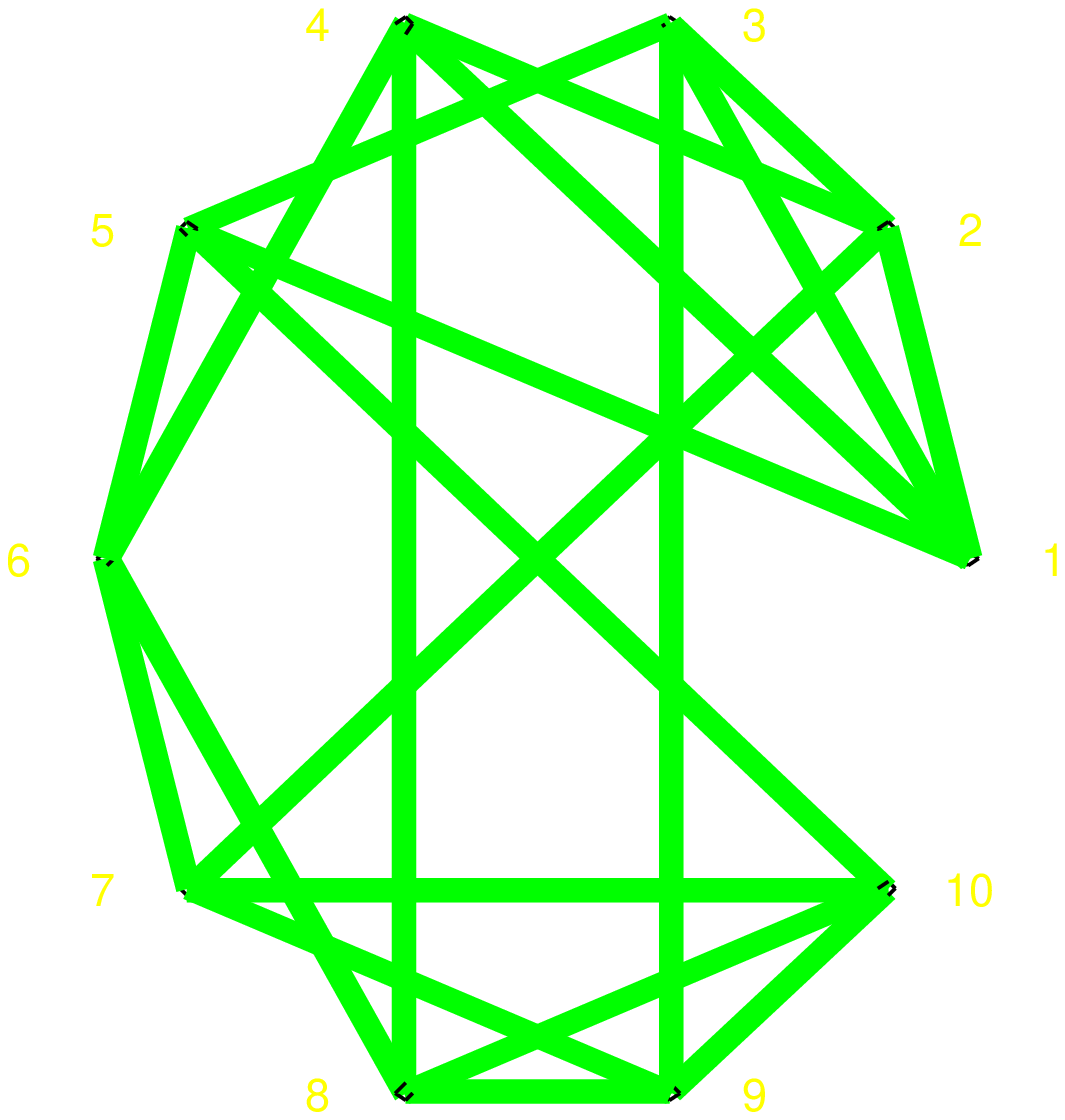}}&742.977~090~366&1&4&$P_3$&1&\\[-6mm]
13&&\multicolumn{6}{l}{$\frac{10087273}{9600}Q_{13,1}+\frac{8007}{200}Q_{13,2}-\frac{813}{16}Q_{13,3}+\frac{2247}{8}Q_3^2Q_7-465Q_3Q_5^2$}\\[1ex]\hline
$P_{8,14}$&\hspace*{-2mm}\raisebox{-9mm}{\includegraphics[width=12mm]{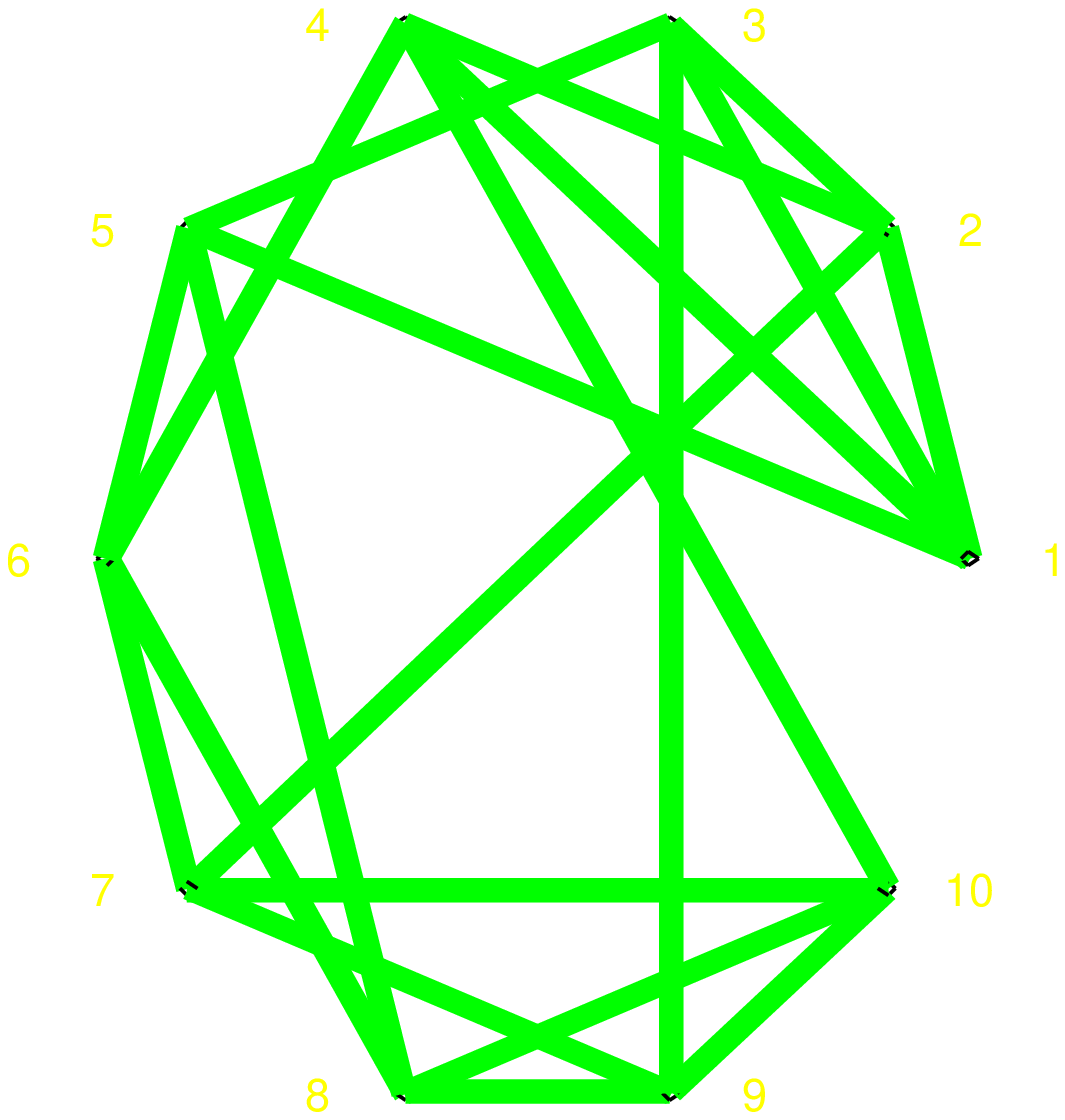}}&749.818~622~995&1&4&$P_3$&1&\\[-6mm]
13&&\multicolumn{6}{l}{$\frac{41038969}{67200}Q_{13,1}-\frac{30129}{1400}Q_{13,2}+\frac{1611}{112}Q_{13,3}+\frac{153}{8}Q_3^2Q_7+105Q_3Q_5^2$}\\[1ex]\hline
$P_{8,15}$&\hspace*{-2mm}\raisebox{-9mm}{\includegraphics[width=12mm]{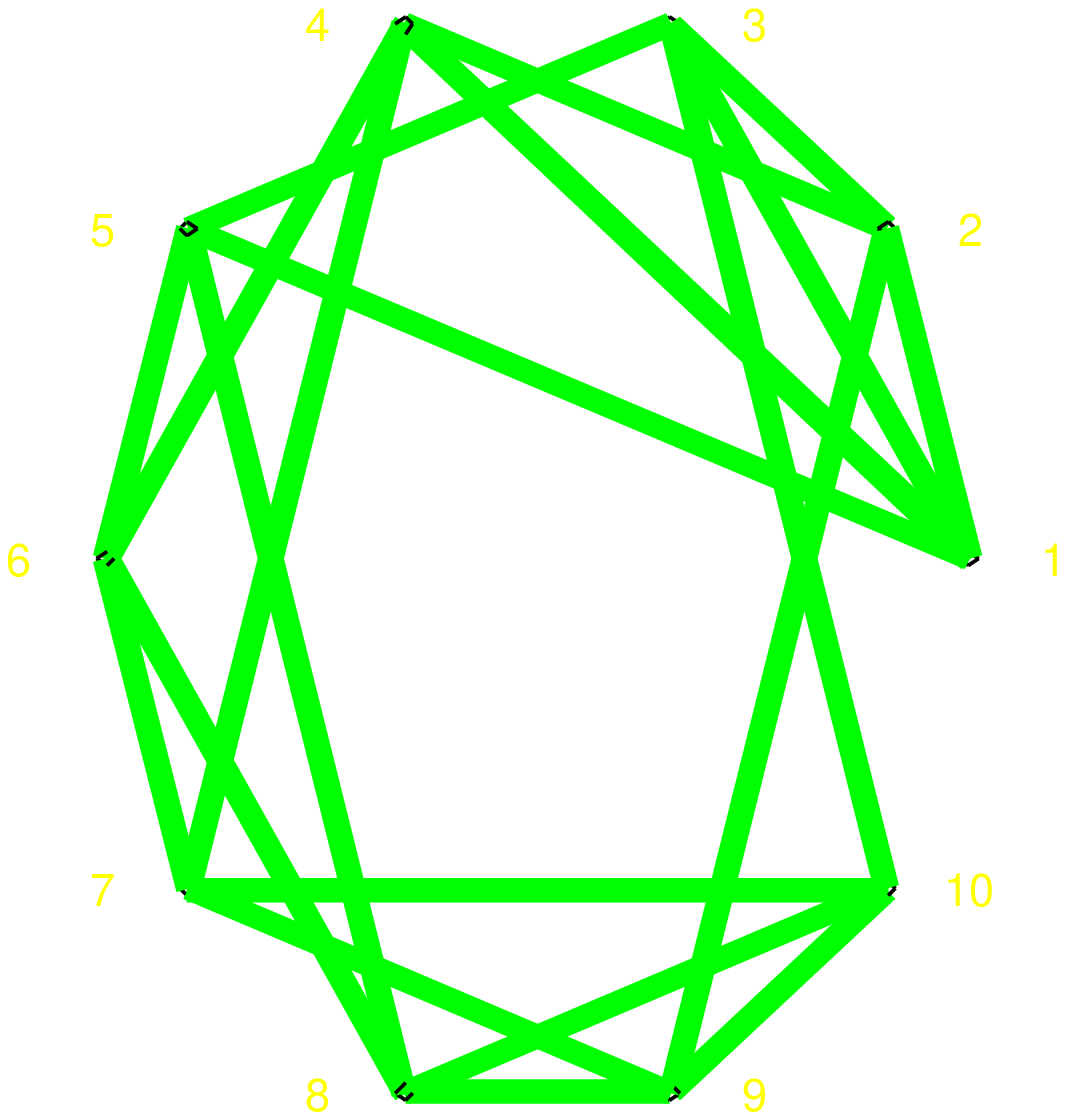}}&805.347~388~507&2&16&$P_3^2$&0&twist\\[-6mm]
12&&\multicolumn{6}{l}{$P_{8,11}$}\\[1ex]\hline
$P_{8,16}$&\hspace*{-2mm}\raisebox{-9mm}{\includegraphics[width=12mm]{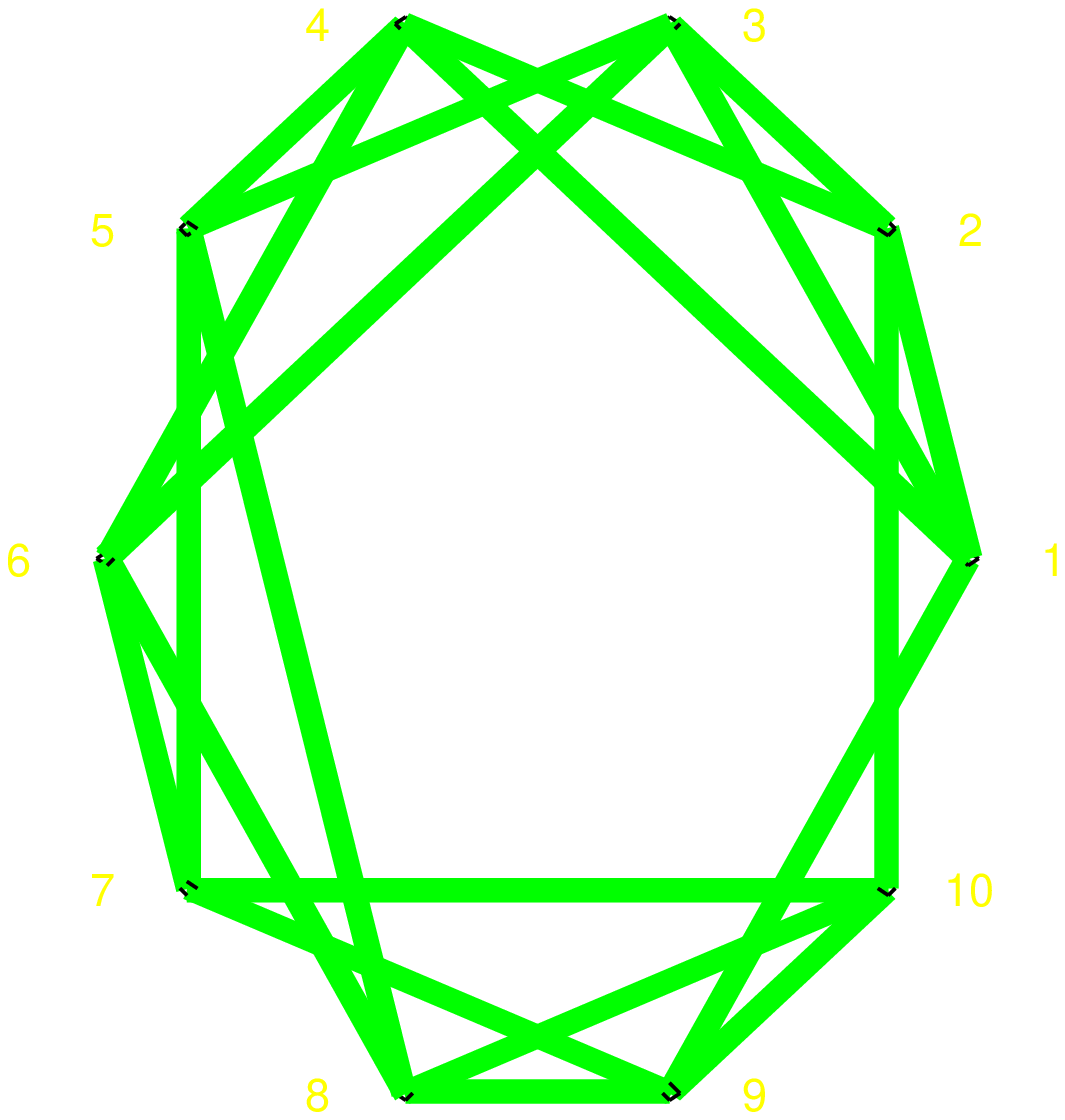}}&633.438~914~549&32&576&$P_3^3$&0&\hfill$[-10080Q_5^2$\\[-6mm]
11,10&&\multicolumn{6}{l}{$-\frac{31851}{5}Q_{11,1}\!+\!\frac{24336}{5}Q_{11,2}\!-\!10240Q_3Q_8\!+\!5040Q_3^2Q_5\!-\!8192Q_{10}\!+\!9648Q_3Q_7$}\\[1ex]\hline
$P_{8,17}$&\hspace*{-2mm}\raisebox{-9mm}{\includegraphics[width=12mm]{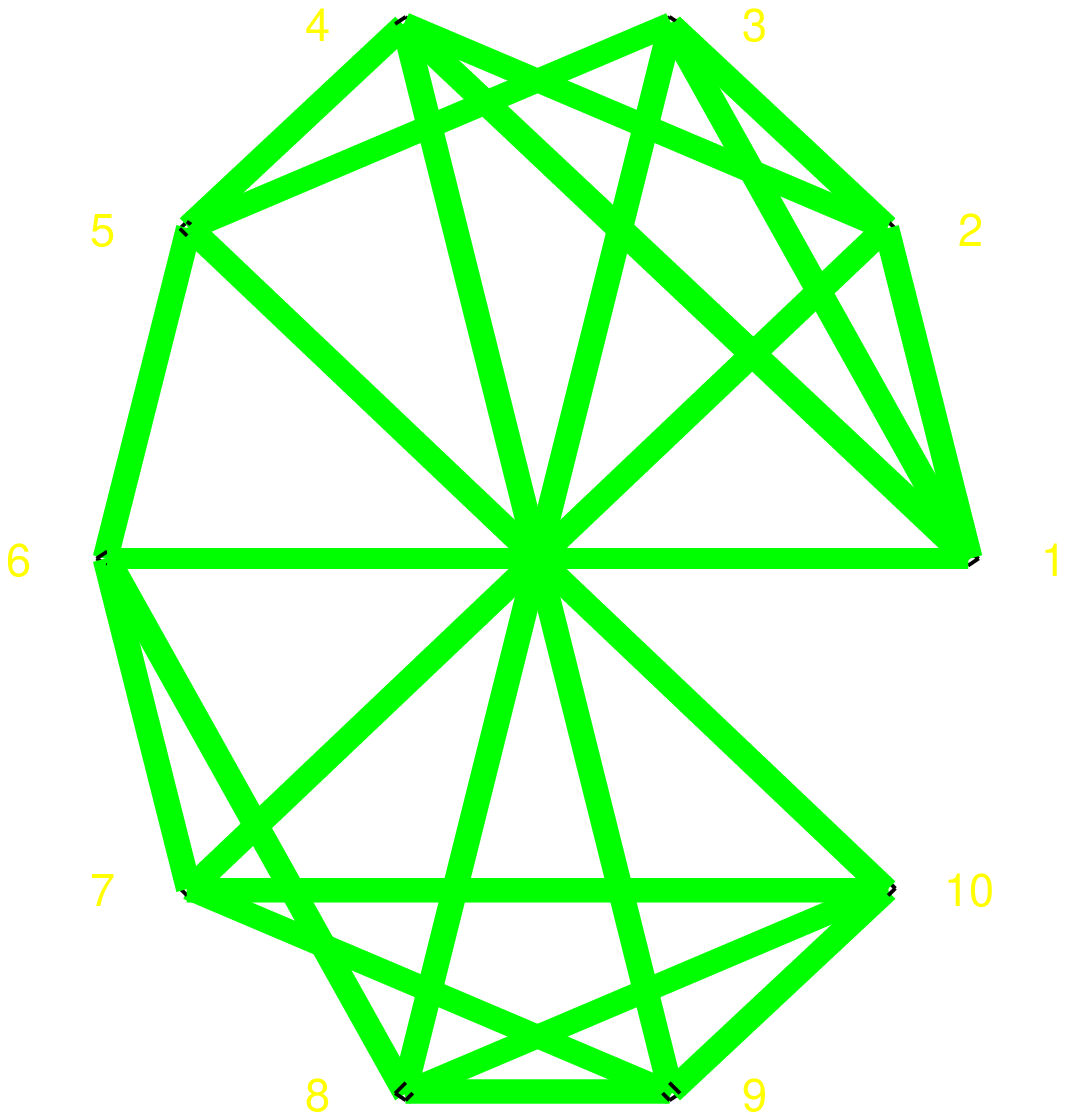}}&589.354~510~434&2&8&$P_3$&1&\hfill$[-1410Q_3Q_5^2$\\[-6mm]
13&&\multicolumn{6}{l}{$\frac{15548993}{4800}Q_{13,1}-\frac{17313}{100}Q_{13,2}+\frac{267}{8}Q_{13,3}+512Q_3Q_{10}+2304Q_5Q_8-\frac{825}{4}Q_3^2Q_7$}\\[1ex]\hline
$P_{8,18}$&\hspace*{-2mm}\raisebox{-9mm}{\includegraphics[width=12mm]{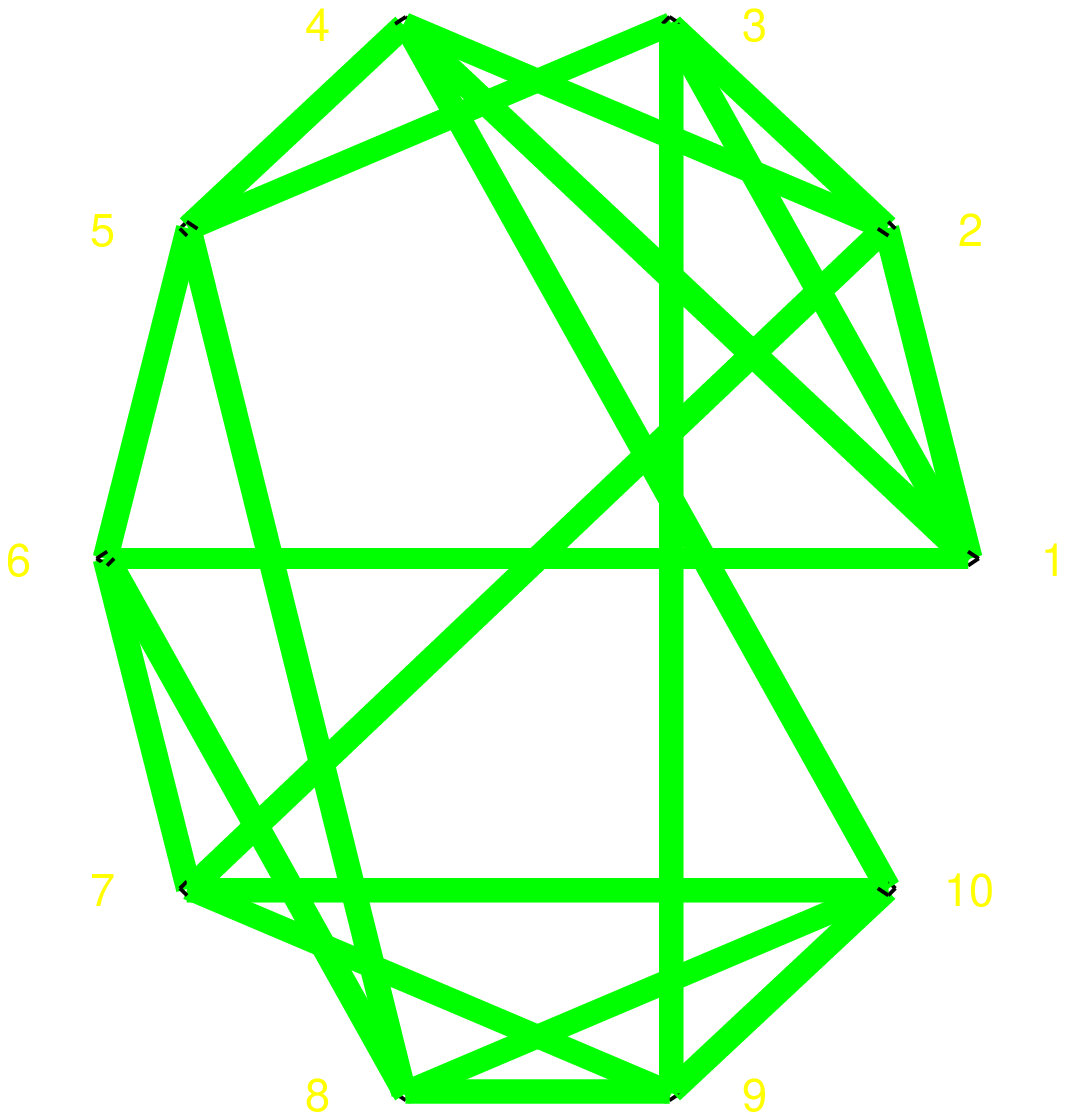}}&641.723~358~297&2&48&$P_3^2$&0&\\[-6mm]
12&&\multicolumn{6}{l}{$727Q_3Q_9-\frac{735}{2}Q_5Q_7+72Q_3^4$}\\[1ex]\hline
$P_{8,19}$&\hspace*{-2mm}\raisebox{-9mm}{\includegraphics[width=12mm]{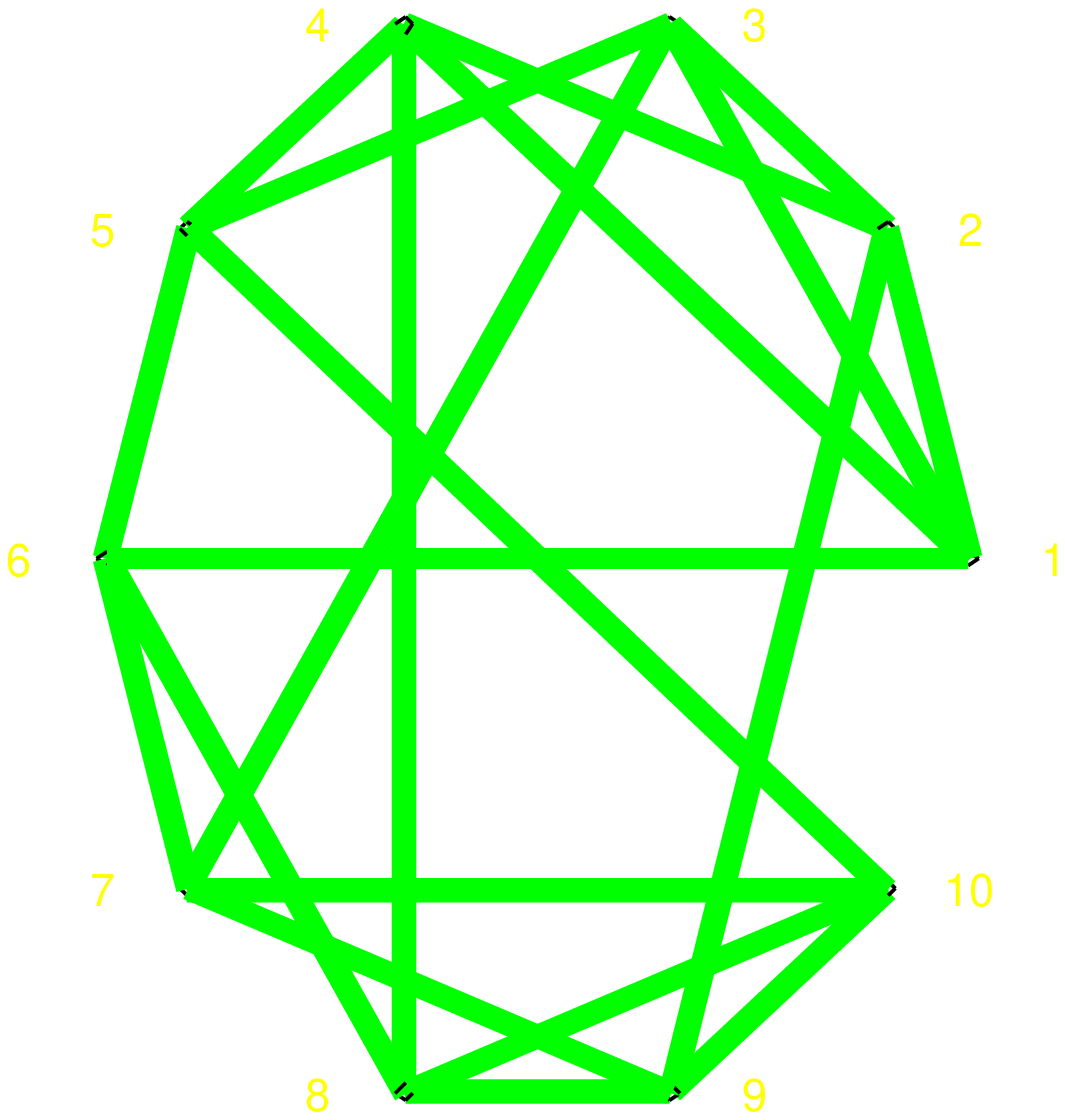}}&598.617~690~750&4&32&$P_3^2$&0&\\[-6mm]
12&&\multicolumn{6}{l}{$\frac{10240}{69}Q_{12,1}+\frac{81920}{69}Q_{12,2}-\frac{2560}{69}Q_{12,3}+\frac{13970}{69}Q_3Q_9+\frac{11020}{23}Q_5Q_7-84Q_3^4$}\\[1ex]\hline
$P_{8,20}$&\hspace*{-2mm}\raisebox{-9mm}{\includegraphics[width=12mm]{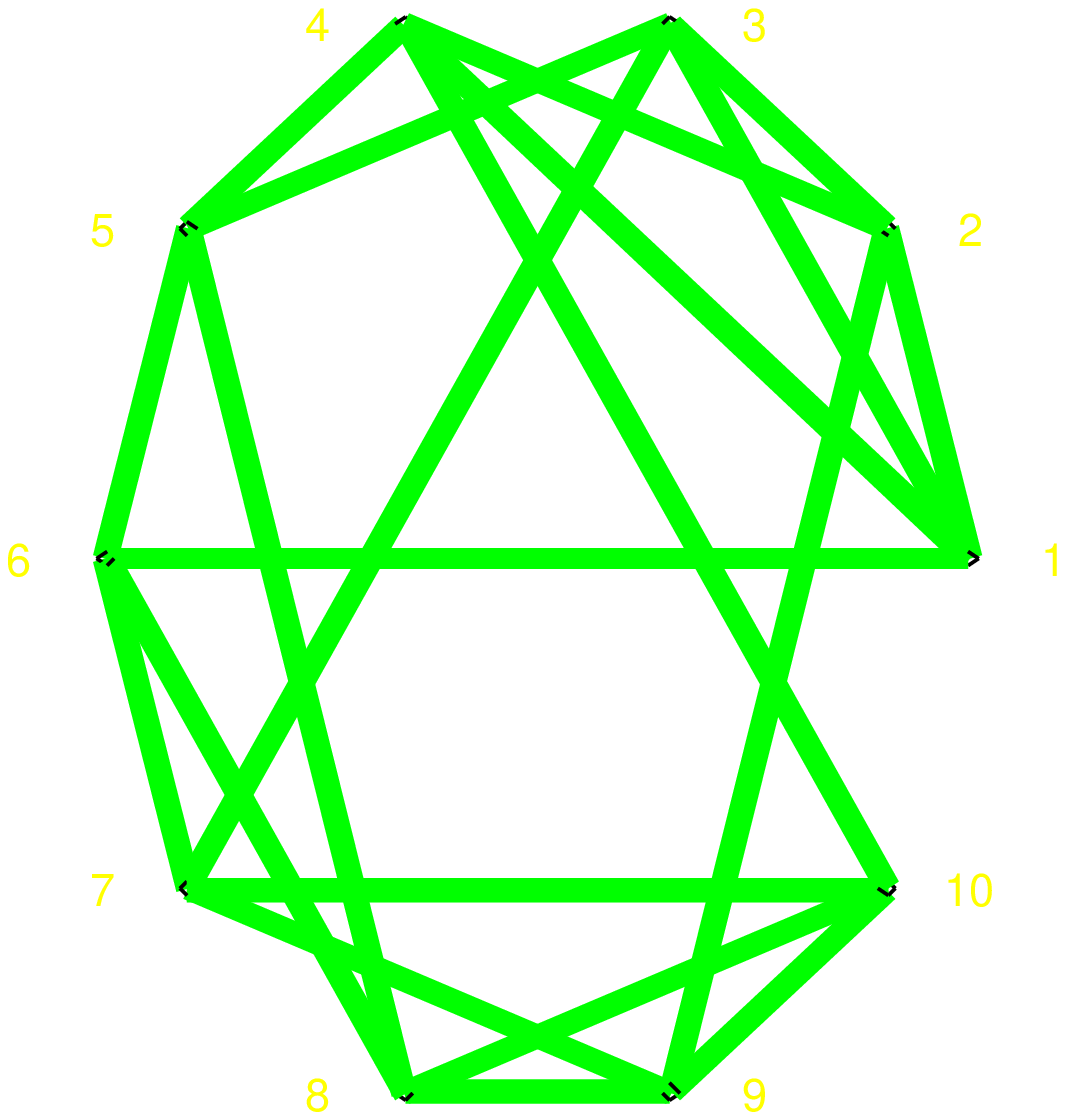}}&641.346~699~620&1&6&$P_3$&1&\hfill$[+\frac{1035}{2}Q_3Q_5^2$\\[-6mm]
13&&\multicolumn{6}{l}{$\frac{4375463}{44800}Q_{13,1}\!+\!\frac{383001}{2800}Q_{13,2}\!-\!\frac{23607}{224}Q_{13,3}\!-\!256Q_3Q_{10}\!+\!256Q_5Q_8\!-\!\frac{1953}{16}Q_3^2Q_7$}\\[1ex]\hline
$P_{8,21}$&\hspace*{-2mm}\raisebox{-9mm}{\includegraphics[width=12mm]{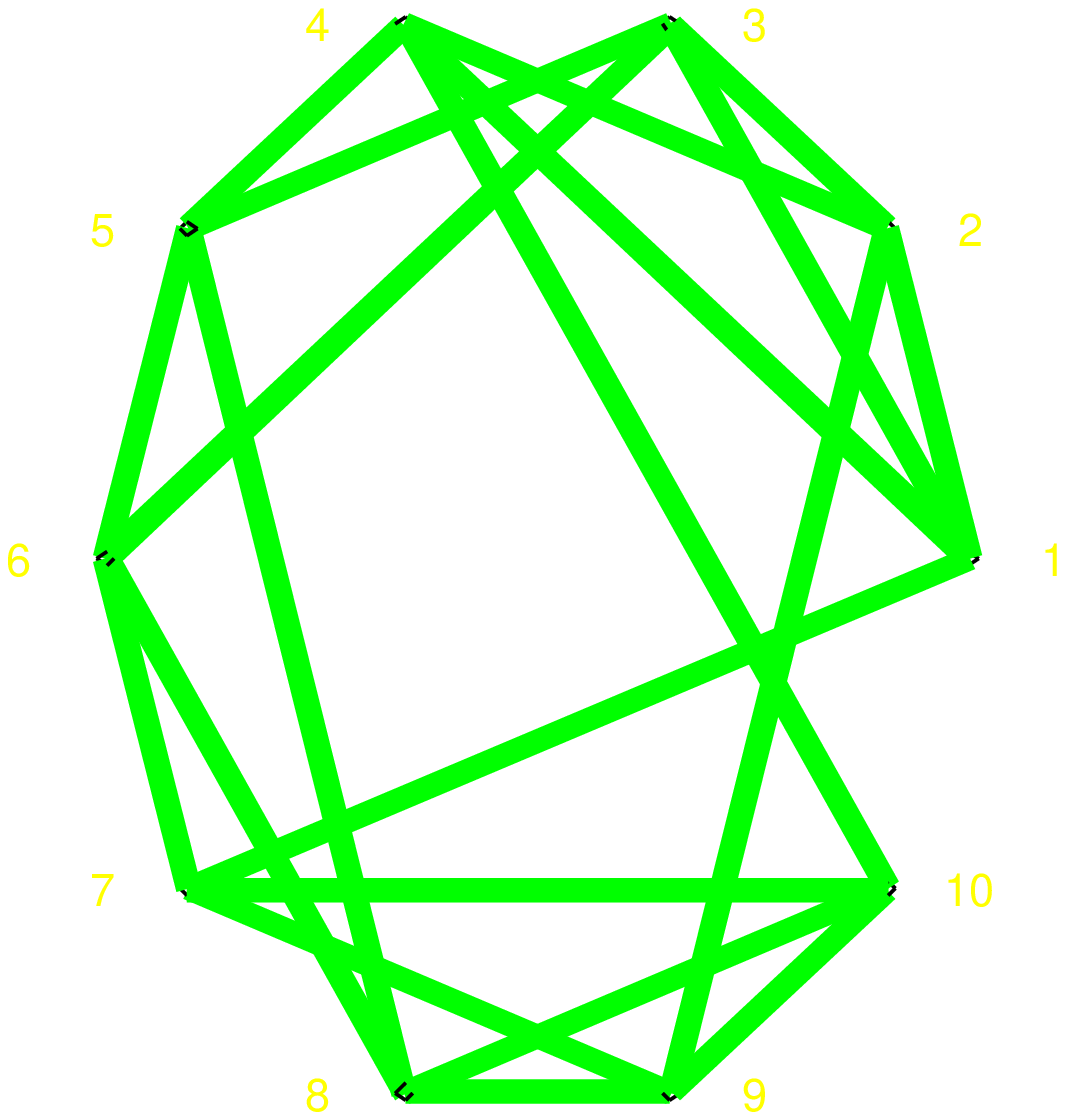}}&742.977~090~366&2&4&$P_3$&1&Fourier, twist\\[-6mm]
13&&\multicolumn{6}{l}{$P_{8,13}$}\\[1ex]\hline
$P_{8,22}$&\hspace*{-2mm}\raisebox{-9mm}{\includegraphics[width=12mm]{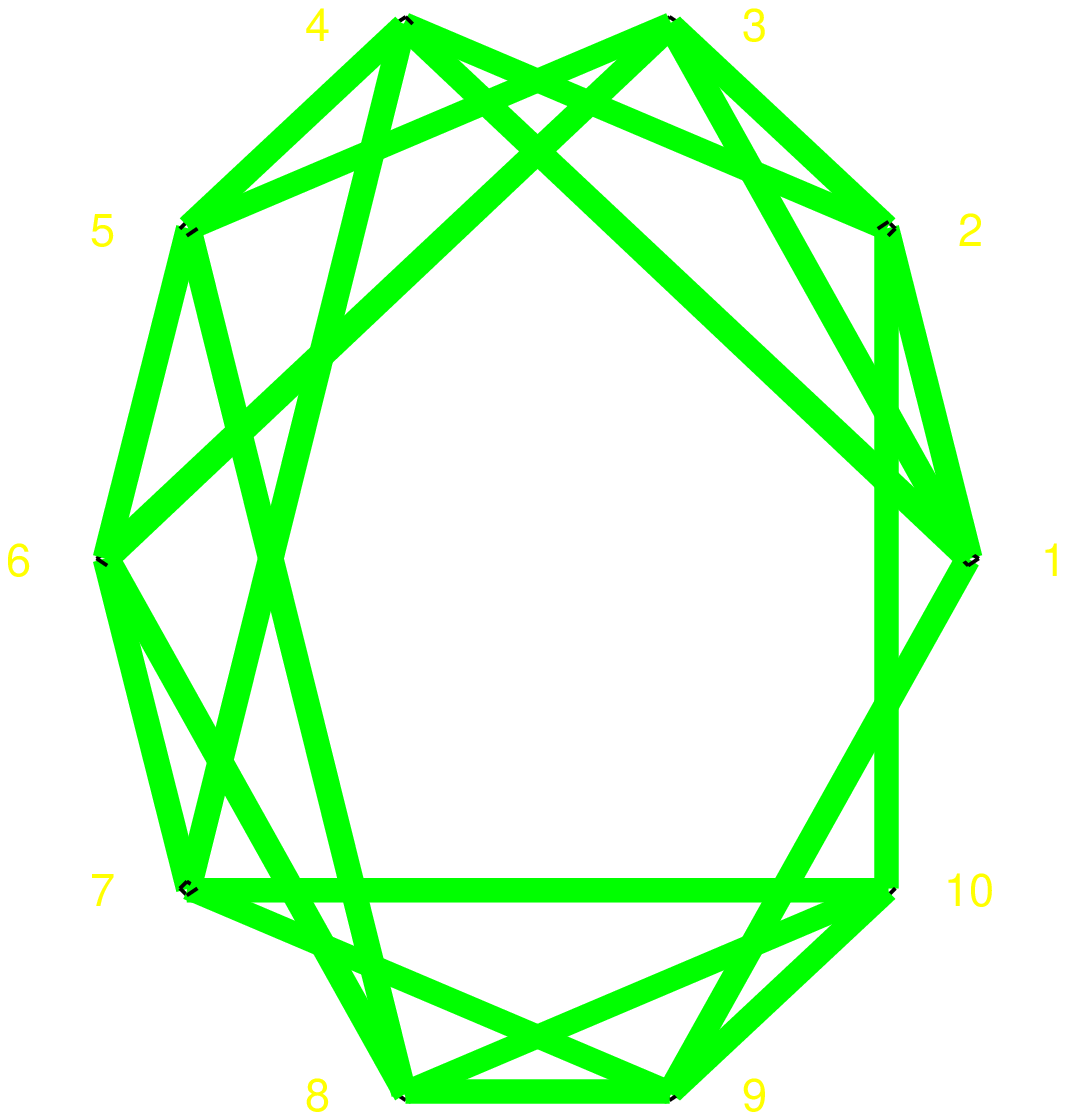}}&735.764~103~468&4&72&$P_3^2$&0&twist\\[-6mm]
12&&\multicolumn{6}{l}{$P_{8,10}$}\\[1ex]\hline
$P_{8,23}$&\hspace*{-2mm}\raisebox{-9mm}{\includegraphics[width=12mm]{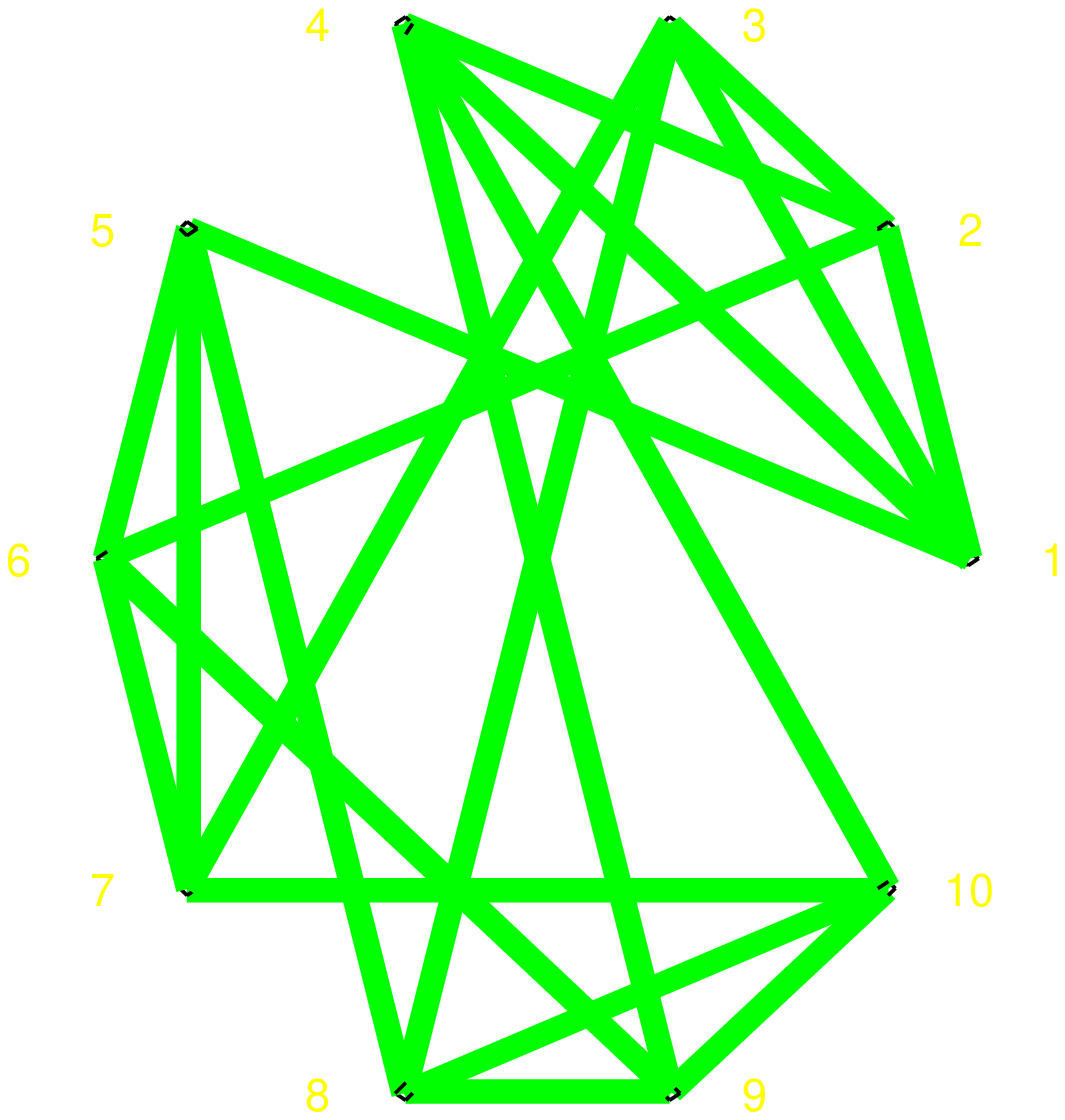}}&589.354~510~434&2&8&$P_3$&1&twist\\[-6mm]
13&&\multicolumn{6}{l}{$P_{8,17}$}\\[1ex]\hline
$P_{8,24}$&\hspace*{-2mm}\raisebox{-9mm}{\includegraphics[width=12mm]{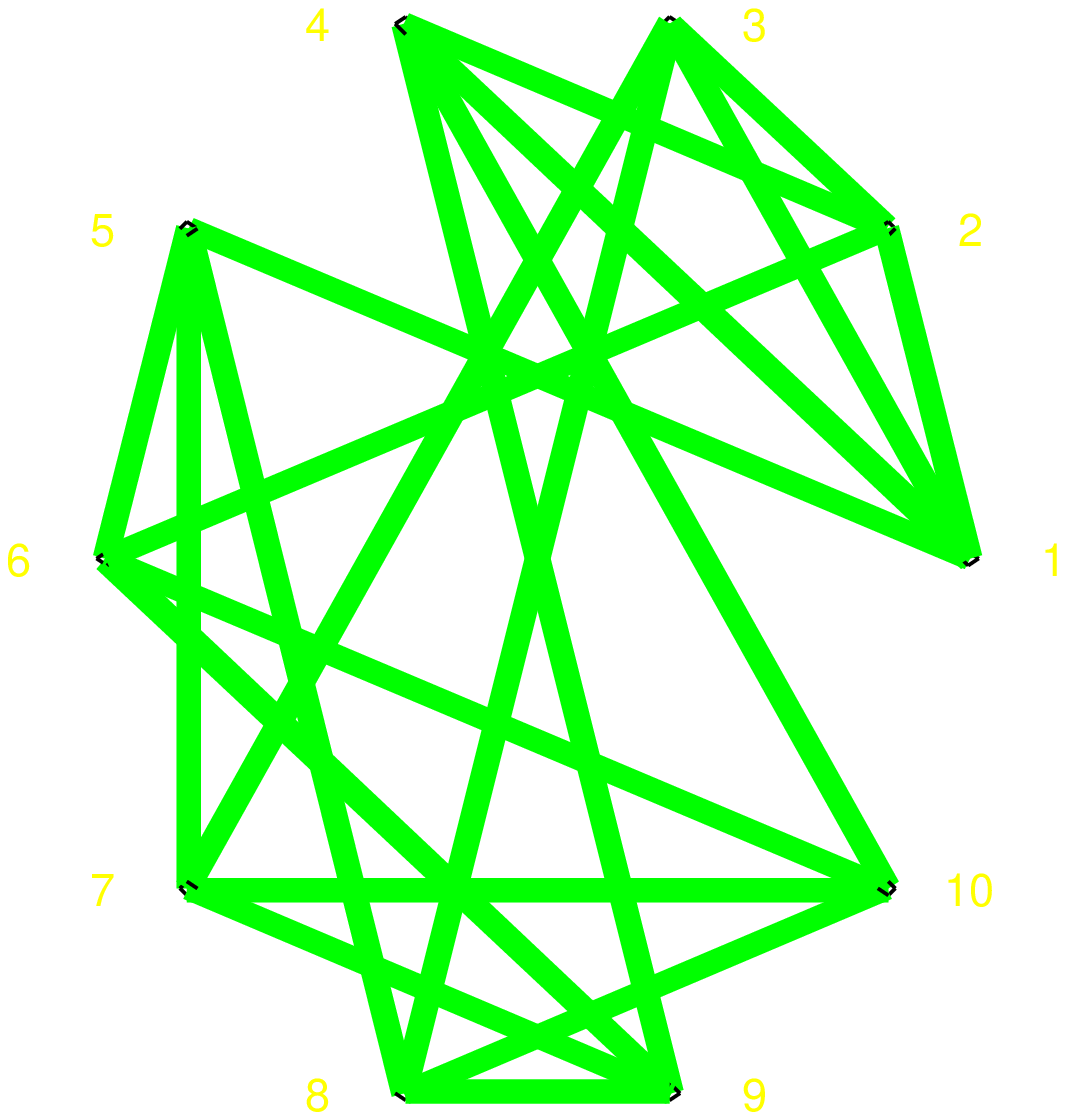}}&414.873~975~722&8&144&$P_{7,8}$&$z_2$&\hfill$[+\frac{17577}{2}Q_3^2Q_7\!+\!10800Q_3Q_5^2$\\[-6mm]
13&&\multicolumn{6}{l}{$-\,\frac{40309047}{1400}Q_{13,1}-\frac{353601}{350}Q_{13,2}+\frac{48051}{28}Q_{13,3}-17920Q_3Q_{10}-19840Q_5Q_8$}\\[1ex]\hline
$P_{8,25}$&\hspace*{-2mm}\raisebox{-9mm}{\includegraphics[width=12mm]{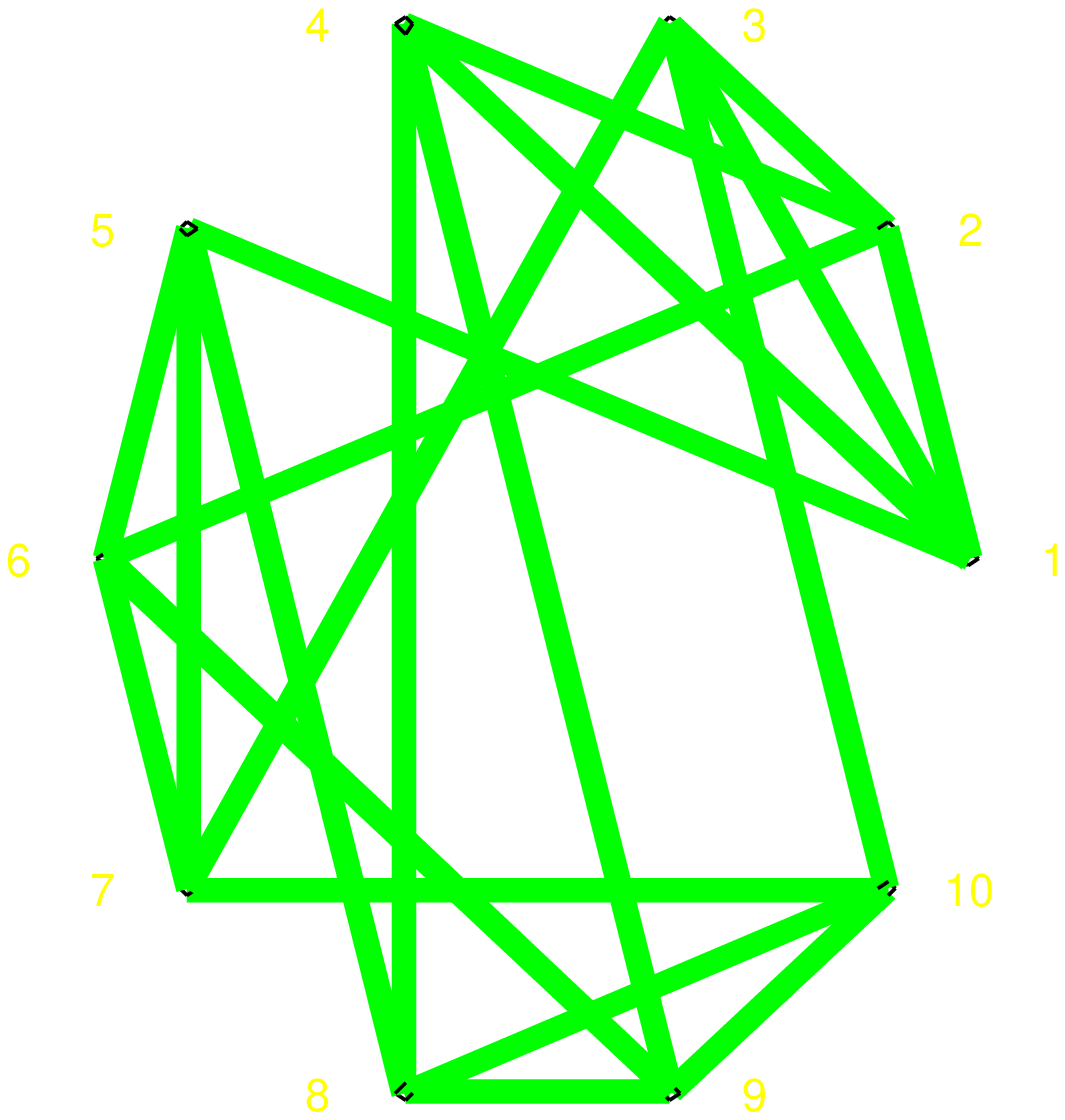}}&641.723~358~297&4&48&$P_3^2$&0&Fourier, twist\\[-6mm]
12&&\multicolumn{6}{l}{$P_{8,18}$}\\[1ex]\hline
$P_{8,26}$&\hspace*{-2mm}\raisebox{-9mm}{\includegraphics[width=12mm]{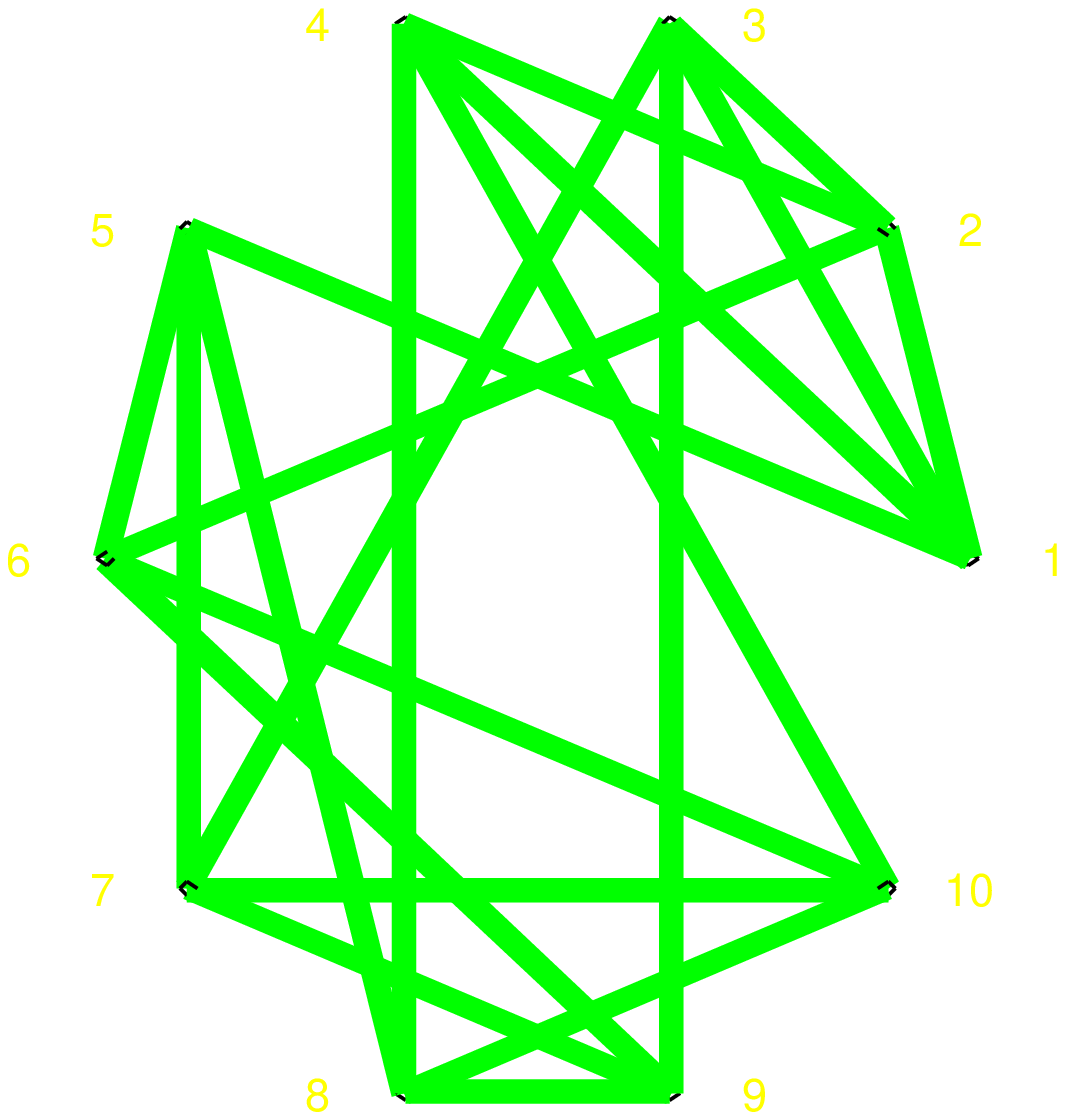}}&500.445~152~216&4&6&$P_{7,9}$&$z_2$&\hfill$[-\frac{1215}{2}Q_3Q_5^2$\\[-6mm]
13&&\multicolumn{6}{l}{$\frac{25114323}{22400}Q_{13,1}-\frac{113979}{1400}Q_{13,2}+\frac{4443}{112}Q_{13,3}-896Q_3Q_{10}+1984Q_5Q_8+\frac{1701}{8}Q_3^2Q_7$}\\[1ex]\hline
$P_{8,27}$&\hspace*{-2mm}\raisebox{-9mm}{\includegraphics[width=12mm]{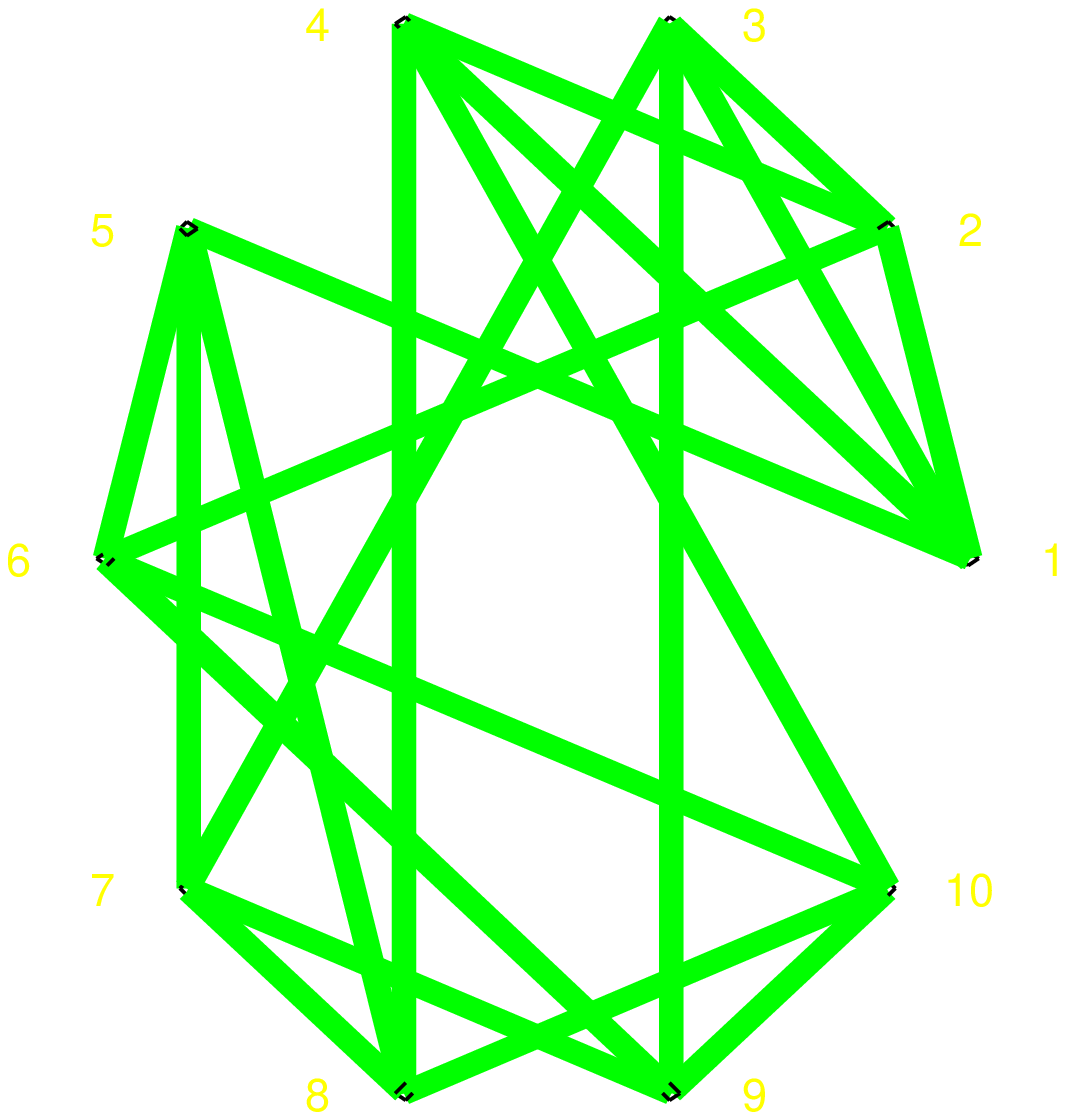}}&598.617~690~750&4&32&$P_{7,10}$&0&Fourier\\[-6mm]
12&&\multicolumn{6}{l}{$P_{8,19}$}
\end{tabular}

\begin{tabular}{llllllll}
name&graph&numerical value&$|$Aut$|$&index&anc.&$-c_2$&remarks, [Lit]\\[1ex]
\multicolumn{2}{l}{weight}&\multicolumn{6}{l}{exact value}\\[1ex]\hline\hline
$P_{8,28}$&\hspace*{-2mm}\raisebox{-9mm}{\includegraphics[width=12mm]{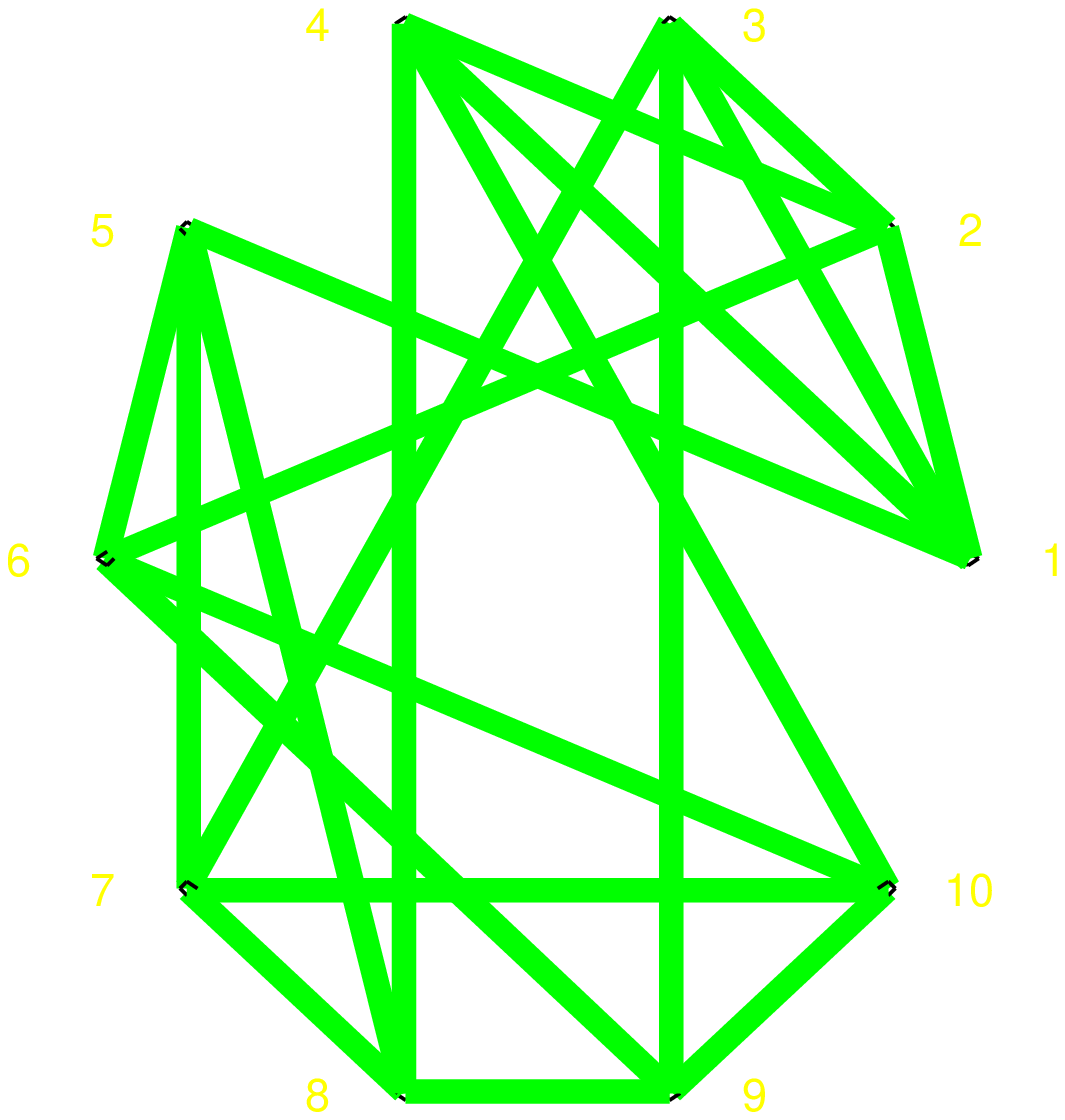}}&500.445~152~216&4&6&$P_{7,9}$&$z_2$&twist\\[-6mm]
13&&\multicolumn{6}{l}{$P_{8,26}$}\\[1ex]\hline
$P_{8,29}$&\hspace*{-2mm}\raisebox{-9mm}{\includegraphics[width=12mm]{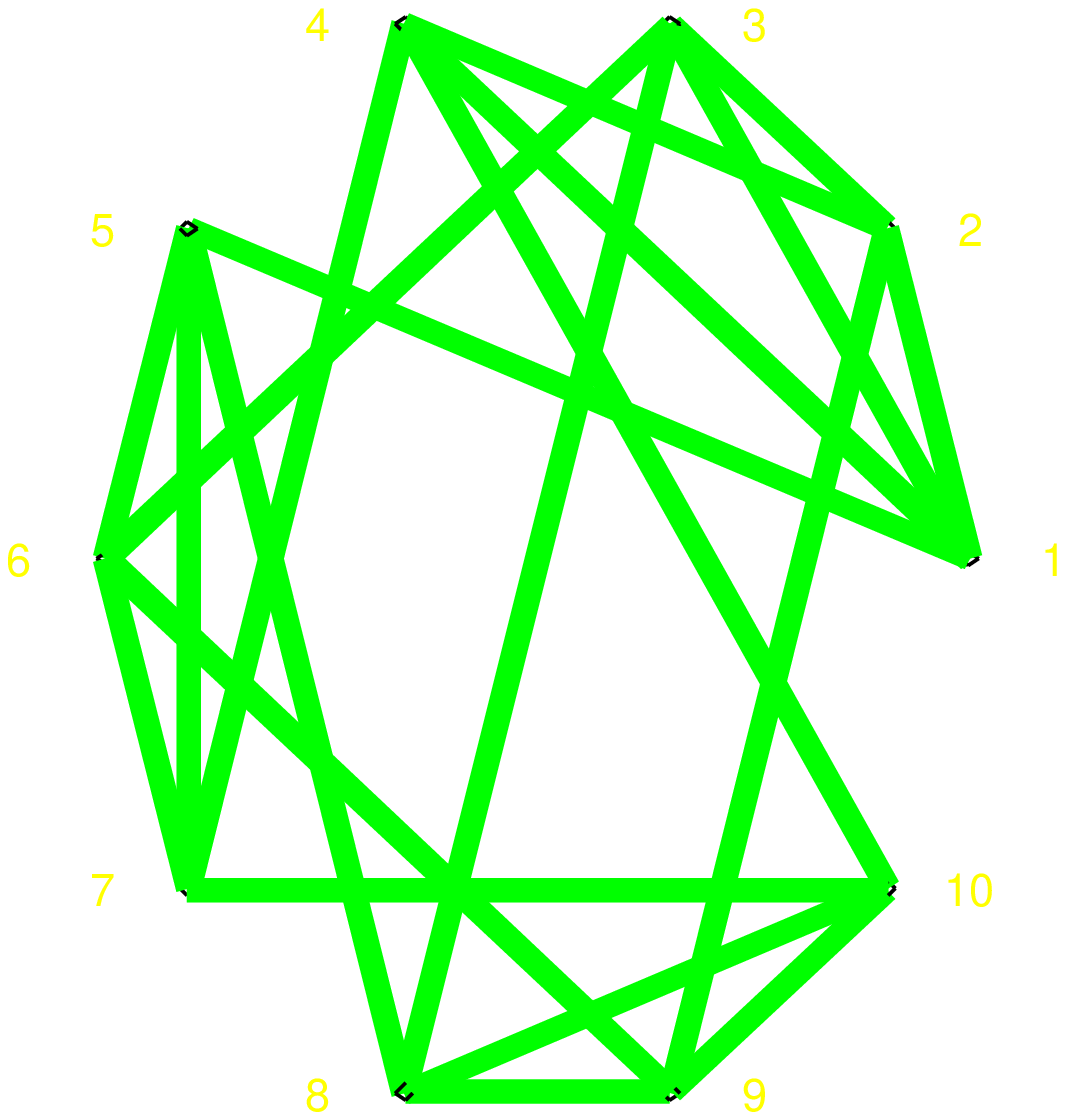}}&553.273~794~612&2&1&$P_{7,9}$&$z_2$&\\[-6mm]
13&&\multicolumn{6}{l}{$\frac{78907643}{89600}Q_{13,1}-\frac{306689}{5600}Q_{13,2}+\frac{16987}{448}Q_{13,3}+\frac{10129}{32}Q_3^2Q_7-\frac{2275}{4}Q_3Q_5^2$}\\[1ex]\hline
$P_{8,30}$&\hspace*{-2mm}\raisebox{-9mm}{\includegraphics[width=12mm]{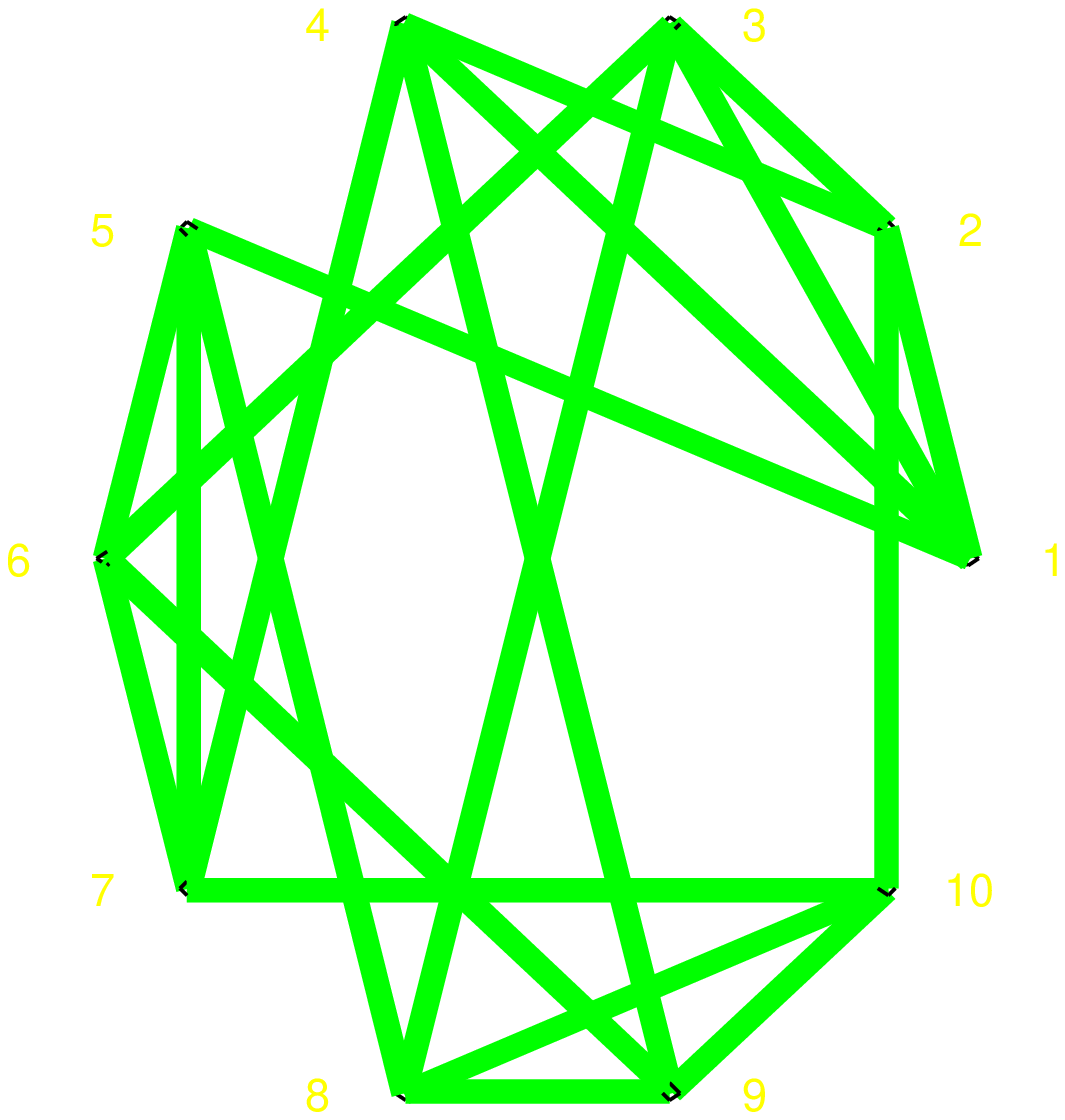}}&$\approx 505.5$&2&?&$P_{7,11}$&$z_3$&\cite{EPHepp}\\[-6mm]
13&&\multicolumn{6}{l}{?}\\[1ex]\hline
$P_{8,31}$&\hspace*{-2mm}\raisebox{-9mm}{\includegraphics[width=12mm]{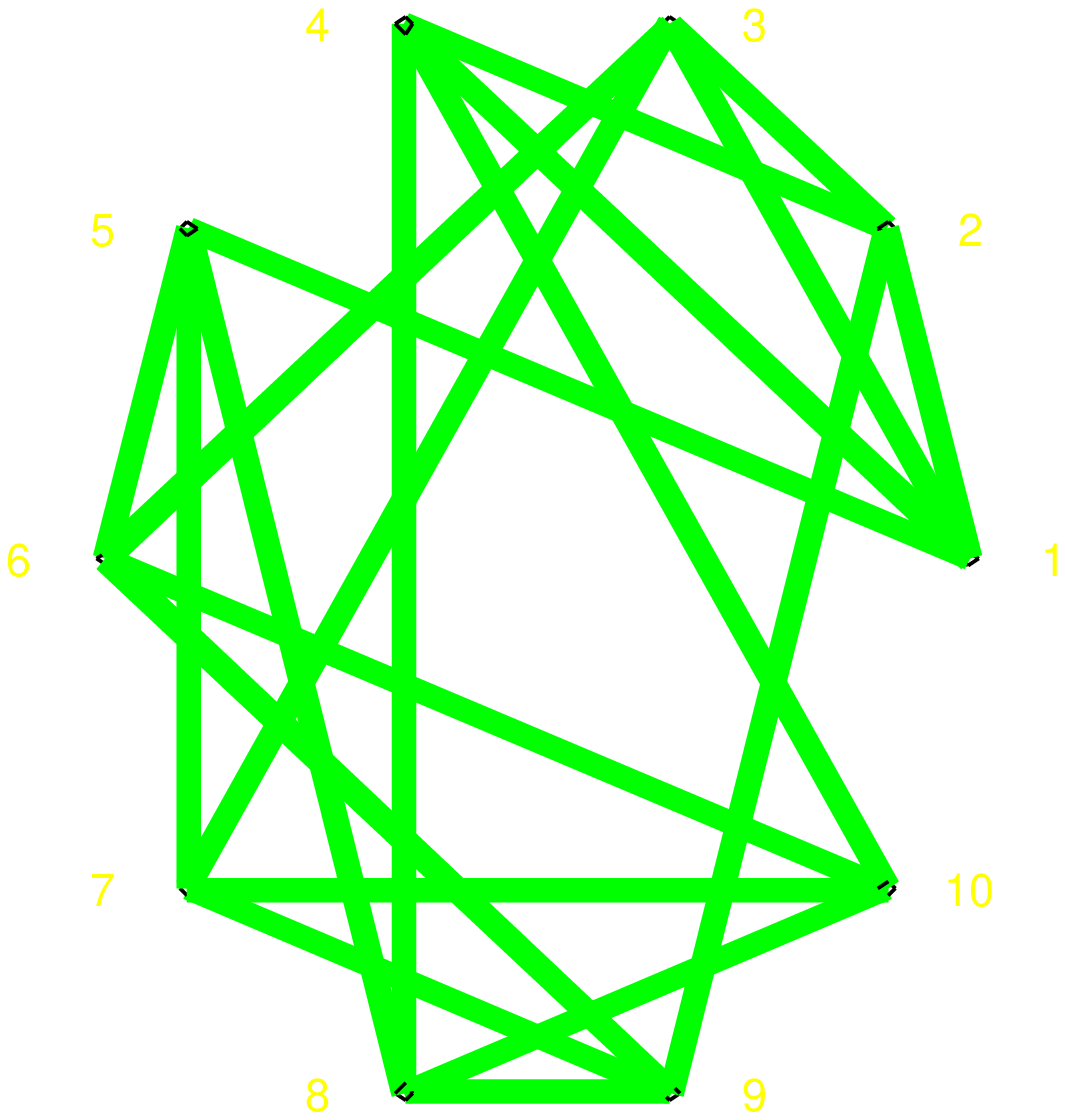}}&460.088~538~246&4&8&$P_{7,8}$&$z_2$&\hfill$[-7330Q_3Q_5^2$\\[-6mm]
13&&\multicolumn{6}{l}{$\frac{67363763}{5600}Q_{13,1}\!-\!\frac{36487}{175}Q_{13,2}\!-\!\frac{1913}{7}Q_{13,3}\!+\!1792Q_3Q_{10}\!+\!7936Q_5Q_8\!+\!98Q_3^2Q_7$}\\[1ex]\hline
$P_{8,32}$&\hspace*{-2mm}\raisebox{-9mm}{\includegraphics[width=12mm]{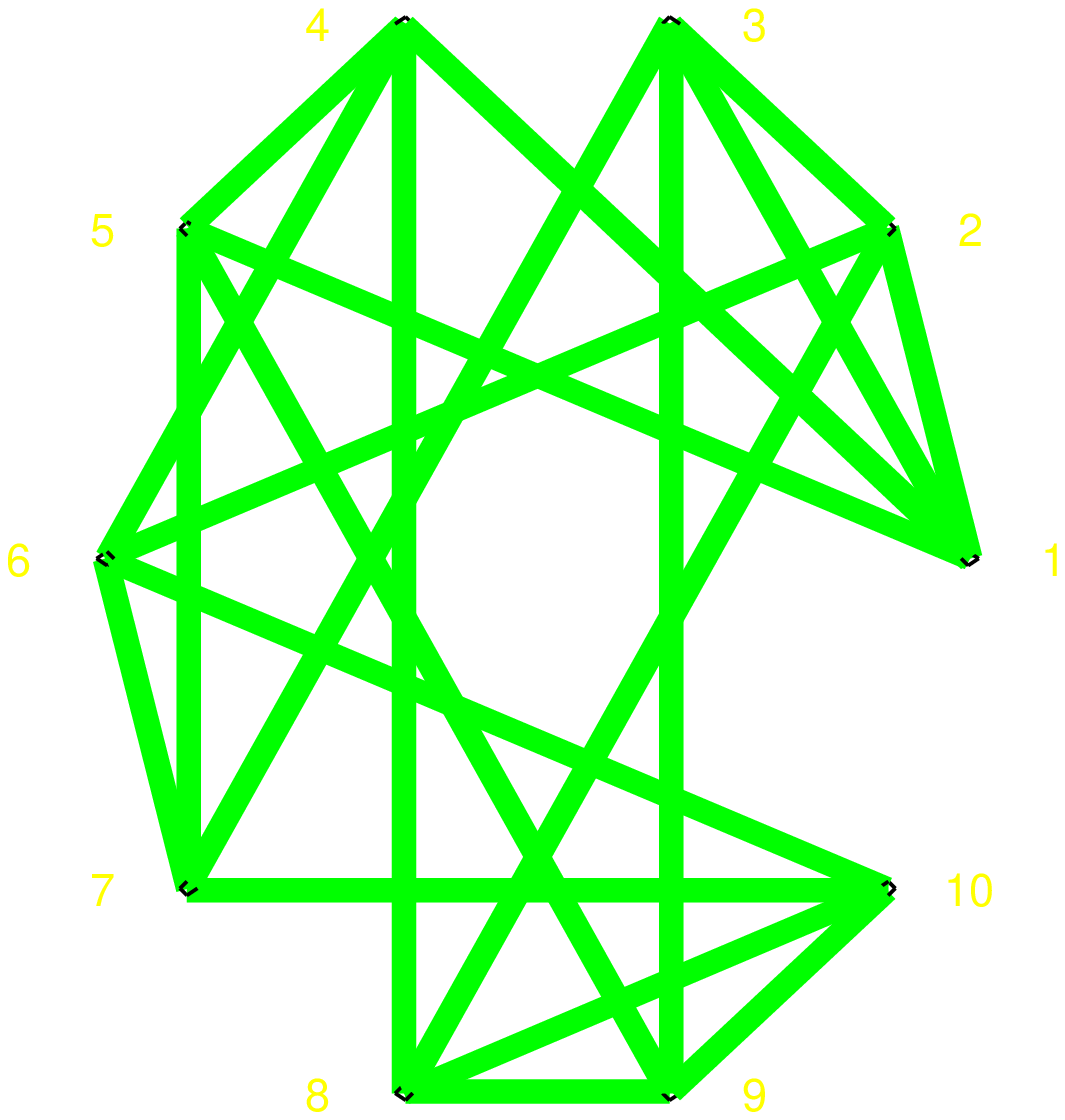}}&470.720~125~534&16&17280&$P_{8,32}$&0&\\[-6mm]
12&&\multicolumn{6}{l}{$-\frac{81920}{23}Q_{12,1}-\frac{655360}{23}Q_{12,2}+\frac{20480}{23}Q_{12,3}+\frac{8760}{23}Q_3Q_9+\frac{15660}{23}Q_5Q_7$}\\[1ex]\hline
$P_{8,33}$&\hspace*{-2mm}\raisebox{-9mm}{\includegraphics[width=12mm]{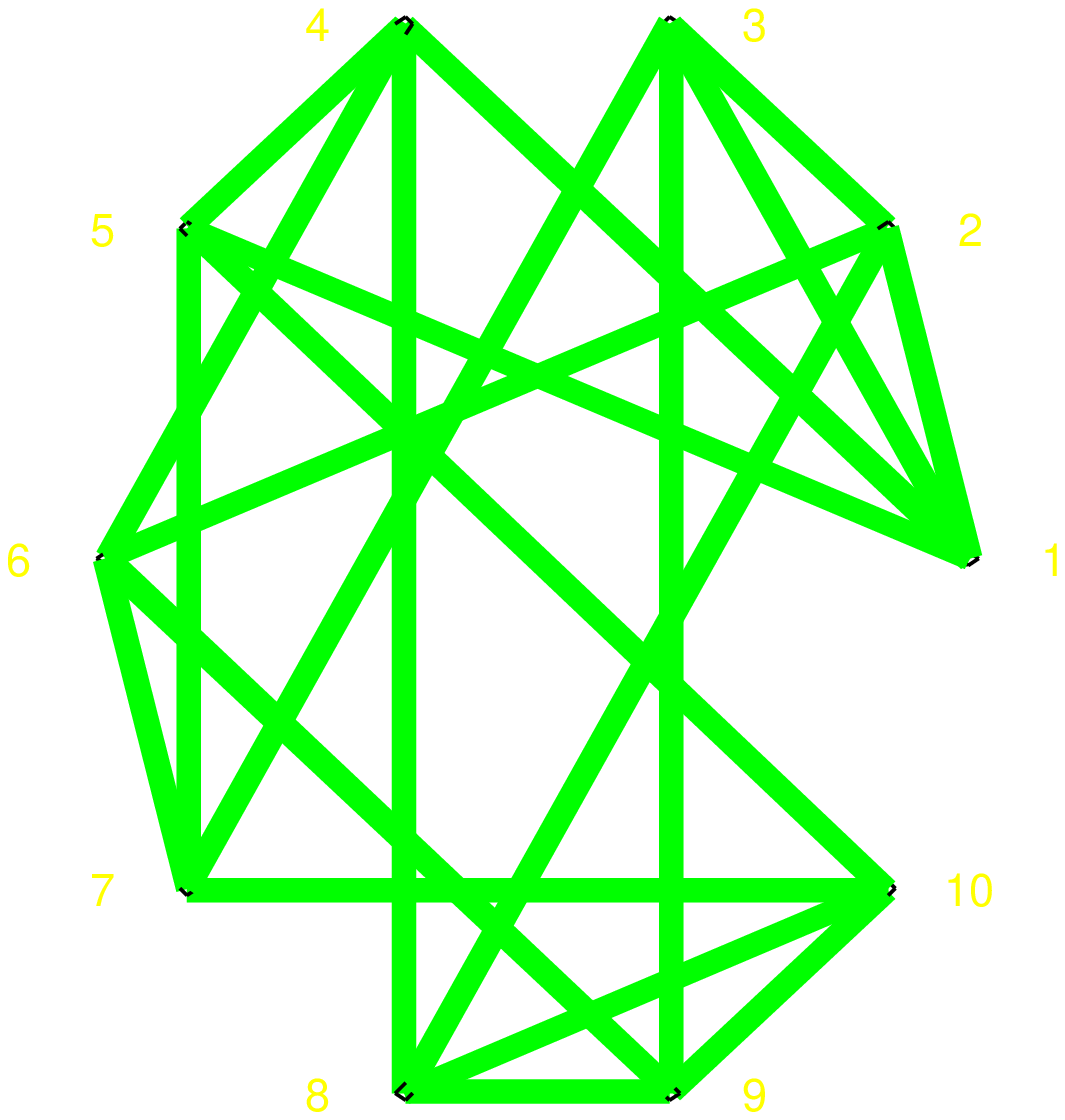}}&468.038~498~992&2&?&$P_{8,33}$&$z_3$&\\[-6mm]
13&&\multicolumn{6}{l}{$Q_{13,4}$}\\[1ex]\hline
$P_{8,34}$&\hspace*{-2mm}\raisebox{-9mm}{\includegraphics[width=12mm]{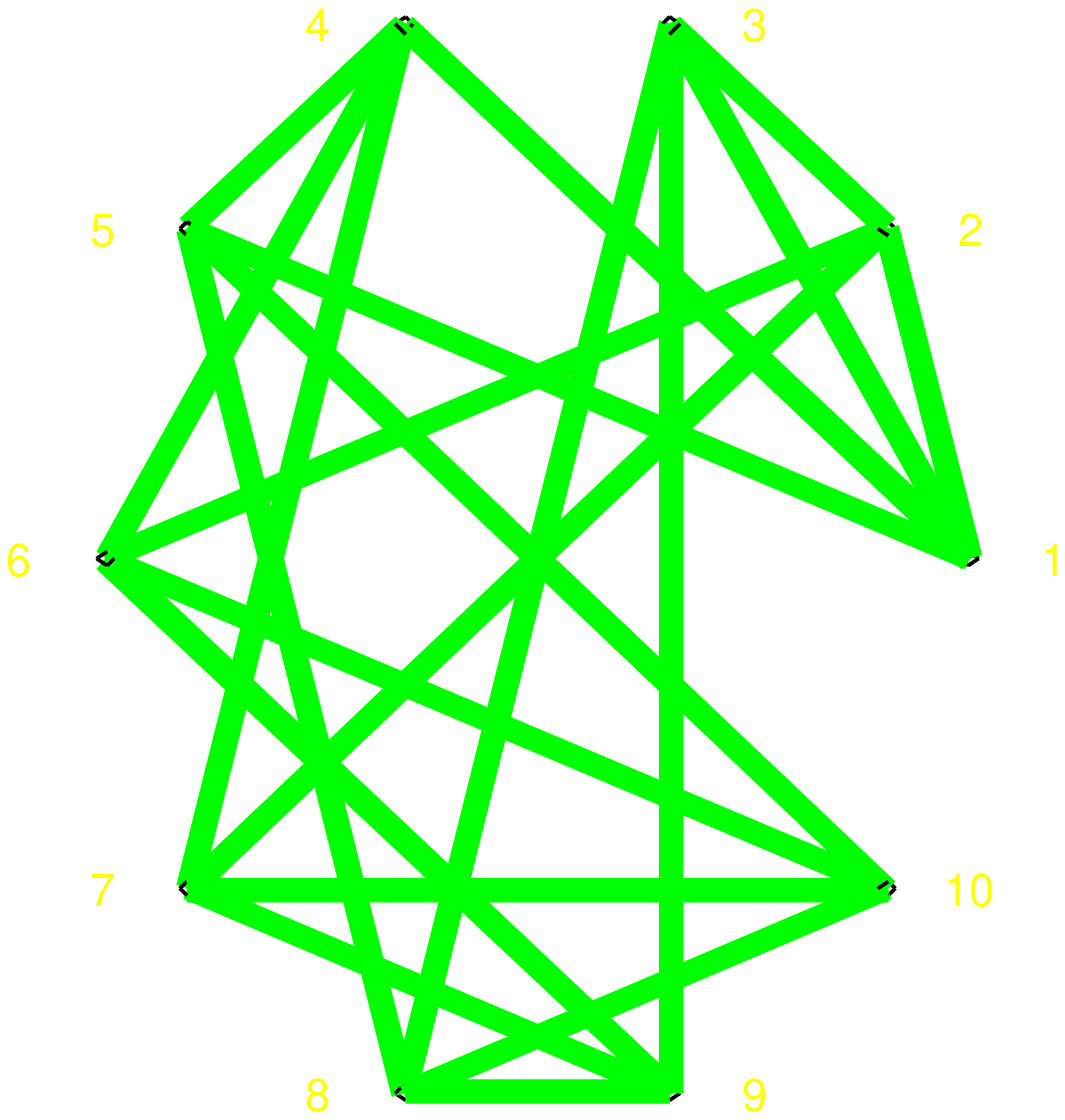}}&470.720~125~534&16&17280&$P_{8,34}$&0&twist\\[-6mm]
12&&\multicolumn{6}{l}{$P_{8,32}$}\\[1ex]\hline
$P_{8,35}$&\hspace*{-2mm}\raisebox{-9mm}{\includegraphics[width=12mm]{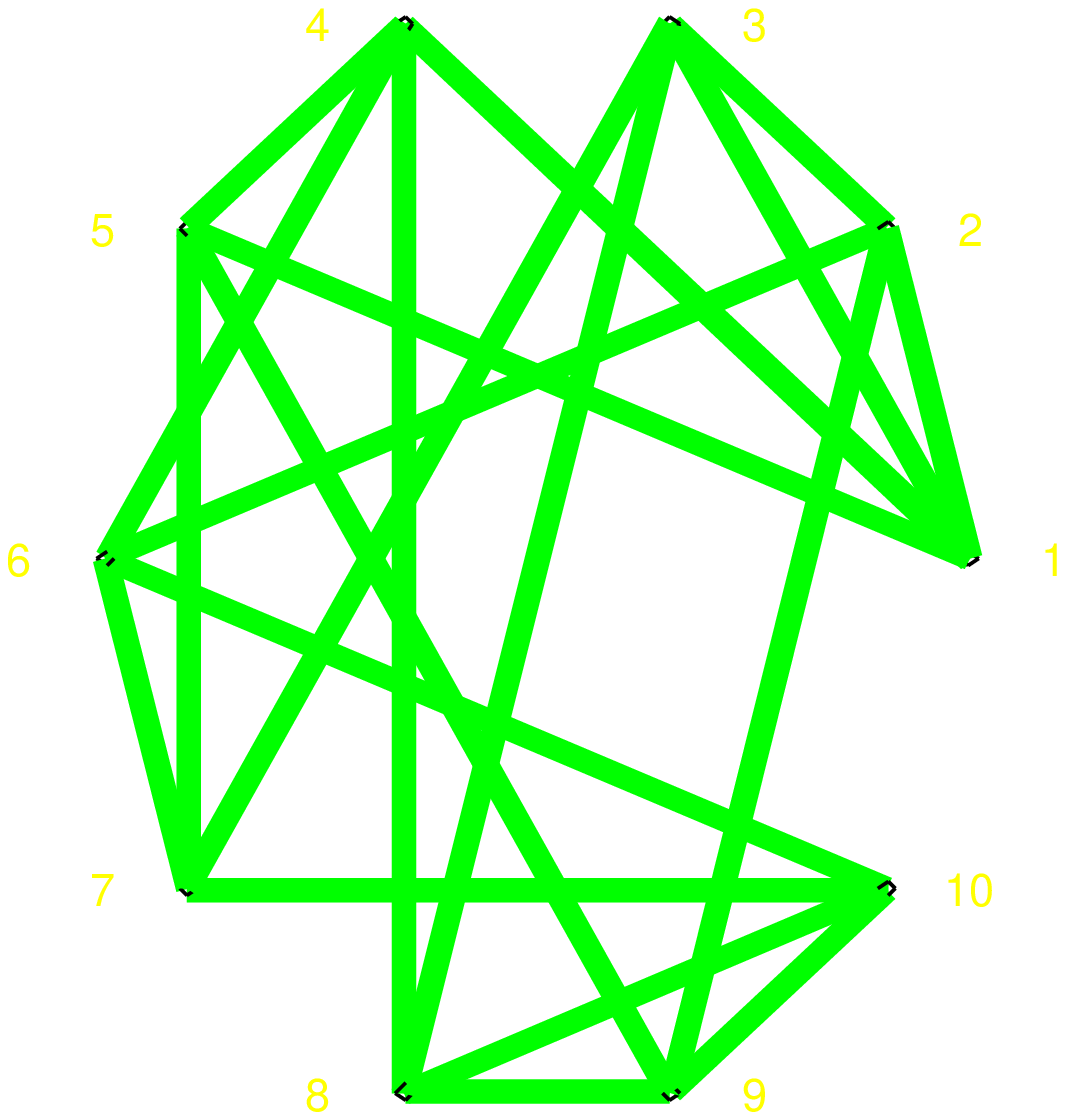}}&$\approx 460.2$&16&?&$P_{8,35}$&$z_2$&Hepp, \cite{EPHepp}\\[-6mm]
13&&\multicolumn{6}{l}{$P_{8,31}?$}\\[1ex]\hline
$P_{8,36}$&\hspace*{-2mm}\raisebox{-9mm}{\includegraphics[width=12mm]{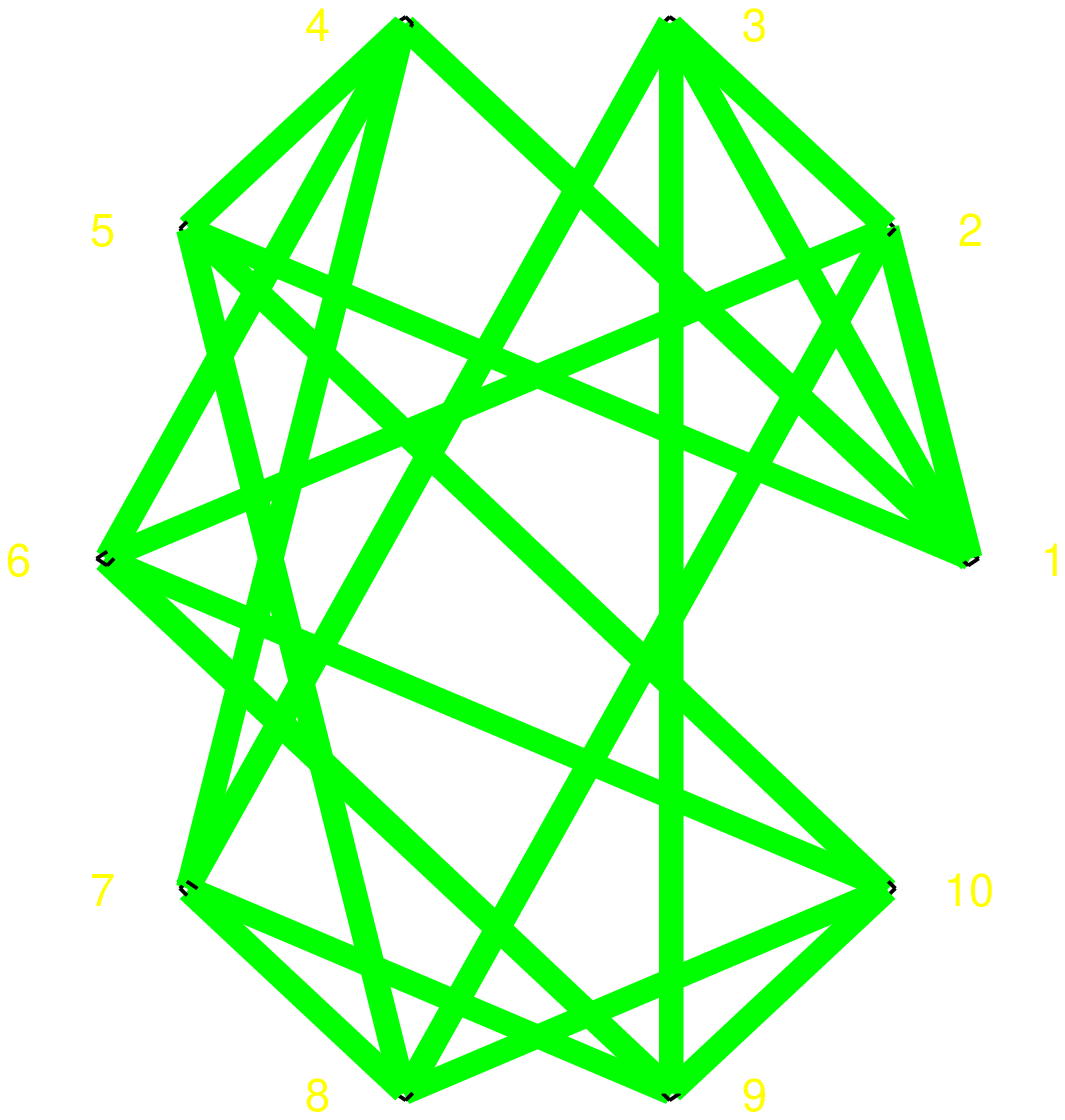}}&$\approx 505.5$&10&?&$P_{8,36}$&$z_3$&Hepp, \cite{EPHepp}\\[-6mm]
13&&\multicolumn{6}{l}{$P_{8,30}?$}\\[1ex]\hline
$P_{8,37}$&\hspace*{-2mm}\raisebox{-9mm}{\includegraphics[width=12mm]{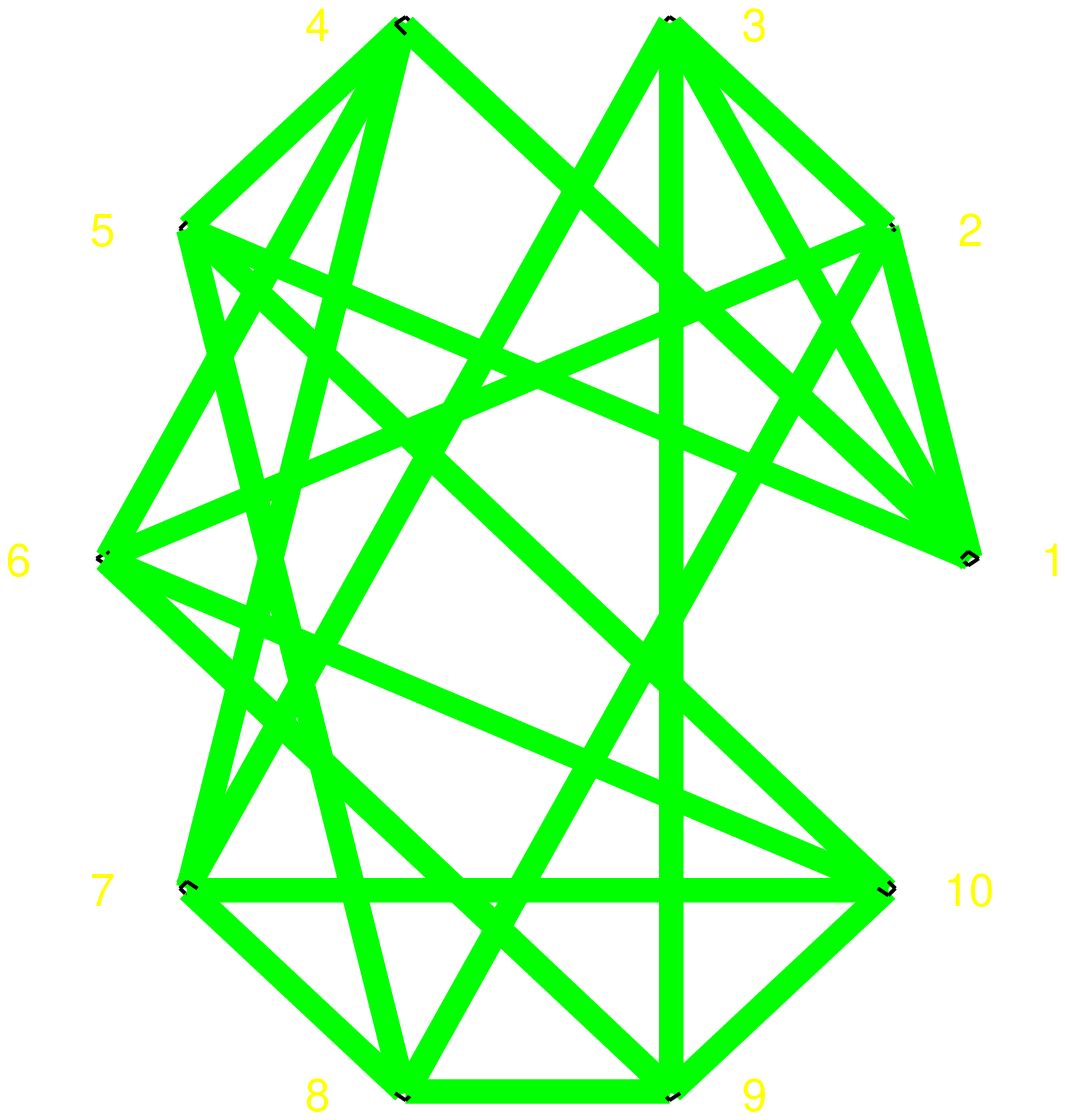}}&$\approx 422.9$&2&?&$P_{8,37}$&$(3,7)$&\cite{EPHepp}\\[-6mm]
?&&\multicolumn{6}{l}{?}\\[1ex]\hline
$P_{8,38}$&\hspace*{-2mm}\raisebox{-9mm}{\includegraphics[width=12mm]{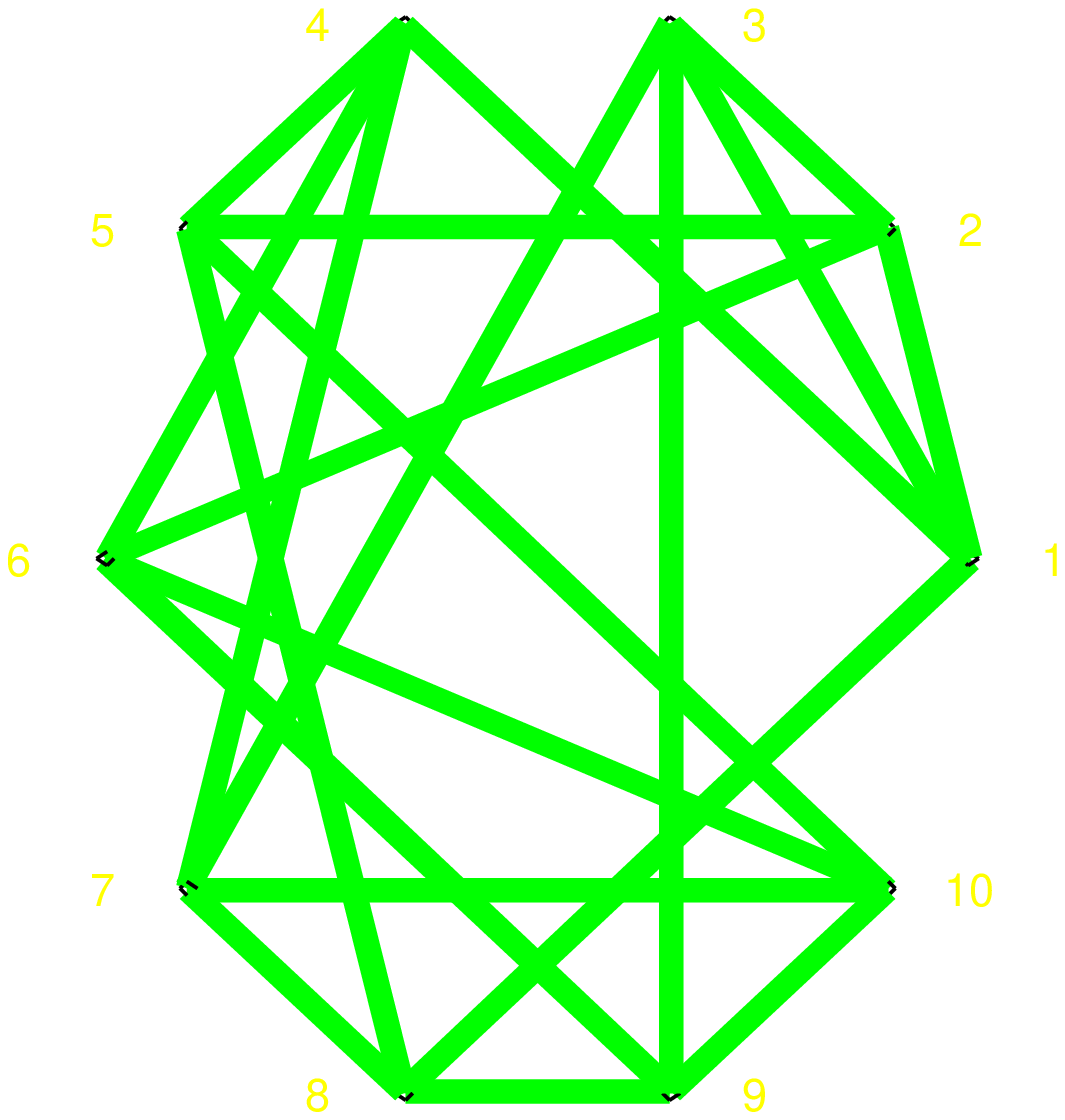}}&$\approx 386.6$&4&?&$P_{8,38}$&$(4,5)$&\cite{EPHepp}\\[-6mm]
?&&\multicolumn{6}{l}{?}\\[1ex]\hline
$P_{8,39}$&\hspace*{-2mm}\raisebox{-9mm}{\includegraphics[width=12mm]{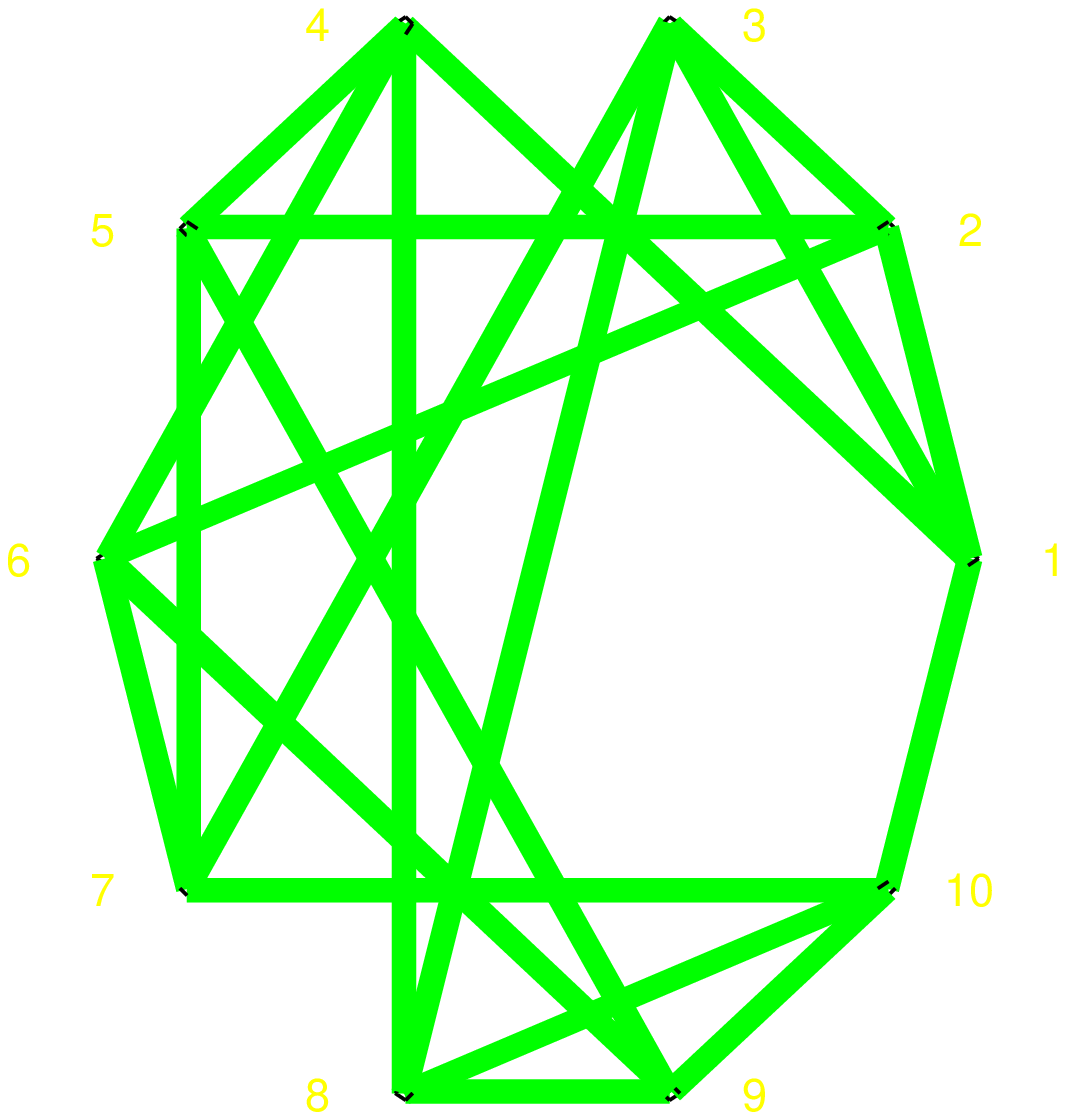}}&$\approx 384.2$&8&?&$P_{8,39}$&$(3,8)$&\cite{EPHepp}\\[-6mm]
?&&\multicolumn{6}{l}{?}\\[1ex]\hline
$P_{8,40}$&\hspace*{-2mm}\raisebox{-9mm}{\includegraphics[width=12mm]{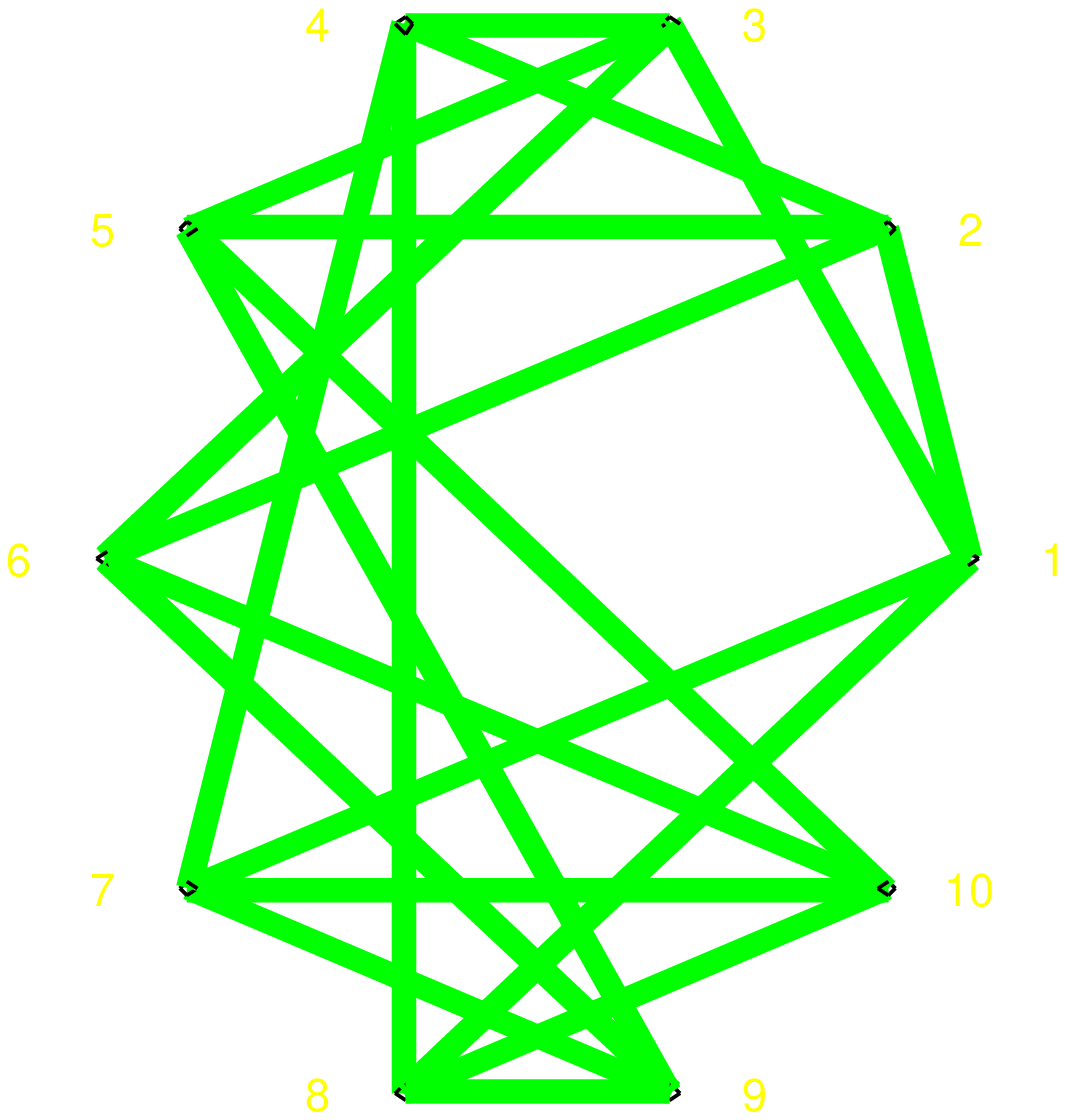}}&$\approx 312.1$&320&?&$P_{8,40}$&$z_4$&$C^{10}_{1,4}$, \cite{EPHepp}\\[-6mm]
13&&\multicolumn{6}{l}{?}\\[1ex]\hline
$P_{8,41}$&\hspace*{-2mm}\raisebox{-9mm}{\includegraphics[width=12mm]{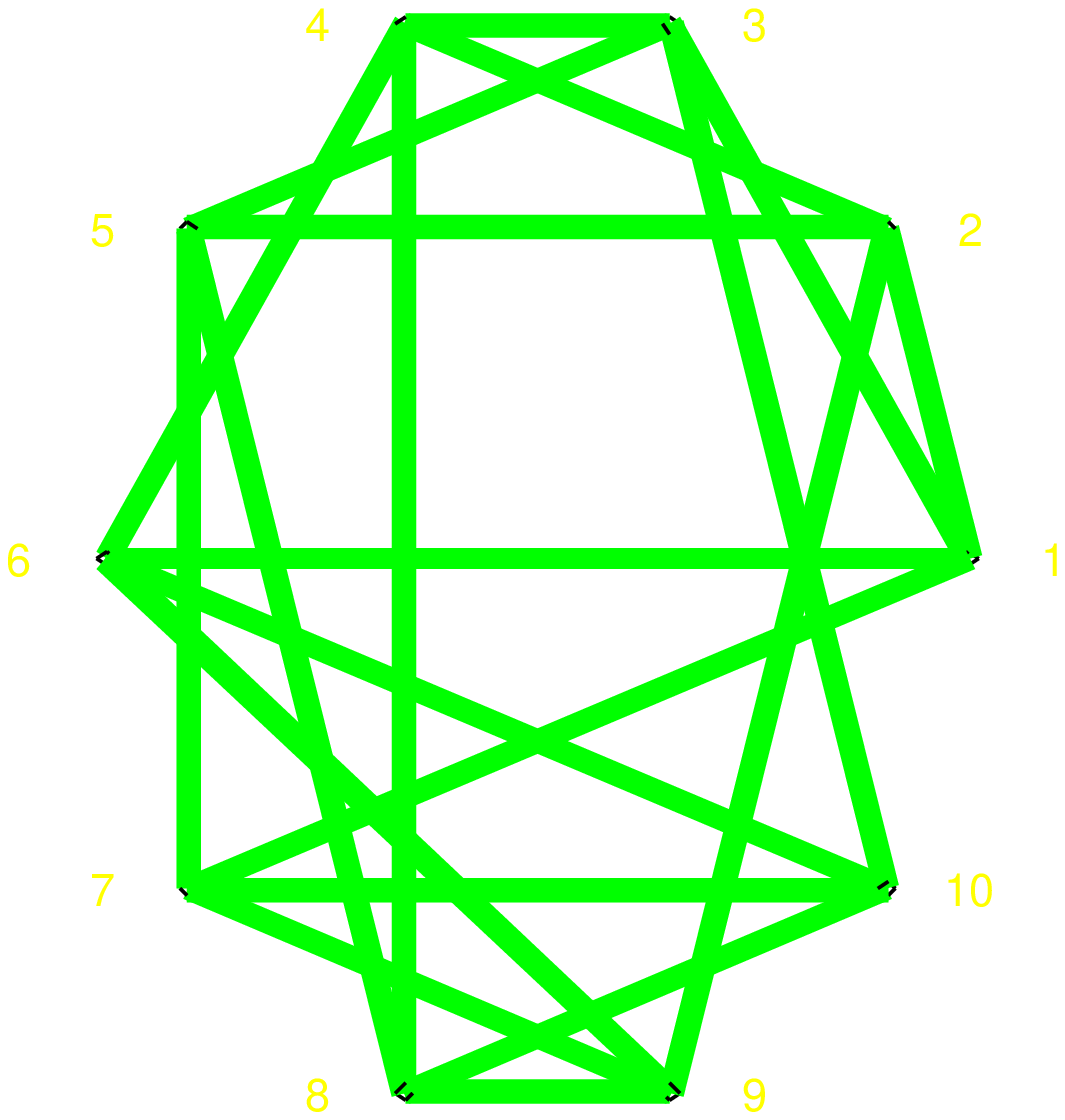}}&$\approx 323.3$&240&?&$P_{8,41}$&$(6,3)$&$C^{10}_{1,3}$, \cite{EPHepp}\\[-6mm]
?&&\multicolumn{6}{l}{?\hspace*{101mm}}
\end{tabular}
\vskip1ex

\noindent
Table 3: The census of $\phi^4$ periods. The numbers $Q_\bullet$ are listed in Table 1. All known periods except for $P_{7,11}$ \cite{Panzer:PhD} can be calculated with \cite{Hyperlogproc}.
See also \cite{Census} for explanations.

\end{document}